\documentclass[12pt]{article} 
\usepackage{amsmath,amsthm,amssymb,bm,mathrsfs,amsfonts}
\usepackage{algorithm,algorithmic,graphicx,subfig}
\usepackage{caption,enumerate,natbib,listings,setspace}
\usepackage[T1]{fontenc}
\usepackage[colorlinks,linkcolor=red,anchorcolor=blue,citecolor=blue]{hyperref}
\usepackage[margin=1in]{geometry}
\usepackage[protrusion=false,expansion=true]{microtype}
\usepackage[title]{appendix}
\usepackage{bbm}
\usepackage{comment}
\theoremstyle{plain}
\let\hat\widehat
\let\tilde\widetilde

\newtheorem{lemma}{{\bf Lemma}}

\newtheorem{corollary}{{\bf Corollary}}

\newtheorem{theorem}{{\bf Theorem}}

\newtheorem{assumption}{{\bf Assumption}}

\newtheorem{remark}{{\bf Remark}}

\usepackage{xr}
\usepackage{color}
\usepackage{tikz}
\usepackage{enumitem}
\usepackage{longtable}
\usetikzlibrary{arrows,calc,positioning}

\tikzstyle{intt}=[draw,text centered,minimum size=6em,text width=5.25cm,text height=0.34cm]
\tikzstyle{intl}=[draw,text centered,minimum size=2em,text width=2.75cm,text height=0.34cm]
\tikzstyle{int}=[draw,minimum size=2.5em,text centered,text width=6.5cm]
\tikzstyle{intg}=[draw,minimum size=2.5em,text centered,text width=6.cm]
\tikzstyle{sum}=[draw,shape=circle,inner sep=2pt,text centered,node distance=3.5cm]
\tikzstyle{summ}=[drawshape=circle,inner sep=4pt,text centered,node distance=3.cm]

\title{\Large{\textbf{Fairness-aware Contextual Dynamic Pricing with Strategic Buyers}} }

\author
{
Pangpang Liu\thanks{Department of Biostatistics, Yale University. Email: pangpang.liu@yale.edu.}\qquad 
Will Wei Sun\thanks{Mitchell E. Daniels, Jr. School of Business, Purdue University. Email: sun244@purdue.edu. Corresponding author.}
}

\date{}
\begin{document} 

\maketitle

\begin{abstract}
\noindent
Contextual pricing strategies are prevalent in online retailing, where the seller adjusts prices based on products' attributes and buyers' characteristics. Although such strategies can enhance seller's profits, they raise concerns about fairness when significant price disparities emerge among specific groups, such as gender or race. These disparities can lead to adverse perceptions of fairness among buyers and may even violate thelaw and regulation. In contrast, price differences can incentivize disadvantaged buyers to strategically manipulate their group identity to obtain a lower price. In this paper, we investigate contextual dynamic pricing with fairness constraints, taking into account buyers' strategic behaviors when their group status is private and unobservable from the seller. We propose a dynamic pricing policy that simultaneously achieves price fairness and discourages strategic behaviors. Our policy achieves an upper bound of $O(\sqrt{T}+H(T))$ regret over $T$ time horizons, where the term $H(T)$ captures the effect of buyers’ perceived price difference. When buyers are able to learn the fairness of the price policy, this upper bound reduces to $O(\sqrt{T})$. We also prove an $\Omega(\sqrt{T})$ regret lower bound of any pricing policy under our problem setting. We support our findings with extensive experimental evidence, showcasing our policy's effectiveness. In our real data analysis, we observe the existence of price discrimination against race in the loan application even after accounting for other contextual information. Our proposed pricing policy demonstrates a significant improvement, achieving an average reduction of 30.71\% in regret compared to the benchmark policy.
\end{abstract}

\bigskip
\noindent{\bf Key Words:}  Contextual Bandit, Dynamic Pricing, Fairness, Reinforcement Learning, Regret Bounds

\newpage
\baselineskip=25pt 

\section{Introduction}
Contextual pricing is widely used in finance, insurance, and e-commerce, with companies customizing prices based on contextual information such as income, purchasing history, and the marketing environment. In the online setting, dynamic pricing entails learning unknown demand parameters and sequentially making pricing decisions. Specifically, at each time step $t$, a buyer enters the market and the seller observes the contextual information, i.e.,  products’ attributes and buyers’ characteristics. The seller decides the price based on these contextual information, and collects the purchasing feedback. In the dynamic pricing problem, the seller needs to update the pricing policy sequentially to maximize the total revenue.  

 However, when pricing discriminates against sensitive features such as race or gender, it may violate regulations or diminish perceived fairness, leading to heightened dissatisfaction and perceived betrayal among customers \citep{Wu2022}. In the United States, the Equal Credit Opportunity Act\footnote{\href{https://www.justice.gov/crt/equal-credit-opportunity-act-3}{https://www.justice.gov/crt/equal-credit-opportunity-act-3}} prohibits creditors from discriminating against credit applicants on the basis of race, color, gender, marital status, etc. The European Court of Justice rules that differences in insurance pricing based purely on a person's gender are discriminatory and are not compatible with the EU's Charter of Fundamental Rights\footnote{\href{https://ec.europa.eu/commission/presscorner/detail/en/MEMO_12_1012}{https://ec.europa.eu/commission/presscorner/detail/en/MEMO\_12\_1012}}. Moreover, perceived price unfairness can result in legal penalty\footnote{\href{https://www.consumerfinance.gov/about-us/newsroom/cfpb-and-doj-order-ally-to-pay-80-million-to-consumers-harmed-by-discriminatory-auto-loan-pricing/}{https://www.consumerfinance.gov/about-us/newsroom/cfpb-and-doj-order-ally-to-pay-80-million-to-consumers-harmed-by-discriminatory-auto-loan-pricing/}}, reputation damage, negative word of mouth, decreases in purchase intentions, or even customer revenge \citep{Domen2016,Isabel2019, Silke2023}. No firm can afford to ignore these negative consequences. Consequently, ensuring pricing fairness in the dynamic pricing policy, particularly concerning these sensitive features, becomes imperative for sellers. Fair pricing policies may seem to yield lower revenue in a short time horizon compared to unfair counterparts. However, they offer a strategic advantage in avoiding these detrimental consequences. Fair pricing policies contribute to the establishment of trust, customer satisfaction, and long-term profitability. 

Practical scenarios often involve price discrimination against specific groups, even after controlling for contextual information \citep{Debbie2008,Zhang2018,Robert2022,Butler2023}.
One such example is the mortgage market. Our motivation comes from the Home Mortgage Disclosure Act (HMDA) data\footnote{\href{https://ffiec.cfpb.gov/data-publication/2022}{https://ffiec.cfpb.gov/data-publication/2022}}, where \cite{Popick2022} found that minority applicants pay higher interest rates than applicants from the majority race, even after controlling for credit risk, underscoring the unfairness of such practices. We refer to Section \ref{sec6} for more discussion of this HMDA dataset. Another instance is from the auto repair industry, where discrimination against female customers is prevalent \citep{Meghan2017}.  Nationwide, women are commonly charged more than men for the same auto repair work. In Los Angeles, 20\% of auto shops surveyed quoted higher prices for women. On average, women are charged 8\% more than men for repair jobs across the country\footnote{\href{https://abc7.com/women-overcharged-in-auto-repair-shops-charges-charged-for-repairs-pal/1660671/}{https://abc7.com/women-overcharged-in-auto-repair-shops-charges-charged-for-repairs-pal/1660671/}}. 

Unfairness in the pricing policy not only can lead to losses for the seller, but also can provoke strategic behaviors among buyers. In personalized pricing, buyers should not be able to easily obtain prices intended for a different consumer group, or if they can, the process should be sufficiently costly \citep{Lukacs2016}. Our study delves into buyers' strategic behaviors, where such behavior is defined as buyers pretending to belong to an alternative group, incurring a fixed cost in the process \citep{Li2023}.
As revealed in our analysis of HMDA data in Section \ref{sec6}, minority applicants tend to pay higher interest rates compared to the majority group. In response, applicants from the minority group may strategically manipulate their identity to appear as members of the majority group, aiming for a lower interest rate. Consistent with this concern, \citet{dobre2023mortgage} showed that measured racial disparities in mortgage outcomes are sensitive to whether race is identified by lenders or by applicants: the estimated gaps increase by 11\%-14\% for minority borrowers when lender-identified race is used.  This discrepancy suggests that misreporting in HMDA is nontrivial and that strategic identity manipulation is empirically relevant.
Practical instances of such strategic behaviors exist in reality. For instance, homeowners from the minority race asked friends from the majority race to pretend to be homeowners during property appraisals, leading to a significant increase in property value\footnote{\href{https://www.cnn.com/2021/12/09/business/black-homeowners-appraisal-discrimination-lawsuit/index.html}{https://www.cnn.com/2021/12/09/business/black-homeowners-appraisal-discrimination-lawsuit/index.html}\\ \href{https://www.indystar.com/story/money/2021/05/13/indianapolis-black-homeowner-home-appraisal-discrimination-fair-housing-center-central-indiana/4936571001/}{\indent https://www.indystar.com/story/money/2021/05/13/indianapolis-black-homeowner-home-appraisal-discrimination-fair-housing-center-central-indiana/4936571001/}}. In cases like these, sellers often obtain buyers' race information on the basis of buyers' provided information, visual observation or surname\footnote{\href{https://www.consumerfinance.gov/rules-policy/regulations/1003/b/}{https://www.consumerfinance.gov/rules-policy/regulations/1003/b/}}. Disadvantaged buyers may seek the assistance of more advantaged friends to appear in the process of purchasing. Other strategic behaviors include manipulating device information in the case of Orbitz' s price discrimination against Mac users \citep{Anna2014}, or forging a student ID if there is a discount for students. When significant unfairness arises, buyers may be motivated to manipulate their group membership to gain access to a lower price, even if it incurs additional costs.

\subsection{Our Contribution}

To address the aforementioned contextual dynamic pricing problem with fair-minded and strategic buyers, we propose a fairness-aware pricing policy designed to deter buyers' strategic actions by fostering a favorable fairness perception. We focus on fairness-aware dynamic pricing when buyers exhibit boundedly rational behaviors, and do not aim to characterize a full Nash equilibrium of the dynamic pricing game under uncertainty.

We first formulate a new dynamic pricing problem, where the true group status (sensitive feature) of buyers is private information and is not observable by the seller. In practice, buyers may engage in strategic behaviors by presenting a self-reported group membership to the seller. Such revealed group status might be different from the true group status. Buyers decide if it is worthwhile to manipulate the group status by learning the price disparity between the two groups based on publicly released data. In reality, certain data releases are mandated by law, as illustrated by the Home Mortgage Disclosure Act, which requires many financial institutions to maintain, report, and publicly disclose loan-level information about mortgages\footnote{\href{https://www.consumerfinance.gov/data-research/hmda/}{https://www.consumerfinance.gov/data-research/hmda/}}. The buyers' learning process is a pivotal aspect of our framework.
The price difference between the two groups, as learned by buyers, is termed ``fairness perception”. A positive fairness perception towards the seller can enhance the seller's reputation and restrain buyers' strategic behavior. Our newly formulated problem encompasses three critical components: price fairness, buyers' learning process and strategic behaviors. To illustrate these three components in dynamic pricing, we depict the workflow of our problem in Figure \ref{pro1}. At time $t$, a buyer with feature $\boldsymbol{x}_t$ and a private true group status $G_t\in \{0, 1\}$ enters the market. The buyer learns the prices $\hat{p}_0(\boldsymbol{x}_t)$ and $\hat{p}_1(\boldsymbol{x}_t)$,  intended for buyers with $\boldsymbol{x}_t$ from group 0 and group 1, respectively, from the history data. After evaluating the cost of manipulating group status in comparison to the price disparity between the two groups, the buyer decides to reveal $G_t'\in \{0, 1\}$. Upon receiving the feature $\boldsymbol{x}_t$ and group status $G_t'$, the seller offers a price adhering to fairness constraints, denoted as $p_t = p_{G_t'}(\boldsymbol{x}_t)$. Finally, the seller receives the purchase feedback $y_t$, and discloses the data $(\boldsymbol{x}_t, G_t', p_t)$ to the public. 
\begin{figure}[!t]
\centering
\includegraphics[scale=0.7]{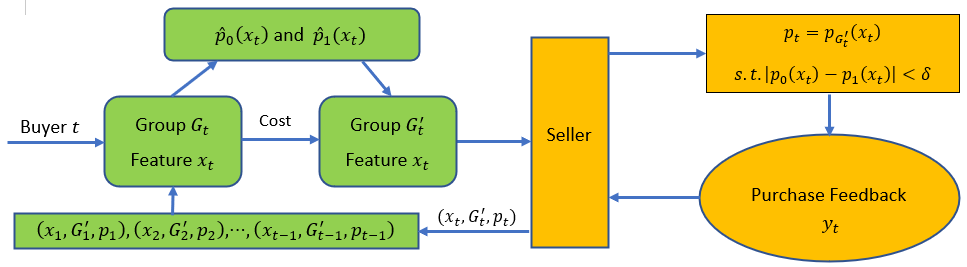}
\caption{Fairness-aware contextual dynamic pricing process with strategic buyers. The seller can only observe the buyer's revealed group status $G_t'$, which may differ from the true group status $G_t$.}
\label{pro1}
\end{figure}    

To solve this problem, we propose a dynamic pricing policy aimed at achieving price fairness between two groups while deterring buyers' manipulation of the group status. The problem faced by the seller is known as the exploration versus exploitation trade-off. On one hand,  the pricing policy influences the seller's ability to learn about demand (exploration), a knowledge that can be leveraged to increase future profits. On the other hand,  the pricing policy impacts immediate revenues (exploitation). 
To balance the trade-off between exploration and exploitation, our policy utilizes a bandit framework and operates in two distinct phases: the exploration phase and the exploitation phase. During the exploration phase, the seller extends the same price to both groups of buyers, and collects true group status data and estimates buyers' preference parameters. The rationale behind revealing the true group status lies in the fact that both groups of buyers receive identical prices, rendering it unprofitable for them to manipulate their group status during this phase. In the subsequent exploitation phase, the seller offers fairness-constrained prices, denoted as $p_0(\boldsymbol{x}_t)$ and $p_1(\boldsymbol{x}_t)$ for group 0 and group 1, respectively. Given that the true group status is unobservable by the seller, revenue loss is incurred when buyers misreport their group status. For instance, buyers from group 0 might misrepresent themselves as belonging to group 1, prompting the seller to offer the price $p_1(\boldsymbol{x}_t)$ based on the observed group status. In this paper, we scrutinize the buyers' fairness learning process, an essential element that deters group status manipulation when the price disparity learned by buyers falls below the manipulation cost, leading to the discouragement of buyers' strategic manipulation.

In a strategic environment, the seller faces the challenge of lacking direct access to the true buyer group status. This absence of direct observation makes it difficult for the seller to offer the optimal fair price to the strategic group. To address this challenge, we formulate a fair pricing policy aimed at discouraging buyers' strategic behavior. To ensure the effectiveness of this discouragement, another challenge is understanding how buyers perceive the fairness of the pricing policy. To tackle this difficulty, we establish that buyers' perceived fairness level is closely aligned with the fairness level set by the seller. Moreover, the performance of the pricing policy is evaluated via the cumulative regret, which is the cumulative expected revenue loss against a clairvoyant policy that possesses complete knowledge of both the demand model parameters and the true group status of buyers in advance, and always offers the revenue-maximizing price while adhering to fairness constraints. We theoretically demonstrate that our strategic dynamic pricing policy achieves a regret upper bound of $O(\sqrt{T}+H(T))$ regret over a time horizon of $T$, where $H(T)$ arises from buyers' assessment accuracy of the fairness of the pricing policy based on their learned price difference. Notably, when buyers can effectively learn and assess the fairness of the pricing policy, this upper bound reduces to $O(\sqrt{T})$. Traditional regret upper bound proofs typically involve bounding the difference between the proposed policy and the clairvoyant policy. However, in our proof, an additional layer of complexity arises as we need to explore the effectiveness of our pricing policy in discouraging strategic behaviors. Importantly, we establish an $\Omega(\sqrt{T})$ regret lower bound of any pricing policy in our problem setting, which indicates the optimality of our pricing policy.

\subsection{Literature Review}
Recently, fairness and buyers' strategic behaviors are gaining prominence in the dynamic pricing domain, and these facets are closely related to our work. In the following paragraphs, we discuss these related literature. Table \ref{tab1} outlines the distinctions between our work and other dynamic pricing research with fairness/strategic buyers. The symbol $\tilde{O}$ denotes the order that hides the logarithmic term. 
\begin{table}[h!]     
\setlength{\tabcolsep}{2pt}
\renewcommand\arraystretch{0.8}  
    \centering
    \caption{Comparison with other dynamic pricing works}
        \begin{tabular}{c|c|c|c|c|c}
    \hline
        Papers & Context &  Fairness &Strategic behavior  & Buyers’ learning & Regret \\ \hline
        \cite{chen2023} &   &  \checkmark& && $\tilde{O}(T^{4/5})$  \\ 
        \cite{Xu2023} &   &  \checkmark& && $\tilde{O}(\sqrt{T})$  \\
        \cite{Cohen2021} &  &  \checkmark& &&$\tilde{O}(\sqrt{T})$  \\ 
        \cite{chen2023utility} &\checkmark &\checkmark &  &&$O(T^{2/3})$\\
        \cite{liu2023contextual} & \checkmark && \checkmark && $O(\sqrt{T})$  \\
        Our work & \checkmark & \checkmark & \checkmark& \checkmark & $O(\sqrt{T}+H(T))^*$ \\ \hline
        \multicolumn{6}{l}{\small *$H(T)$ disappears when buyers effectively learn the fairness of the pricing policy. See Corollary \ref{coro}.}
    \end{tabular}
 \label{tab1}
\end{table}

\textbf{Dynamic Pricing with Fairness}. Dynamic pricing has been an active research area in operations research and machine learning \citep{luo2022,Fan2022}. In the online pricing realm, \cite{Xu2023,Cohen2021,chen2023} have explored the non-contextual pricing problem with fairness constraint. Specifically, \cite{Cohen2021} and \cite{chen2023}  examined pricing problems where the identical deterministic price is offered within each buyer group, while \cite{Xu2023} introduced random prices generated by probability distributions within each buyer group.  Consequently, in these studies, buyers from the same group are subject to the same price or prices from the same probability distribution. In contrast, our work integrates contextual information into the pricing policy, offering prices based on features while simultaneously ensuring fairness among buyers with the same features from different groups.  The work by \cite{chen2023utility} studied the contextual pricing problem with fairness constraint. In all the aforementioned literature, the true group status is observable by the seller while the true group status is unobservable in our paper.  Moreover, we consider inequity-averse buyers who actively seek lower prices by manipulating their group status - an aspect not addressed in \cite{Xu2023,chen2023utility,chen2023,Cohen2021}. Another critical difference is that the aspect of how buyers learn about price fairness is overlooked in \cite{ Xu2023,chen2023utility, chen2023,Cohen2021}. In contrast, our approach delves into the intricacies of how buyers learn about fairness. The existence of buyers' strategic behaviors is well motivated in many pricing applications, rendering existing pricing tools not applicable. To address this challenge, we must devise new tools capable of handling both strategic behaviors and the learning process of fairness.

\textbf{Dynamic Pricing with Strategic Buyers}. 
Existing literature on pricing with strategic buyers has primarily focused on timing \citep{Chen2018}, untruthful bidding in pricing and auction design \citep{Amin2014, Mohri2015}, and feature manipulation \citep{liu2023contextual}. Other literature also exists on feature manipulation within classification \citep{Hardt2016, Dong2018, Chen20201, ghalme2021strategic, Bechavod2021}. Our work specifically addresses strategic behaviors related to feature manipulation, closely related to the study by \cite{liu2023contextual}.  However, the policy presented by \cite{liu2023contextual} incurs a social loss as it cannot effectively curb the futile strategic manipulation. In contrast, our paper employs fairness as a tool to discourage strategic behaviors. Furthermore, \cite{liu2023contextual} could not enforce fairness in pricing policies and hence is not applicable to address price discrimination. 

\subsection{Organization}
The remainder of the paper is organized as follows. In Section \ref{sec2}, we introduce the new problem and the necessary components: fairness and strategic behaviors. In Section \ref{sec3}, we propose the fairness-aware pricing policy with strategic buyers. In Section \ref{sec4}, we provide the theoretical analysis. In Section \ref{sec5}, we conduct simulation studies to examine our proposed policy. A real data analysis is provided in Section \ref{sec6}, followed by some discussion of future work in Section \ref{sec7}. The supplementary materials include additional information and proofs.

\section{Problem Formulation}\label{sec2}
We first introduce the setting of the contextual dynamic pricing.  At each time $t$, a buyer with feature $\boldsymbol{x}_t\in\mathbb{R}^d$ and a group status $G_t\in\{0, 1\}$ enters the market. The true group status of each buyer is considered as a private type, which is unobservable by the seller. The buyer can manipulate it to the other group by incurring some cost. We denote $C_0> 0$ as the fixed unit cost of manipulation, averaged over the total number of products purchased. Upon receiving $\boldsymbol{x}_t$ and a reported group status $G_t'$, which may be different from $G_t$, the seller offers a price $p_t$. At time $t$, the demand of a buyer with feature $\boldsymbol{x}_t$ and group status $G_t=j\in\{0,1\}$ is 
\begin{equation}\label{demand}
 y_{jt}=\alpha_jp_{t}+\boldsymbol{\beta}_j^{\top} \tilde{\boldsymbol{x}}_{t}+\epsilon_t,   
\end{equation}
where $\alpha_j<0$ and $\boldsymbol{\beta}_j$ are unknown parameters, and $\tilde{\boldsymbol{x}}_{t}=(1, \boldsymbol{x}_t^\top)^\top$. For convenience, we denote the unknown parameters by $\boldsymbol{\theta}_j=(\alpha_j, \boldsymbol{\beta}_j^\top)^\top$. The demand depends on the parameter $\boldsymbol{\theta}_j$ corresponding to the true group status rather than those associated with the reported group status. Here, $y_{jt}$ is the demand quantity, such as the amount of the loan the borrower wants to apply for. This linear demand model has been widely considered in the pricing literature \citep{simchi2023pricing,Ningyuan2023}. We use $\boldsymbol{x}_t$ to denote the feature vector (not including group status $G_t$) and is unchangeable by buyers. Without loss of generality, $\mathbb{E}\boldsymbol{x}_t$ is normalized to $\boldsymbol{0}$ \citep{cai2023doubly}. The noise $\epsilon_t$ is an independent and identical distributed $(i.i.d.)$ $\sigma^2_{\epsilon}$-sub-Gaussian variable and $\mathbb{E}(\epsilon_t|\boldsymbol{x}_t,p_{t})=0$.

Now we consider the dynamic pricing problem with price fairness. Let $p_{0t}$ and $p_{1t}$ denote the prices intended for group 0 and group 1, respectively. To incorporate the price fairness, we consider a price constraint, $|p_{0t}-p_{1t}|\leq \delta$, where $\delta\geq 0$ is the parameter for the fairness level that is selected by the seller to meet the internal goal of the company or satisfy regulatory requirements. This price constraint indicates that the price difference of the buyers from two different groups should not exceed $\delta$ after controlling other features. 

We evaluate the performance of a pricing policy by the revenue difference compared to the oracle pricing policy conducted by a fairness-aware clairvoyant seller who knows the true demand parameters $\boldsymbol{\theta}_j$ for $j\in\{0,1\}$ and the private type (group status) $G_t$ of each buyer. The seller's expected revenue from the buyer with feature $\boldsymbol{x}_t$ in group $j$, is
$R_j(p, \boldsymbol{x}_t)=p(\alpha_jp+\boldsymbol{\beta}_j^{\top}\tilde{\boldsymbol{x}}_t)$ at price $p$. Denote the proportion of buyers from group 0 as $q\in (0, 1)$. Without loss of generality, we consider group 0 as the discriminated group, i.e., $p_{0}>p_{1}$ for the same feature. At each time $t$, the fairness-aware clairvoyant seller maximizes the weighted revenue by solving the following constrained optimization problem,
\begin{equation}\label{price}
\begin{aligned}
\mathop{\max}_{p_0,p_1}&\ qR_0(p_{0},\boldsymbol{x}_t)+(1-q)R_1(p_{1},\boldsymbol{x}_t)\\
&s.t. \ p_{0}-p_{1}\leq\delta.
\end{aligned}    
\end{equation}
The objective function of the weighted revenue in (\ref{price}) is an extension of \cite{Xu2023} from the non-contextual setting to the contextual setting. It also includes \cite{Cohen2021,chen2023} as a special case with $q=1/2$. In the main paper, we focus on the two-group case and provide an extension to the multi-group setting in Appendix \ref{multi}.

The constrained optimization problem (\ref{price}) serves as a full-information benchmark for evaluating our pricing policy. By solving \eqref{price} as detailed in Appendix \ref{derkkt}, we obtain the optimal prices $p_{0t}^*=p_0^*(\boldsymbol{x}_t)$ and $p_{1t}^*=p_1^*
(\boldsymbol{x}_t)$, i.e., for $j=0,1$,
\begin{equation}\label{eq6}
p_{j}^*(\boldsymbol{x}_t)=\left\{
\begin{aligned}
&-\frac{\boldsymbol{\beta}_j^\top \tilde{\boldsymbol{x}}_t}{2\alpha_j}, & \text{if}\ \  \frac{\boldsymbol{\beta}_1^\top \tilde{\boldsymbol{x}}_t}{2\alpha_1}-\frac{\boldsymbol{\beta}_0^\top \tilde{\boldsymbol{x}}_t}{2\alpha_0}\leq \delta,\\
&\boldsymbol{\gamma}_1^\top \tilde{\boldsymbol{x}}_t-j\cdot \delta+\gamma_2, & \text{if}\  \ \frac{\boldsymbol{\beta}_1^\top \tilde{\boldsymbol{x}}_t}{2\alpha_1}-\frac{\boldsymbol{\beta}_0^\top \tilde{\boldsymbol{x}}_t}{2\alpha_0}> \delta,
\end{aligned}
\right.
\end{equation}
where the pricing parameters are
\begin{equation}\label{gam0}
 \boldsymbol{\gamma}_1=-\frac{q\boldsymbol{\beta}_0+(1-q)\boldsymbol{\beta}_1}{2q\alpha_0+2(1-q)\alpha_1},\ \gamma_2=\frac{(1-q)\alpha_1\delta}{q\alpha_0+(1-q)\alpha_1}.  
\end{equation}
When $\frac{\boldsymbol{\beta}_1^\top \tilde{\boldsymbol{x}}_t}{2\alpha_1}-\frac{\boldsymbol{\beta}_0^\top \tilde{\boldsymbol{x}}_t}{2\alpha_0}< \delta$, the unconstrained optimal solution of \eqref{price} satisfies the fairness constraint and is applied. Otherwise, the fairness constraint becomes tight, and the constrained solution is derived.
To evaluate the pricing policy, we leverage cumulative regret over a time horizon of $T$,
 \begin{equation}\label{regret}
 Regret_T=\sum_{t=1}^T\mathbb{E}\left\{q[R_0(p_{0t}^*, \boldsymbol{x}_t)-R_0(p_{0t},\boldsymbol{x}_t)]+(1-q)[R_1(p_{1t}^*, \boldsymbol{x}_t)-R_1(p_{1t},\boldsymbol{x}_t)]\right\},   
 \end{equation}
 which is the difference between the fairness-aware clairvoyant revenue and the revenue under one pricing policy. The expectation in (\ref{regret}) is taken with respect to randomness in the feature $\boldsymbol{x}_t$, the demand and the pricing policy.

Now, we discuss buyers' strategic behavior. The group status of the buyer is considered as a private type and is unobserved.  
The buyers can strategically deceive the seller about the group status to pursue a lower price. Remind that the cost of manipulating group status is $C_0>0$, which is public information \citep{Li2023}. The cost of misreporting group status reflects the efforts and time a buyer undertakes to successfully mimic the other type. Assuming a known switching cost is standard in the strategic learning literature \citep{shavit2020causal, ghalme2021strategic, bechavod2022information, harris2022strategic}.  While the exact value of $C_0$ may not be directly observable in practice, it can be reasonably approximated. For example, if manipulation requires assistance from a member of the advantaged group, publicly available wage data could serve as a practical proxy for estimating the cost. Given that group 0 is the discriminated group, only the buyers from group 0 are likely to manipulate the group status. Let $\hat{p}_0(\boldsymbol{x})$ and $\hat{p}_1(\boldsymbol{x})$ be the prices for group 0 and group 1 that the buyer has estimated using the history data. For a buyer from group 0 with feature $\boldsymbol{x}$, the total estimated cost is $\hat{p}_0(\boldsymbol{x})$ without group manipulation and $\hat{p}_1(\boldsymbol{x})+C_0$ after group manipulation. Therefore, the buyer from group 0 will strategically report the group status as 
\begin{equation}\label{G}
 G'=\left\{
\begin{array}{rcl}
0,       &      & \text{if}\ \hat{p}_0(\boldsymbol{x})-\hat{p}_1(\boldsymbol{x})\leq C_0,\\
1,     &      & \text{if}\ \hat{p}_0(\boldsymbol{x})-\hat{p}_1(\boldsymbol{x})> C_0.
\end{array} \right.    
\end{equation}
We aim to design a pricing policy to restrain buyers' strategic behavior. Intuitively, the price difference should not exceed the manipulating cost $C_0$. To discourage the strategic behavior, the seller chooses $\delta$ in (\ref{price}) such that $\delta<C_0$. In the next section, we show that the pricing policy would incur a linear regret if the seller ignored the buyer's strategic behavior. 
\begin{remark}
The fairness level $\delta$ in \eqref{price} is fixed and chosen to be less than $C_0$. In practice, buyers can learn a $\hat{\delta}$ that approximates $\delta$, but this estimate may be either higher or lower than $\delta$. If the seller sets $\delta=C_0$, then $\hat{\delta}$ may exceed $C_0$, making misreporting potentially worthwhile.  To robustly deter such strategic behaviors in the presence of estimation error, the seller needs to choose $\delta<C_0$. While one could treat $\delta$ as a time-varying strategy and achieve $\lim_{t\rightarrow\infty}\delta_t= C_0$, which could potentially produce higher profits.  In that case, a different full-information benchmark is used to evaluate the pricing policy. This would significantly change the scope of the problem and increase the complexity in the analysis. Since our focus is to evaluate regret rather than identifying which full-information benchmark yields the highest profit, we treat $\delta$ as a fixed parameter rather than a strategy of the seller. 
\end{remark}
\subsection{Linear Regret for Existing Fair Pricing Policy}\label{linearreg}
Consideration of buyers' fairness learning is critical in the dynamic pricing with strategic buyers. Existing fair pricing policies \citep{chen2023utility,chen2023,Cohen2021} do not consider buyers' strategic behaviors and ignore buyers' fairness learning process. In this case, even the seller provides prices with fairness constraints, the disadvantaged buyers always act strategically by manipulating the group status. In this section, we show in Theorem \ref{lem0} that  when buyers are strategic, a pricing policy without considering buyers' fairness learning process incurs a linear regret lower bound of $\Omega(T)$. We now present some standard assumptions in the dynamic pricing literature. In later sections, we will show that our proposed pricing policy achieve a sub-linear regret under the same assumptions.

\begin{assumption}\label{ass0}
For each group $j=0, 1$, the following conditions hold:
The price $p_{jt}\in (0, B)$, the feature vector $\|\boldsymbol{x}_t\|_{2}\leq x_{max}$, the demand parameters $a_{min}\leq |\alpha_j|\leq a_{max}$, $\|\boldsymbol{\beta}_j\|_1\leq b_{max}$, for some positive constants $B, x_{max}, a_{min}, a_{max}, b_{max}$.
\end{assumption}
Assumption \ref{ass0} indicates that the price, features and demand parameters are all bounded. The bounded assumptions are practical and also commonly used in pricing literature \citep{luo2022,luo2023,zhao2023high,wang2022,Fan2022}. In particular, the condition that $|\alpha_j|$ is lower bounded by a positive constant is adopted in prior works \citep{Nambiar2019,simchi2023pricing,Ningyuan2023}, as it guarantees that price has a non-negligible influence on demand. This is essential for the effectiveness of pricing optimization. We define the parameter sets $\mathcal{A}=\{\alpha\in \mathbb{R}: a_{min}\leq |\alpha|\leq a_{max}, \alpha<0\}$ and $\mathcal{B}=\{\boldsymbol{\beta}\in \mathbb{R}^d: \|\boldsymbol{\beta}\|_1\leq b_{max}\}$. Consequently, we assume $\alpha_j\in\mathcal{A}$ and $ \boldsymbol{\beta}_j\in\mathcal{B}$ for $j=0, 1$. 
\begin{assumption}\label{ass1}
Feature vectors are generated independently from a fixed distribution. The minimum eigenvalue of the second-moment matrix $\Sigma_x=\mathbb{E}(\boldsymbol{x}_t\boldsymbol{x}_t^\top)$ is positive, i.e., $\lambda_{min}(\Sigma_x)>0.$ 
\end{assumption}
Assumption \ref{ass1} is mild and requires that no features are perfectly collinear in order to identify the true demand parameters \citep{ Fan2022,chai2024localized,liu2023contextual,zhao2024contextual}. 

\begin{theorem}\label{lem0}
 Let Assumptions \ref{ass0} and \ref{ass1} hold. Let $C_0$ be the manipulation cost and $q$ be the proportion of the disadvantage group. If buyers behave strategically and fail to discern the price difference imposed by the fair pricing constraint, there exist parameters $C_0, q, \alpha_j$ and $\boldsymbol{\beta}_j$ for $j=0,1$  such that any fair pricing policy that neglects buyers' fairness learning incurs a cumulative regret of at least $\Omega(T)$ over the time horizon $T$.   
\end{theorem}
Theorem \ref{lem0} shows that the fair pricing policy without buyers' fairness learning incurs a linear regret lower bound of $\Omega(T)$, indicating the importance of studying the buyers' learning process of fairness when designing fair pricing policies. Intuitively, when buyers from the disadvantage group do not learn the price difference, they always report a false group status. In this case, there exists a problem instance such that it always incurs $\Omega(1)$ regret whenever a buyer from the disadvantage group enters the market. Then, the cumulative regret at time $T$ is at least $\Omega(T)$.  To address this, in Section \ref{sec3}, we develop a new fair dynamic pricing policy by taking buyers' strategic behaviors and buyers' fairness learning process into consideration.

\section{Fairness-aware Pricing with Strategic Buyers}\label{sec3}
We begin by introducing the notion of fairness perception and delving into the process of how buyers learn the price fairness. Subsequently, we present the details of our proposed pricing policy and show how our policy prevents buyers' strategic behaviors.
\subsection{Buyers' Fairness Perception and Learning Process}
Price fairness perception arises when buyers compares their paid price with the price paid by comparative others \citep{Lan2004,Lee2011}.  Unfairness perceptions can emerge if buyers from one group are charged a higher price than their counterparts from another group. Therefore, understanding consumer perceptions of price unfairness is crucial. Moreover, perceived unfairness can drive strategic behaviors among buyers. Specifically, buyers from the disadvantaged group may manipulate their group type to that of the advantaged group, incurring a switching cost $C_0$.

If a buyer from group 0 with features $\boldsymbol{x}_t$ knows that the price disparity $p_{0}(\boldsymbol{x}_t) - p_{1}(\boldsymbol{x}_t) $ is less than the switching cost $C_0$, the buyer has no incentive to manipulate their group status to group 1. However, if the disparity exceeds $C_0$, the buyer may have an incentive to misreport their group status and claim to belong to group 1. In practice, though, the seller typically does not disclose both \(p_{0}(\boldsymbol{x}_t)\) and \(p_{1}(\boldsymbol{x}_t)\) directly to buyers. A more realistic approach is for buyers to learn $p_{0}(\boldsymbol{x}_t)$ and $p_{1}(\boldsymbol{x}_t)$ from the released public historical data. In our HMDA dataset, the financial institutions are required by the Home Mortgage Disclosure Act to publicly disclose loan-level information about mortgages. Based on the historical data, the buyer learns the prices $\hat{p}_0(\boldsymbol{x})$ and $\hat{p}_1(\boldsymbol{x})$ for group 0 and group 1, and then compare the price difference $\hat{\delta}=\hat{p}_0(\boldsymbol{x})-\hat{p}_1(\boldsymbol{x})$ with $C_0$ to assess the fairness perception. 

We adopt the notion of a  \textit{general offline regression oracle} \citep{simchi2022bypassing} to describe buyers' learning process. Given a general function class $\mathcal{P}$, a general offline regression oracle associated with $\mathcal{P}$, denoted by \texttt{OffReg}$_\mathcal{P}$ is defined as a procedure that generates a prediction $\hat{p}_j: \mathcal{X}\rightarrow \mathbb{R}^+$ for $j=0, 1$. We make the following generic assumption on the statistical learning guarantee of  \texttt{OffReg}$_\mathcal{P}$.
\begin{assumption}\label{ass2}
Given $t$ training samples of the form $(\boldsymbol{x}_i, G_i', p_i)$, with $p_i$ the price set by the seller, the offline regression oracle \texttt{OffReg}$_\mathcal{P}$ returns a prediction $\hat{p}_j: \mathcal{X}\rightarrow \mathbb{R}^+$ for $j=0, 1$. For some $\eta_t\in(0,1)$, with probability at least $1-\eta_t$, for all $\boldsymbol{x}$ from $\mathcal{X}$, we have
$$|\hat{p}_j(\boldsymbol{x})-p_{j}(\boldsymbol{x})|\leq \mathcal{E}_{\mathcal{P},\eta_t}(t),$$
where the offline learning guarantee $\mathcal{E}_{\mathcal{P},\eta_t}(t)$ is a function that decreases to $0$ as $t\rightarrow \infty$. 
\end{assumption}
The offline learning guarantee $\mathcal{E}_{\mathcal{P},\eta_t}(t)$ bounds the absolute distance between $\hat{p}_j(\boldsymbol{x})$ and $p_{j}(\boldsymbol{x})$. Assumption \ref{ass2} facilitates the flexibility of buyers' learning methods and serves as a link between the buyer's estimation error and the regret analysis.
The predictor $\hat{p}_j$ is assumed to come from a general function class $\mathcal{P}$, which may include neural networks, decision trees, or other supervised learning models. Because we do not assume which specific algorithm is used by buyers, the exact form of the estimation error 
 $\mathcal{E}_{\mathcal{P},\eta_t}$ is left unspecified and may vary depending on the choice of learning algorithm. Theoretically, the realizability ($p_i(\cdot)\in\mathcal{P}$) is required.  This condition ensures that the true pricing functions lie within the hypothesis class used by the buyer for learning. Under this assumption, the error $\mathcal{E}_{\mathcal{P}, \eta_t}(t)$ in Assumption 3 can be guaranteed to decrease to 0 as $t\rightarrow\infty$. Without realizability, the error would include a non-vanishing approximation error term, and $\mathcal{E}_{\mathcal{P}, \eta_t}(t)$ may not decrease to 0, which would undermine the regret analysis.  Assumption \ref{ass2} requires that the estimation error of the learned price should be small for all $\boldsymbol{x}$ from a fixed distribution. This assumption is reasonable because the covariates are sampled i.i.d. from a fixed distribution during both exploration and exploitation, ensuring that the coverage of the feature space expands as more data are collected. The only potential source of covariate shift comes from buyers’ reported group status in the exploitation phase. However, this concern diminishes over time, since the whole training dataset combines observations from both phases. In particular, data from the exploration phase are drawn i.i.d., which ensures diversity and helps mitigate covariate shift.

\subsection{Fair Pricing Policy}\label{sec:fair_policy}
In this section, we introduce a novel fair dynamic pricing policy to deter buyers' strategic behavior. The algorithm comprises the exploration and exploitation phases. The exploration phase gathers information to learn parameters, while the exploitation phase entails implementing optimal fair pricing based on the acquired knowledge. The length of the exploration phase is $T_0$, and the length of the exploitation phase is $T-T_0$. Our fairness-aware pricing policy with strategic buyers is presented in Algorithm \ref{alg1}. 

\begin{algorithm}[!ht]
\setlength{\baselineskip}{18bp}
\caption{Fairness-aware Pricing Policy with Strategic Buyers}
\label{alg1}
\begin{algorithmic}[1]
\STATE  \textbf{Input}: $T, B, \tau, c_\delta$
\STATE $T_0=\lceil \tau\sqrt{T}\rceil$
\STATE \textbf{Exploration Phase (Uniform Pricing Policy)}:
\FOR{$t=1,2,..., T_0$}
\STATE The buyer with feature $\boldsymbol{x}_t$ and group status $G_t\in\{0, 1\}$ comes to the platform.
\STATE The buyer reveals $G_t'=G_t$.
\STATE The seller sets a price $p_{t}\sim$ Unif$(0, B)$. 
\STATE The seller receives a demand $y_{t}$.
\STATE The seller releases $(p_{t},  G_t, \boldsymbol{x}_t)$ the public.
\ENDFOR
\STATE Denote $\tilde{\boldsymbol{x}}_t=(1, \boldsymbol{x}_t^\top)^\top$ and the seller updates the parameter estimate, for $j\in \{0,1\}$, 
\begin{equation}\label{est}
\hat{\alpha}_{j},\hat{\boldsymbol{\beta}}_{j}=\mathop{\arg\min}_{\alpha\in \mathcal{A},\boldsymbol{\beta}\in\mathcal{B}}\sum_{t=1}^{T_0}\mathbb{I}(G_t=j)(y_{t}-\alpha p_{t}-\boldsymbol{\beta}^\top \tilde{\boldsymbol{x}}_t)^2.
\end{equation}
\STATE \textbf{Exploitation Phase (Fairness-aware Optimal Pricing Policy)}:
\FOR{$t=T_0+1,...,T$}
\STATE The buyer with $\boldsymbol{x}_t$ and group status $G_t\in\{0, 1\}$ comes to the platform.
\STATE The buyer learns the price difference and reveals the group status $G_t'$ according to \eqref{G}.
\STATE The seller offers the price $p_t=p_{G'}(\boldsymbol{x}_t)$ with
\begin{equation}\label{p3}
p_{G'}(\boldsymbol{x}_t)=\left\{
\begin{aligned}
&-\frac{\hat{\boldsymbol{\beta}}_{G_t'}^\top \tilde{\boldsymbol{x}}_t}{2\hat{\alpha}_{G_t'}}, &&  \text{if}\ \  \frac{\hat{\boldsymbol{\beta}}_{1}^\top \tilde{\boldsymbol{x}}_t}{2\hat{\alpha}_{1}}-\frac{\hat{\boldsymbol{\beta}}_{0}^\top \tilde{\boldsymbol{x}}_t}{2\hat{\alpha}_{0}}\leq  \delta-c_\delta\sqrt{\frac{\log T_0}{T_0}},  \\
&\hat{\boldsymbol{\gamma}}_{1}^\top \tilde{\boldsymbol{x}}_t-\delta \cdot G_t'+\hat{\gamma}_{2}, && \text{if}\  \ \frac{\hat{\boldsymbol{\beta}}_{1}^\top \tilde{\boldsymbol{x}}_t}{2\hat{\alpha}_{1}}-\frac{\hat{\boldsymbol{\beta}}_{0}^\top \tilde{\boldsymbol{x}}_t}{2\hat{\alpha}_{0}}> \delta-c_\delta\sqrt{\frac{\log T_0}{T_0}}, 
\end{aligned}
\right.
\end{equation}
where
$$\hat{\boldsymbol{\gamma}}_{1}=-\frac{q\hat{\boldsymbol{\beta}}_{0}+(1-q)\hat{\boldsymbol{\beta}}_{1}}{2q\hat{\alpha}_{0}+2(1-q)\hat{\alpha}_{1}},\ \hat{\gamma}_{2}=\frac{(1-q)\hat{\alpha}_{1}\delta}{q\hat{\alpha}_{0}+(1-q)\hat{\alpha}_{1}}.$$
\STATE The seller receives a demand $y_{t}$.
\STATE The seller releases $(p_{t},  G_t', \boldsymbol{x}_t)$ to the public.
\ENDFOR
\end{algorithmic}
\end{algorithm}

Algorithm \ref{alg1} needs the time horizon $T$ and three additional input parameters. We can use the doubling trick \citep{lattimore} to handle cases where $T$ is unknown, achieving the same regret bound. The first input is the upper bound of the price $B$, assumed to be known in Assumption \ref{ass0}, aligning with previous works such as  \cite{luo2022,Fan2022}. Here, we only need an upper bound on the price and a rough upper bound $B$ is sufficient. The second input $\tau$ determined the exploration length. The third input $c_\delta$ relates to different price offers. We would like to mention that the choices of these hyperparameters are not sensitive and in Section \ref{sec5} we provide a sensitivity test on the choices of these three input parameters. In Algorithm \ref{alg1}, the proportion $q$ is assumed to be fixed and known to the seller, as done in \cite{Xu2023}.  In practice, if $q$ is unknown, the seller can estimate it using prior information. We now delve into the exploration phase and the exploitation phase.

\textbf{Exploration Phase}. During this phase, the seller announces that a uniform pricing policy is implemented. At each time $t$, a buyer with $\boldsymbol{x}_t$ and $G_t$ enters the market. The seller provides the same price $p_t\sim$ Unif$(0,B)$ for both groups. In this case, buyers lack incentives to modify their private types and hence will reveal the true group status $G_t$ \citep{harris2022,liu2023contextual}. After observing $p_t$, the buyer decides the demand $y_t$. The seller collects data $(\boldsymbol{x}_t, G_t, y_t, p_t)$. At the end of the exploration phase, the seller use the sale dataset $\{(\boldsymbol{x}_t, G_t, y_t, p_t)\}_{t=1}^{T_0}$, to estimate $\boldsymbol{\theta}_j=(\alpha_j, \boldsymbol{\beta}_j^\top)^\top$ for $j\in\{0,1\}$, see \eqref{est}. 

\textbf{Exploitation Phase}. During this phase, the seller enacts an optimal fairness-aware pricing policy. At each time $t$, a buyer with $\boldsymbol{x}_t$ and $G_t$ enters the market and reports the group status as $G_t'$. The advantage group releases $G_t'=G_t$ and the disadvantage group reports $G_t'$ by \eqref{G}. The seller provides a price upon observing $\boldsymbol{x}_t$ and $G_t'$ using \eqref{p3}. The pricing function \eqref{p3} is not simply a direct plug-in estimate of \eqref{eq6}. Here, we discuss the intuition behind \eqref{p3}. The additional term $c_\delta\sqrt{\frac{\log T_0}{T_0}}$ accounts for the uncertainty in the estimates $\hat{\boldsymbol{\theta}}_0$ and $\hat{\boldsymbol{\theta}}_1$. While $\frac{\hat{\boldsymbol{\beta}}_{1}^\top \tilde{\boldsymbol{x}}_t}{2\hat{\alpha}_{1}}-\frac{\hat{\boldsymbol{\beta}}_{0}^\top \tilde{\boldsymbol{x}}_t}{2\hat{\alpha}_{0}}$  is not identical to $\frac{\boldsymbol{\beta}_1^\top \tilde{\boldsymbol{x}}_t}{2\alpha_1}-\frac{\boldsymbol{\beta}_0^\top \tilde{\boldsymbol{x}}_t}{2\alpha_0}$, Lemma \ref{lem1} ensures that the two quantities are close. To minimize regret, the seller must rely on to $\frac{\hat{\boldsymbol{\beta}}_{1}^\top \tilde{\boldsymbol{x}}_t}{2\hat{\alpha}_{1}}-\frac{\hat{\boldsymbol{\beta}}_{0}^\top \tilde{\boldsymbol{x}}_t}{2\hat{\alpha}_{0}}$ to infer the value of $\frac{\boldsymbol{\beta}_1^\top \tilde{\boldsymbol{x}}_t}{2\alpha_1}-\frac{\boldsymbol{\beta}_0^\top \tilde{\boldsymbol{x}}_t}{2\alpha_0}$. When $\frac{\hat{\boldsymbol{\beta}}_{1}^\top \tilde{\boldsymbol{x}}_t}{2\hat{\alpha}_{1}}-\frac{\hat{\boldsymbol{\beta}}_{0}^\top \tilde{\boldsymbol{x}}_t}{2\hat{\alpha}_{0}}\leq  \delta-c_\delta\sqrt{\frac{\log T_0}{T_0}}$, the seller can be highly confident that $\frac{\boldsymbol{\beta}_1^\top \tilde{\boldsymbol{x}}_t}{2\alpha_1}-\frac{\boldsymbol{\beta}_0^\top \tilde{\boldsymbol{x}}_t}{2\alpha_0}\leq \delta$ holds. After receiving the price $p_t$, the buyer decides the demand $y_t$. Finally, the seller releases $(\boldsymbol{x}_t, G_t', p_t)$ to the public complying with the regulations.

\section{Theoretical Analysis}\label{sec4}
In this section, we conduct a theoretical analysis of the proposed pricing policy. We first provide the upper bounds for the estimation errors of the demand parameters and the fairness level. Subsequently, we prove that our pricing policy achieves a sublinear upper regret bound. Finally, we establish a lower regret bound for any pricing policy that adheres to the price fairness constraint, which indicates the rate-optimality of our policy.

We start with a lemma to establish an upper bound on the estimation error of the demand parameters using \eqref{est} at the end of the exploration phase. 

\begin{lemma}\label{lem1}
Suppose Assumptions \ref{ass0} and \ref{ass1} hold. The estimated parameter $\hat{\boldsymbol{\theta}}_{j}=(\hat{\alpha}_{j}, \hat{\boldsymbol{\beta}}_{j}^\top)^\top, j=0,1, $ is obtained by (\ref{est}). Let $T_0$ be the length of the exploration phase and $q$ be the proportion of buyers in group 0.  When $T_0\geq \frac{12L}{\lambda_0\mathop{\min}\{1-q, q\}}$, we have
\begin{align*}
\mathbb{E}(\|\hat{\boldsymbol{\theta}}_{j}-\boldsymbol{\theta}_{j}\|_2^2)\leq  \frac{4L[\sigma_{\epsilon}^2+\lambda_0(a_{max}^2+b_{max}^2)(d+2)]}{\lambda_0^2q_jT_0},
\end{align*} 
where $q_0=q$, $q_1=1-q$, $L=B^2+1+x_{max}^2$, and $\lambda_0=\mathop{\min}\{(B^2+3-\sqrt{B^4+3B^2+9})/6,\lambda_{min}(\Sigma_x)\}$.
\end{lemma}
Lemma \ref{lem1} indicates that the expected squared estimation error of $\hat{\boldsymbol{\theta}}_{0}$ and $\hat{\boldsymbol{\theta}}_{1}$ decreases with the exploration length $T_0$. As $T_0$ increases, the number of the samples used to estimate $\boldsymbol{\theta}_0$ and $\boldsymbol{\theta}_1$ becomes larger, leading to a better estimation accuracy. The expected squared estimation error is also influenced by the feature dimension $d$. With a larger $d$, more parameters are to be estimated, resulting in a less accurate estimation. Moreover, the error is affected by the proportion of buyer groups. When the proportion of one buyer group is larger, more samples are used to estimate its parameters, leading to a smaller estimation error of that corresponding group's parameter.

The next lemma focuses on the buyer's learning process and quantifies the estimation error of the price difference $\hat{\delta}=\hat{p}_0(\boldsymbol{x})-\hat{p}_1(\boldsymbol{x})$ obtained through an offline regression oracle.
\begin{lemma}\label{lem2}
Suppose that Assumptions \ref{ass0}, \ref{ass1} and \ref{ass2} hold. We assume $\delta<C_0$ where $C_0>0$ represents manipulation cost of the group status. At time $t$ in the exploitation phase, buyers learn the price difference $\hat{\delta}$ using the offline regression oracle \texttt{OffReg}$_\mathcal{P}$.  There exists a positive constant $c$ such that, for $t>c$, with probability at least $1 -2\eta_{t-1}$, we have $\hat{\delta}\leq C_0$.
\end{lemma}
We provide some intuitions on why our policy is able to prevent buyers' strategic behaviors. Lemma \ref{lem2} assures that $\hat{\delta} \leq C_0$ holds with high probability. By \eqref{G}, buyers are deterred from manipulation, as the manipulation cost outweighs the benefits from manipulation.

Now, we establish the upper bound of the regret of the proposed policy in Algorithm \ref{alg1}.

\begin{theorem}\label{thm2}
Let the assumptions of Lemma \ref{lem2} hold. There exist positive constants $c_1, c_2$ and $c_3$ such that when $T>c_1$,  the total expected regret of the proposed pricing policy over the time horizon $T$ satisfies
$$Regret_T\leq \sqrt{\frac{c_2dT}{q(1-q)}}+c_3qH(T),$$
where $H(T)=\sum_{t=1}^T\eta_{t}$ with $\eta_t$ defined in Assumption \ref{ass2}.
\end{theorem}
The regret bound is influenced by several key parameters. First, the buyer group proportion (\(q\)) contributes to two components of regret: \(\sqrt{\frac{1}{q(1-q)}}\), which arises from estimation errors in demand parameters and is minimized when \(q = 1/2\), indicating equal group proportions lead to the smallest average estimation error; and \(q\), which represents the regret due to the proportion of strategic buyers and increases as more disadvantaged buyers manipulate their group status. Second, the feature dimension (\(d\)) affects the regret bound, with higher \(d\) leading to greater regret due to increased parameter estimation complexity. Finally, the time horizon (\(T\)) influences the regret in two ways: the first term grows proportionally to \(\sqrt{T}\), and the second term is \(\sum_{t=1}^T \eta_{t-1}\), which captures buyers' assessment accuracy of the pricing policy's fairness from learned price differences. The next corollary will discuss detailed scenarios when this term goes to zero.  

\begin{corollary}\label{coro}
Under the assumptions of Theorem 2, if $\eta_t\leq\frac{1}{\sqrt{t}}$ when $t$ exceeds a certain constant, the total expected regret of the proposed pricing policy over the time horizon $T$ satisfies $Regret_T=O(\sqrt{T})$.
\end{corollary}
Corollary \ref{coro} follows from Theorem \ref{thm2} under the condition $\eta_t\leq\frac{1}{\sqrt{t}}$ using the fact $\sum_{t=1}^T \frac{1}{\sqrt{t}} = O(\sqrt{T})$.  This condition is attainable when buyers employ specific methods to learn the prices. For example, if buyers learn the price using the neural network, the condition $\eta_t \leq \frac{1}{\sqrt{t}}$ is met, with the offline learning guarantee $\mathcal{E}_{\mathcal{P},\eta_t}(t)$ decreasing over time \citep{banee}. In Sections \ref{sec5}, we explore different learning methods used by buyers to assess our policy.

Our next theorem establishes a theoretical lower bound on the regret for any pricing policy that adheres to the price fairness constraint. We construct a problem instance by setting the second to $d$-th components of $\boldsymbol{\beta}_j$ to be 0, for $j=0,1$.
\begin{theorem}\label{thm3}
Consider a problem instance such that the expected demand is $\mathbb{E}(y_t|G_t,p_t)=1/2+\alpha [(G_t+1)p_t-1-G_t/2]$ with $\alpha\in[-1/2,-1/5]$ and $p_t\in[1/2,9/8]$, the group status $G_t=(t\bmod 2)$, $q=1/2$, and buyers can perfectly learn the price difference. For any pricing policy $\psi$ satisfying the price fairness constraint $p_0-p_1=\delta$ with $\delta=1/4$, there exists a parameter $\alpha$ such that
$$Regret_T\geq \frac{1}{15360}\sqrt{T}.$$
\end{theorem}
Theorem \ref{thm3} gives a lower regret bound of order at least $\Omega(\sqrt{T})$ on any pricing policy satisfying the price fairness constraint over $T$ time periods. In comparison, Corollary \ref{coro} demonstrates that our proposed Algorithm \ref{alg1} achieves an upper bound of order $O(\sqrt{T})$, indicating the optimality of our pricing policy when buyers are able to learn the fairness of the price policy.

Next, we provide an intuitive explanation for proving Theorem \ref{thm3}. The detailed proof is in Section \ref{ssec5} of the appendix. The lower bound is established by constructing an uninformative price \citep{Broder2012,Xu2021,Fan2022}. A price is deemed uninformative because all demand curves intersect at a common price, and no policy can gain information about the demand parameter.
We choose the true demand parameter $\alpha$ as  $\alpha_0=-2/5$. Using \eqref{eq6}, we derive the optimal prices for group 0 and group 1 under the fairness constraint, denoted as $p_0^*(\alpha_0)=1$ and $p_1^*(\alpha_0)=3/4$, respectively. We then find that $\mathbb{E}(y_t|0, p_0^*(\alpha_0))=1/2$ and $\mathbb{E}(y_t|1, p_1^*(\alpha_0))=1/2$, indicating that all demand curves intersect at the optimal prices when the underlying parameter is $\alpha_0$. These optimal prices provide no information on the estimation of the demand parameter. We demonstrate that if a policy tries to learn model parameters, it must set prices away from the uninformative prices $p_0^*(\alpha_0)$ and $p_1^*(\alpha_0)$, thereby incurring large regret when the underlying parameter is $\alpha_0$. Furthermore, we establish that any policy failing to accurately learn the demand parameter $\alpha$ must also incur a cost in regret. Combining these facts, we can prove that any fair pricing policy achieves a regret lower bound $\Omega(\sqrt{T})$ in the setting presented in Theorem \ref{thm3}.

\subsection{Outline of the Proof of Theorem \ref{thm2}}
In the following we give an outline for the proof of Theorem \ref{thm2}, summarizing its main steps. The main idea behind our regret analysis is a balance between exploration and exploitation, and the discouragement of the strategic behavior. Our proof differs from existing proofs for fair dynamic pricing policies, requiring careful quantification of the regret loss due to strategic behaviors.

The time period is segmented into the exploration phase and the exploitation phase.  The seller’s revenue at time $t$ is $R_j(p_t)=R_j(p_t, \boldsymbol{x}_t)$ for $j=0, 1$. Let $p_{0t}$ and $p_{1t}$ be the prices offered to group 0 and group 1, respectively. Let
$reg_t=qR_0(p^*_{0t})+(1-q)R_1(p^*_{1t})-qR_0(p_{0t})-(1-q)R_1(p_{1t})$ be the instance regret under Algorithm \ref{alg1} at time period $t$. Under Assumption \ref{ass0}, the regret at time $t$ in the exploration phase is $\mathbb{E}(reg_t)=O(1)$.

Now, we focus on the analysis of the regret during the exploitation phase. During the exploitation phase, The pricing function \eqref{p3} is equivalent to 
\begin{equation*}
p_t=\left\{
\begin{aligned}
&-\frac{\hat{\boldsymbol{\beta}}_{G_t'}^\top \tilde{\boldsymbol{x}}_t}{2\hat{\alpha}_{G_t'}}, &&  \text{if}\ \  \frac{\hat{\boldsymbol{\beta}}_{1}^\top \tilde{\boldsymbol{x}}_t}{2\hat{\alpha}_{1}}-\frac{\hat{\boldsymbol{\beta}}_{0}^\top \tilde{\boldsymbol{x}}_t}{2\hat{\alpha}_{0}}\leq  \delta-c_\delta\sqrt{\frac{\log T_0}{T_0}},  \\
&\hat{\boldsymbol{\gamma}}_{1}^\top \tilde{\boldsymbol{x}}_t-\delta \cdot G_t'+\hat{\gamma}_{2}, && \text{if}\  \ \frac{\hat{\boldsymbol{\beta}}_{1}^\top \tilde{\boldsymbol{x}}_t}{2\hat{\alpha}_{1}}-\frac{\hat{\boldsymbol{\beta}}_{0}^\top \tilde{\boldsymbol{x}}_t}{2\hat{\alpha}_{0}}\geq \delta+c_\delta\sqrt{\frac{\log T_0}{T_0}}, \\
&\hat{\boldsymbol{\gamma}}_{1}^\top \tilde{\boldsymbol{x}}_t-\delta \cdot G_t'+\hat{\gamma}_{2}, && \text{if}\  \ \delta-c_\delta\sqrt{\frac{\log T_0}{T_0}}<\frac{\hat{\boldsymbol{\beta}}_{1}^\top \tilde{\boldsymbol{x}}_t}{2\hat{\alpha}_{1}}-\frac{\hat{\boldsymbol{\beta}}_{0}^\top \tilde{\boldsymbol{x}}_t}{2\hat{\alpha}_{0}}< \delta+c_\delta\sqrt{\frac{\log T_0}{T_0}},
\end{aligned}
\right.
\end{equation*}
where $G_t'$ signifies the disclosed group status by the buyer. Given the strategic nature of buyers, there exists the possibility of them revealing a false group status.
We show that the probability $\mathbb{P}\left(\delta-c_\delta\sqrt{\frac{\log T_0}{T_0}}<\frac{\hat{\boldsymbol{\beta}}_{1}^\top \tilde{\boldsymbol{x}}_t}{2\hat{\alpha}_{1}}-\frac{\hat{\boldsymbol{\beta}}_{0}^\top \tilde{\boldsymbol{x}}_t}{2\hat{\alpha}_{0}}< \delta+c_\delta\sqrt{\frac{\log T_0}{T_0}}\right)$ is $O(1/T_0)$.

The buyers from the group 1 (advantage group) do not manipulate and reveal the true group type, the price for these buyers under our policy is defined as
\begin{equation*}
p_{1t}=\left\{
\begin{aligned}
&-\frac{\hat{\boldsymbol{\beta}}_{1}^\top \tilde{\boldsymbol{x}}_t}{2\hat{\alpha}_{1}}, & \text{if}\ \  \frac{\hat{\boldsymbol{\beta}}_{1}^\top \tilde{\boldsymbol{x}}_t}{2\hat{\alpha}_{1}}-\frac{\hat{\boldsymbol{\beta}}_{0}^\top \tilde{\boldsymbol{x}}_t}{2\hat{\alpha}_{0}}\leq \delta-c_\delta\sqrt{\frac{\log T_0}{T_0}},\\
&\hat{\boldsymbol{\gamma}}_{1}^\top \tilde{\boldsymbol{x}}_t-\delta+\hat{\gamma}_{2}, & \text{if}\  \ \frac{\hat{\boldsymbol{\beta}}_{1}^\top \tilde{\boldsymbol{x}}_t}{2\hat{\alpha}_{1}}-\frac{\hat{\boldsymbol{\beta}}_{0}^\top \tilde{\boldsymbol{x}}_t}{2\hat{\alpha}_{0}}> \delta+c_\delta\sqrt{\frac{\log T_0}{T_0}}.
\end{aligned}
\right.
\end{equation*}
The price for buyers from group 0 is contingent on the group status they reveal. Let $\hat{\delta}$ be the price difference that the buyers from group 0 estimated based on the public data. If $\hat{\delta}>C_0$, the buyers from group 0 reveal a manipulated group type. Conversely, if $\hat{\delta}\leq C_0$, they disclose their true group type. Consequently, under our policy, the price for buyers in group 0 is given by
\begin{equation*}
p_{0t}=\left\{\begin{array}{lr}
p_{0t}', \ \text{if}\  \hat{\delta}\leq C_0,\\
p_{1t}, \ \text{if}\  \hat{\delta}> C_0,
\end{array}
\right.    
\end{equation*}
where
\begin{equation*}
p_{0t}'=\left\{\begin{aligned}
    &-\frac{\hat{\boldsymbol{\beta}}_{0}^\top \tilde{\boldsymbol{x}}_t}{2\hat{\alpha}_{0}},\ & \text{if}\  \ \frac{\hat{\boldsymbol{\beta}}_{1}^\top \tilde{\boldsymbol{x}}_t}{2\hat{\alpha}_{1}}-\frac{\hat{\boldsymbol{\beta}}_{0}^\top \tilde{\boldsymbol{x}}_t}{2\hat{\alpha}_{0}}\leq  \delta-c_\delta\sqrt{\frac{\log T_0}{T_0}},\\
     &\hat{\boldsymbol{\gamma}}_{1}^\top \tilde{\boldsymbol{x}}_t+\hat{\gamma}_{2},\  & \text{if}\  \ \frac{\hat{\boldsymbol{\beta}}_{1}^\top \tilde{\boldsymbol{x}}_t}{2\hat{\alpha}_{1}}-\frac{\hat{\boldsymbol{\beta}}_{0}^\top \tilde{\boldsymbol{x}}_t}{2\hat{\alpha}_{0}}> \delta+c_\delta\sqrt{\frac{\log T_0}{T_0}}.
\end{aligned}
\right.
\end{equation*}
The regret of our policy depends on the probability $\mathbb{P}(\hat{\delta}\leq C_0)$. We define the historical information up to time $t$ as $\mathcal{H}_t=\{\boldsymbol{x}_1,\cdots,\boldsymbol{x}_t, G_1',\cdots, G_{t+1}', y_1,\cdots, y_{t}, p_1,\cdots,p_t\}$. We also define $\tilde{\mathcal{H}}_t=\mathcal{H}_t\cup \{\boldsymbol{x}_{t+1}, G_t'\}$ as the filtration including $\boldsymbol{x}_{t+1}$ and $G_{t+1}'$. Given that the true group status at time $t$ is unknown, the expected regret at time $t$ in the exploitation phase is 
\begin{equation*}
\begin{aligned}
\mathbb{E}(reg_t|\tilde{\mathcal{H}}_{t-1})&\leq \underbrace{\mathbb{E}\{q[R_0(p^*_{0t})-R_0(p_{0t}')]+(1-q)[R_1(p^*_{1t})-R_1(p_{1t})]|\tilde{\mathcal{H}}_{t-1}\}}_{J_1}\\
&~~~+\underbrace{q\mathbb{E}\{q[R_0(p^*_{0t})-R_0(p_{1t})]+(1-q)[R_1(p^*_{1t})-R_1(p_{1t})]|\tilde{\mathcal{H}}_{t-1}\}\mathbb{P}(\hat{\delta}> C_0)}_{J_2}.
\end{aligned}
\end{equation*}
The expected regret at time $t$ is upper bounded by two parts: $J_1$ and $J_2$. In $J_1$, the price offered to group 0 is $p_{0t}=p_{0t}'$, indicating that no strategic behavior happens. Conversely, in $J_2$, the price offered to group 0 is $p_{0t}=p_{1t}$, signifying that the buyer from group 0 misreports as belonging to group 1.

We can prove that $J_1$ is upper bounded by $\|\hat{\boldsymbol{\theta}}_{0}-\boldsymbol{\theta}_0\|_2^2+\|\hat{\boldsymbol{\theta}}_{1}-\boldsymbol{\theta}_1\|_2^2$. Noting that the number of samples from the exploration phase to estimate $\boldsymbol{\theta}_0$ and $\boldsymbol{\theta}_1$ is $T_0$, we can prove $J_1=O(1/T_0)$ using Lemma \ref{lem1}. To establish an upper bound on $J_2$, we leverage Lemma \ref{lem2} to upper bound $\mathbb{P}(\hat{\delta}> C_0)$. We can prove $J_2=O(\eta_{t-1})$. The expected regret at time $t$ in the exploitation phase is 
\begin{equation*}
\mathbb{E}(reg_t)=\mathbb{E}[\mathbb{E}(reg_t|\tilde{\mathcal{H}}_{t-1})]=O\left(\frac{1}{T_0}\right)+O(\eta_{t-1}).
\end{equation*}
Finally, the cumulative regret across the exploration phase and the exploitation phase  is
\begin{equation*}
Regret_T= O(T_0)+ O\left(\frac{T-T_0}{T_0}\right)+O\left(\sum_{t=1}^T\eta_{t}\right)= O(\sqrt{T})+O\left(\sum_{t=1}^T\eta_{t}\right),
\end{equation*}
where the last equality holds at $T_0=O(\sqrt{T})$, achieving an optimal trade-off between exploration and exploitation.

\section{Simulation Study}\label{sec5}
In this section, we implement simulation studies to demonstrate the effectiveness of our policy. In the experiments (except for Section \ref{sec5.2}), we set the feature dimension as $d=3$. The demand parameters for buyers from group 0 is $\alpha_0=-1, \boldsymbol{\beta}_0=(2, 1/2, 1, 1)^\top$.  The demand parameters for buyers from group 1 is $\alpha_1=-1, \boldsymbol{\beta}_1=(1, 1/4, 1/2, 1/2)^\top$. 
Assume the cost of manipulation is $C_0=0.8$. The seller would naturally select \( \delta < C_0 \) to discourage strategic behavior. In the experiments, the fairness level in the pricing policy is set to \( \delta = 0.799 \). We denote the feature vector as $\boldsymbol{x}_t=(x_{1t}, x_{2t}, x_{3t})^\top$ and the features $x_{1t}, x_{2t}$ and $x_{3t}$ are all i.i.d. from Unif(-2, 2). The noise distribution is $\epsilon_t\sim N(0, 1)$. Due to space limitations, additional simulations, including sensitivity analyses on hyperparameter choices and misspecified demand models, are provided in the appendix.

As previous online pricing policies with fairness constraint \citep{Xu2023,chen2023utility,chen2023, Cohen2021} do not take strategic behaviors into consideration, they are not applicable in our problem. For the comparison, we consider the pricing policy without  buyers' fairness learning process mentioned in Section \ref{linearreg} as a benchmark policy. In this policy, the seller provides prices with fairness constraints to strategic buyers while the buyers do not learn the fairness level. Without the fairness perception, the buyers always act strategically. 

In both algorithms, we divide the time horizon into an exploration phase with length $T_0$, and an exploration phase with length $T-T_0$. In the exploration phase, the seller randomly samples $p_t$ at each time $t$, and obtains the estimate $\hat{\boldsymbol{\theta}}_{0}$ and $\hat{\boldsymbol{\theta}}_{1}$ at the end of the exploration phase. Without loss of generality, we consider that the group 0 is the disadvantaged group. During the exploitation phase, the buyers from group 0 learn the price disparity $\hat{\delta}$ using two learning methods: decision tree regression and neural networks. In the decision tree regression approach, buyers input the feature vector $\boldsymbol{x}$ and group status $G$ into a decision tree with a maximum depth of 5, which outputs the price $p$. In the neural network approach, the feature vector $\boldsymbol{x}$ and groups status $G$ are fed into a neural network consisting of 5 hidden layers, each with 5 neurons, to predict the price $p$.  The price disparity $\hat{\delta}$ is calculated as $\hat{\delta}=\hat{p}_0(\boldsymbol{x})-\hat{p}_1(\boldsymbol{x})$ using the estimated prices from the two models. If $\hat{\delta}\leq C_0$, they report the true group status. Otherwise, they misreport the group status. The buyers from group 1 always report their true group status.

\subsection{Regret Comparison}\label{sec5.1}
Buyers' perception of fairness is crucial. In our policy, we incorporate the learning process of fairness for buyers, which discourages strategic behaviors. We set $B=3$ and $\tau=10$. Figure \ref{fig6} shows the regrets of our policy and the compared benchmark policy without buyers' fairness learning. The benchmark policy exhibits larger regrets compared to ours. When buyers use a neural network to learn the price difference, the regret is smaller compared to using decision tree regression. This is because the neural network is more effective at modeling the price, and presents a smaller $H(T)$ defined in Theorem \ref{thm2}, leading to a smaller regret. We illustrate the estimation errors of $\delta$ via the neural network and decision tree models in Figure~\ref{fignn} in Section B of the appendix. In all subsequent experiments, we will concentrate on using the neural network as the buyer's learning method.
\begin{figure}[t!]
    \centering
    \begin{tabular}{cc}  
        \includegraphics[scale = 0.48]{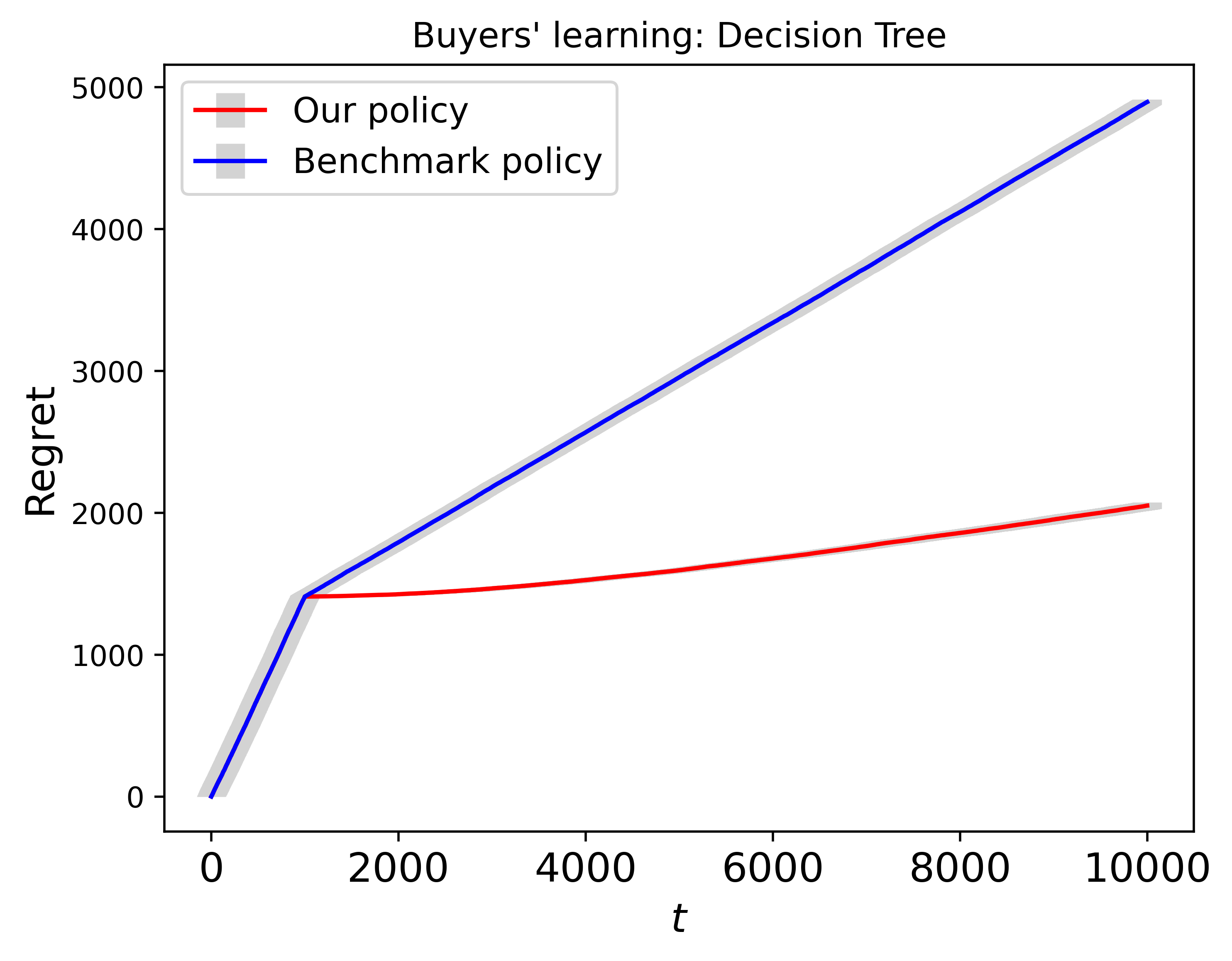}&
        \includegraphics[scale = 0.48]{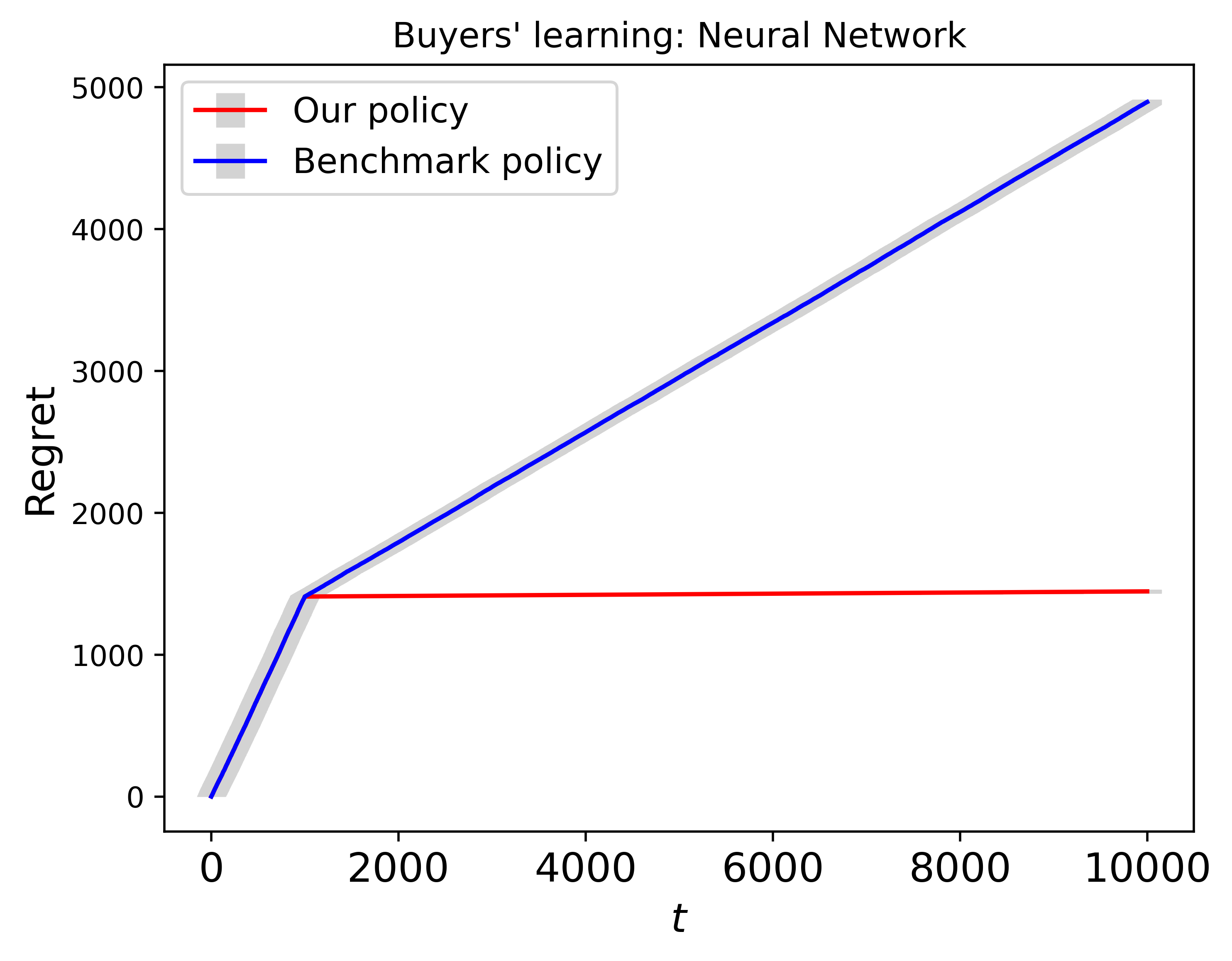}
    \end{tabular}
     \caption{Regret plots for the two policies. The two subplots show the regrets of two different scenarios: decision tree regression and neural network. The red and blue lines represent the mean regret of our policy and the benchmark policy, respectively, over 20 independent runs. The light gray areas around these lines depict the standard errors of the estimates.}
         \label{fig6}
\end{figure}
\subsection{Impacts of $d$ and $q$}\label{sec5.2}
In this section, we explore the impacts of the feature dimension $d\in\{3, 10\}$ and the proportion of strategic buyers $q\in\{0.5,0.8\}$ on our proposed pricing policy. In the left sub-figure of Figure \ref{fig30}, we fix $q=0.5$ and vary the dimensions of feature as $d=3$ and $d=10\ (\boldsymbol{\beta}_0=(2, 1/2, 1, 1,1/2,1/2,1/2,1/2,1/2,1/2,1/2)^\top, \boldsymbol{\beta}_1=\boldsymbol{\beta}_0/2)$. We observe that our policy under the higher dimension incurs a higher regret. 
In the right sub-figure of Figure \ref{fig30}, the regret for $q=0.8$ is larger than that for $q=0.5$. When the proportion of the strategic buyers is smaller, manipulation behaviors occur less frequently, resulting in a lower regret. Besides, an equal proportion of two groups achieves the smallest average estimation error, leading to a lower regret. These observations align with the theoretical findings of Theorems \ref{thm2}.
\begin{figure}[t!]
    \centering
    \begin{tabular}{ccc}  
        \includegraphics[scale = 0.48]{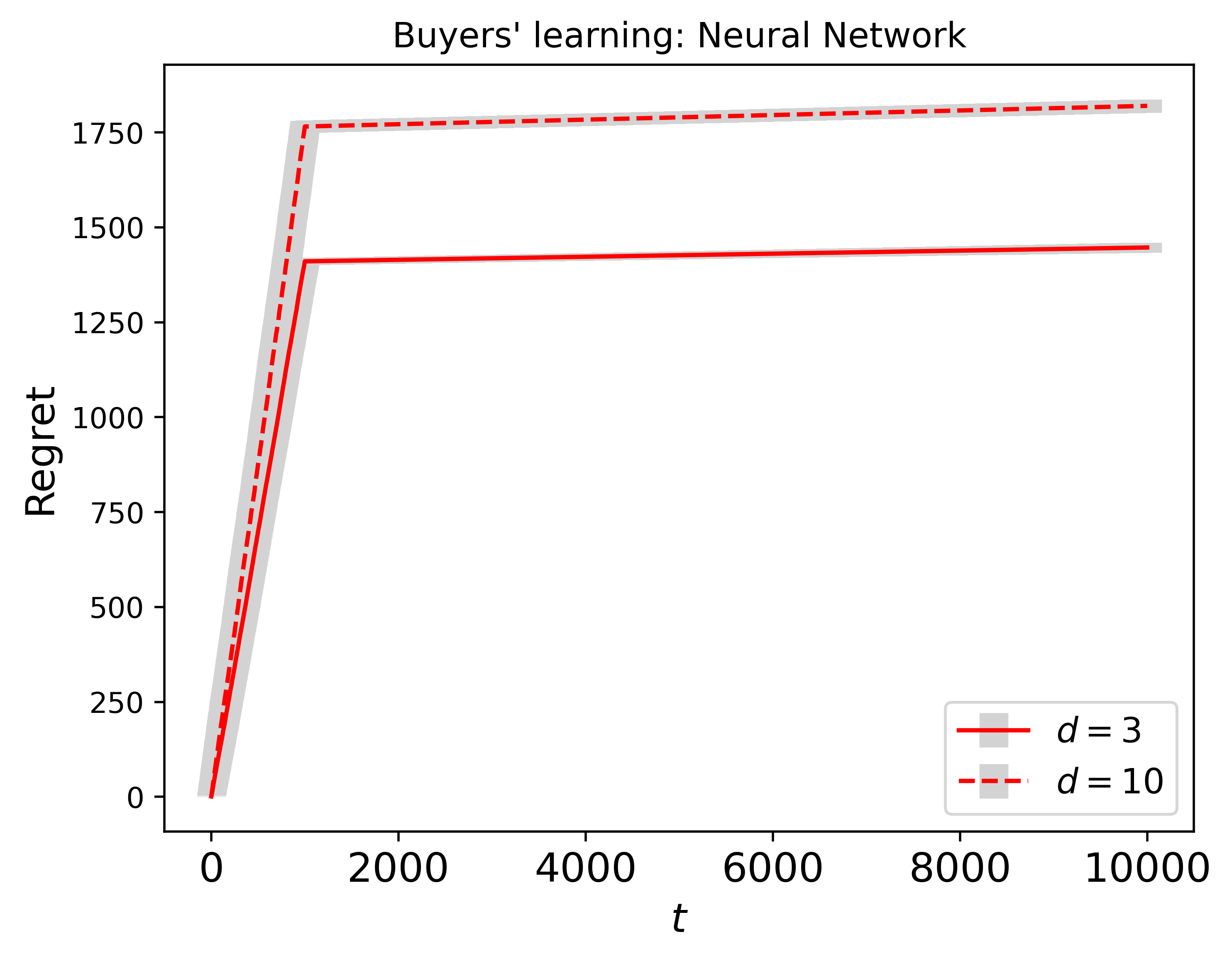}&
        \includegraphics[scale = 0.48]{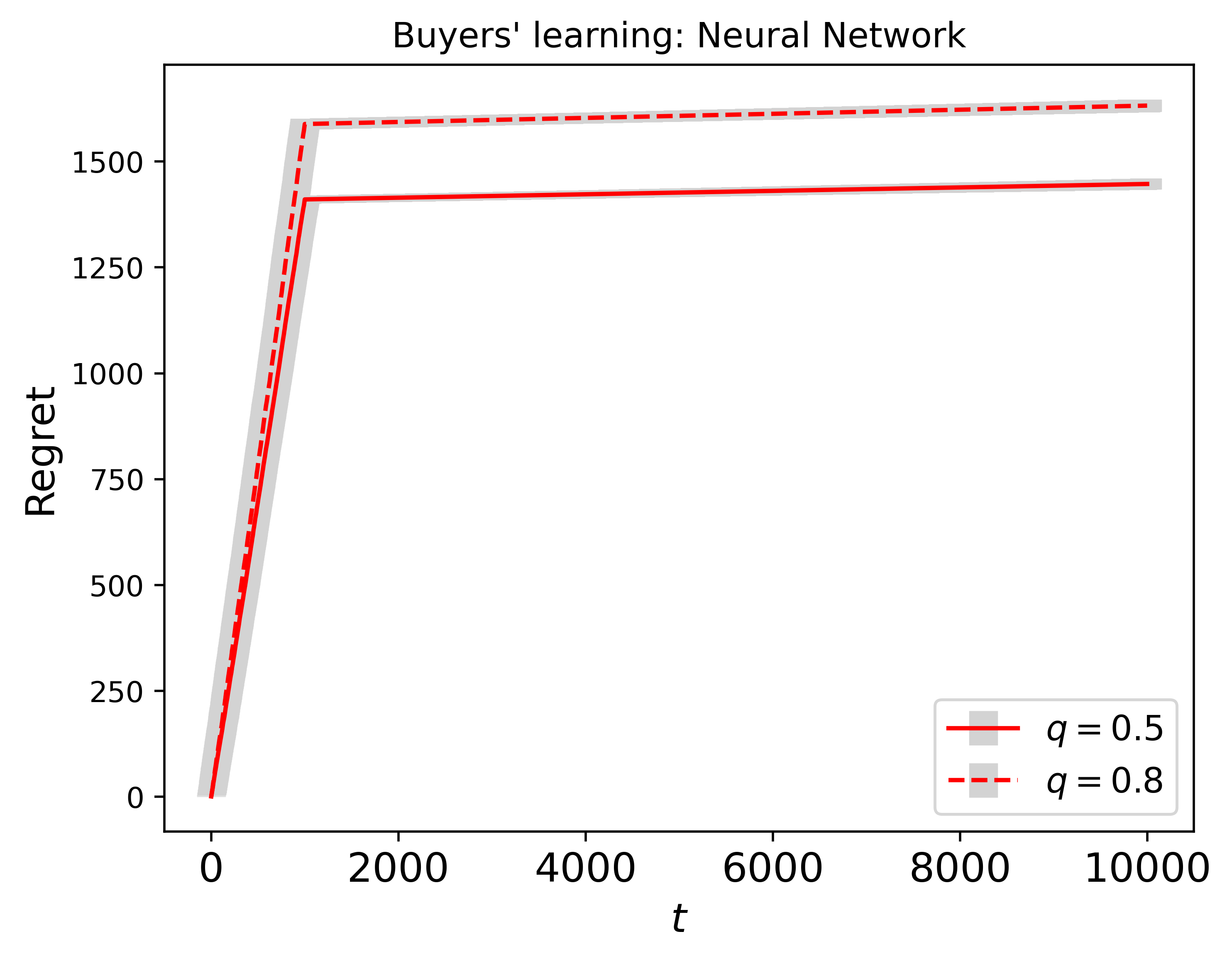}
    \end{tabular}
     \caption{Regret plots for our policy. The two subplots show the regrets of the policy at different values of $d$ and $q$. The remaining caption is the same as Figure \ref{fig6}.}
         \label{fig30}
\end{figure}

\section{Real Application}\label{sec6}
In this section, we evaluate the efficacy of our policy using the public Home Mortgage Disclosure Act dataset that includes customer characteristics and loan features. In our study, we employ the 2022 HMDA dataset\footnote{\href{https://ffiec.cfpb.gov/data-publication/dynamic-national-loan-level-dataset/2022}{https://ffiec.cfpb.gov/data-publication/dynamic-national-loan-level-dataset/2022}}. HMDA mandates that numerous financial institutions maintain, report, and publicly disclose mortgage-related information. This dataset has been a focal point of prior research \citep{Debbie2008,Zhang2018,Robert2022,Popick2022,Butler2023}, revealing that borrowers from minority groups often face higher interest rates even after accounting for contextual factors. 
\subsection{Data Description and Preprocessing}
We only retain data for loans that have been approved and disbursed to borrowers. We designate race as the group status, loan amount as the demand, interest rate as the price, and contextual information comprises income, age, property value securing the loan, debt-to-income ratio, combined loan-to-value ratio, loan term, state code, discount points, lender credits, and loan purpose. Race is categorized into two groups: White and non-White, encompassing Black or African American, American Indian or Alaska Native, Native Hawaiian or Other Pacific Islander, Asian and other minority races. We consider borrowers aged between 25 and 74. Initially provided in discrete intervals, we transform age into continuous data by averaging within each interval. The debt-to-income ratio is limited to the range of 20\% to 60\%. To ensure data integrity, we eliminate outliers, excluding the upper 5\% and lower 5\% of values, for loan amount, interest rate, income, property value, combined loan-to-value ratio, and loan term. To assess robustness, we conduct a sensitivity analysis by removing the upper and lower 1\% of these covariates in Appendix \ref{ara}. Following data preprocessing, our dataset comprises 480,776 records. Among them, the White group constitutes 78.13\% of the dataset. 

\subsection{Data Analysis}
Before applying our pricing policy on the dataset, we employ the propensity score matching method \citep{Ho2011} to examine the interest rate difference between the White group and the non-White group after controlling the contextual information (see details in Appendix \ref{ara}).  In the matched dataset, the average interest rate for the non-White group is 0.013 higher than that of the White group. This finding aligns with previous research \citep{Debbie2008,Zhang2018,Robert2022,Popick2022,Butler2023}.  In addition, we describe the misreporting mechanism in Appendix \ref{ara}.

\begin{figure}[t!]
    \centering
    \begin{tabular}{ccc}  
                \includegraphics[scale = 0.31]{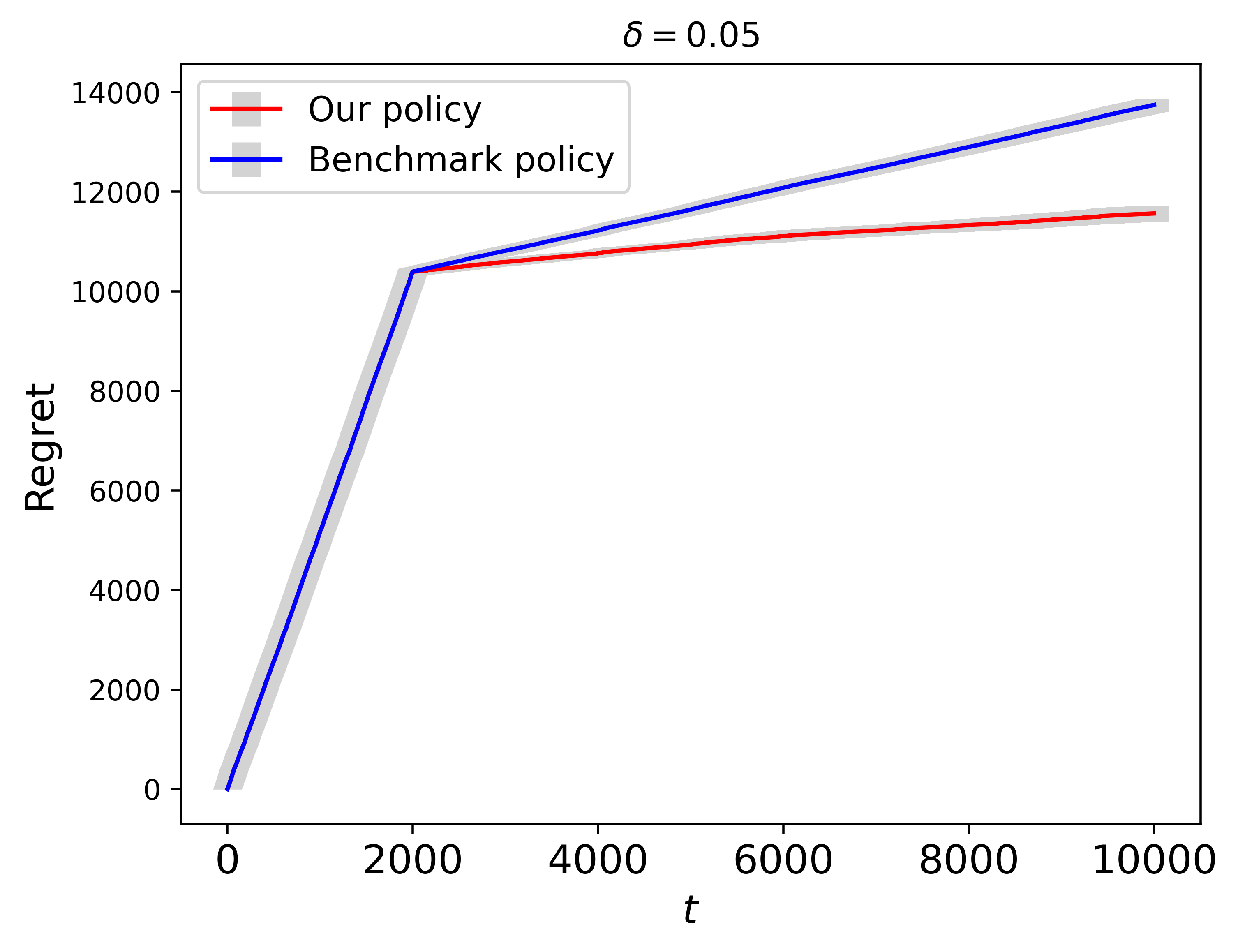}&
        \includegraphics[scale = 0.31]{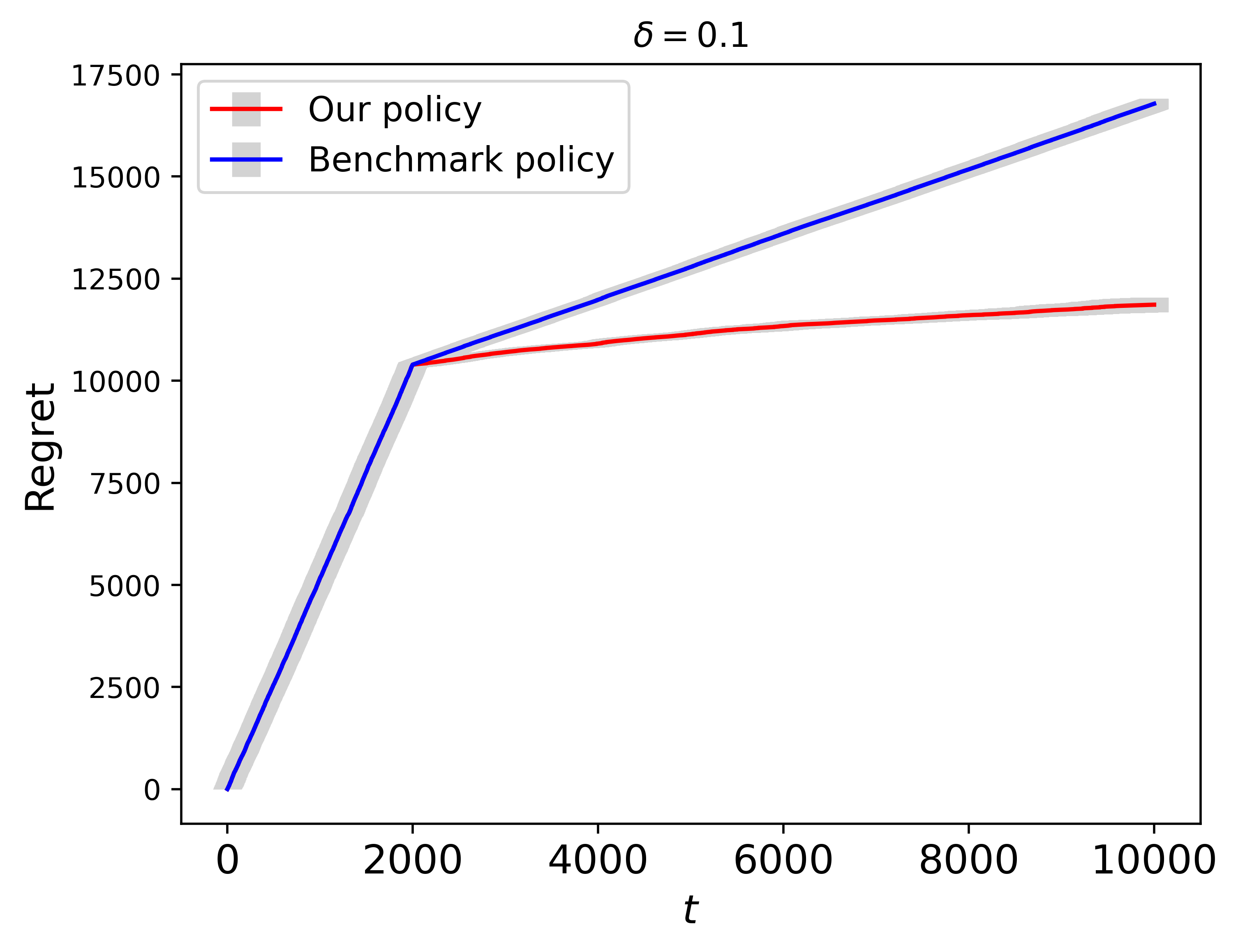}&
        \includegraphics[scale = 0.31]{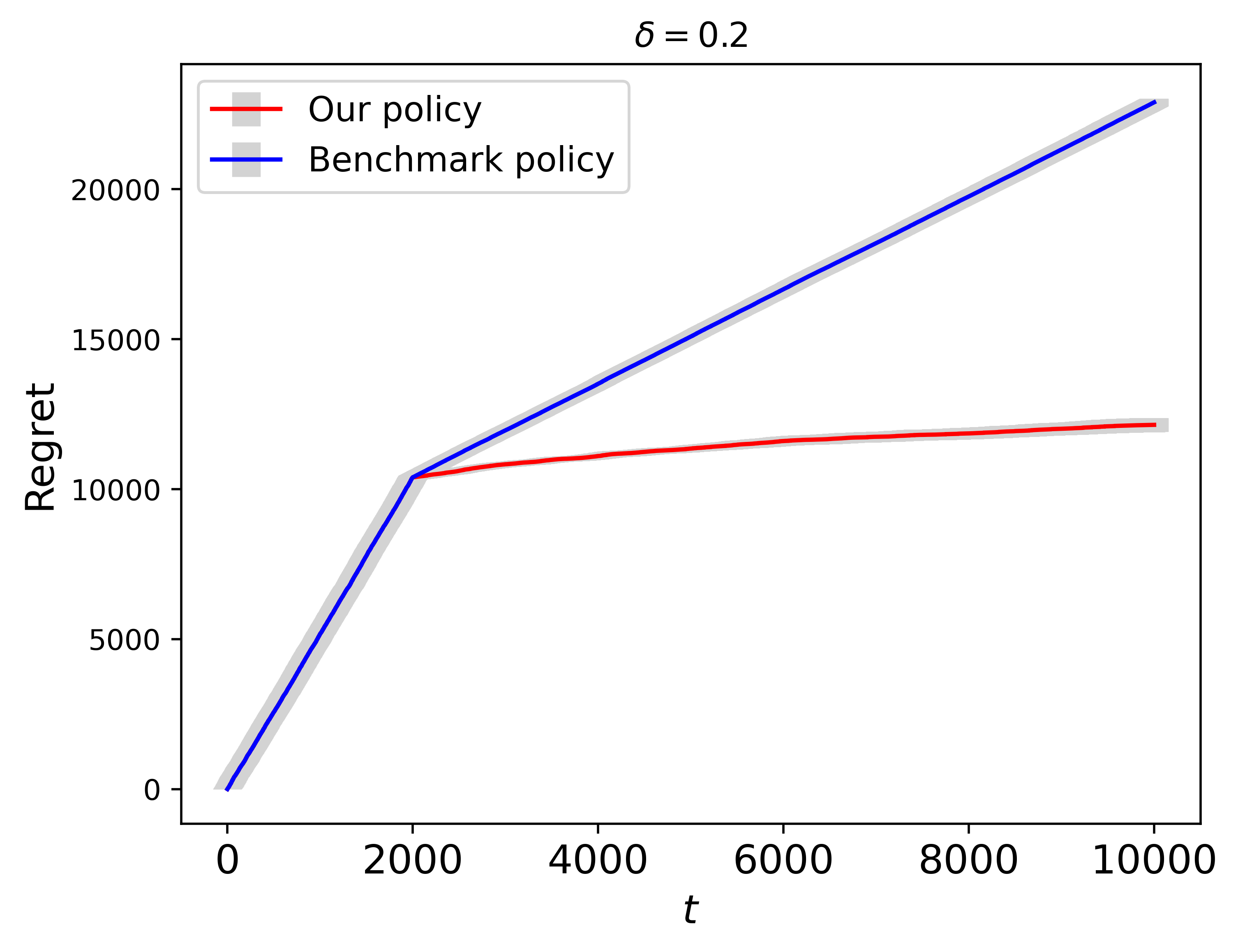}
    \end{tabular}
     \caption{Regret plots for the two policies. The three subplots show the regrets of three different scenarios, $\delta\in\{0.05, 0.1, 0.2\}$. }
          \label{fig7}
\end{figure}
Now, we analyze the dataset using our online pricing policy. In practice, obtaining real-time feedback from buyers regarding any dynamic pricing strategy is challenging until the pricing policy has been executed in the data collection system. Consequently, we adhere to the calibration approach outlined in previous studies \citep{Fan2022,wang2022,liu2023contextual} by initially estimating a linear demand model. The details of the estimated demand model can be found in Appendix \ref{ara}. This model serves as a ground truth to evaluate dynamic pricing policies. To determine a reasonable fairness threshold $\delta$, we examine the rate difference within each matched pair. The third quartile of these within-pair differences is 0.844, suggesting that most differences fall below this value. Based on these findings,  the fairness constraint is effective if the fairness threshold $\delta$ is smaller than 0.844. We set the fairness threshold values to $\delta=0.05, 0.1$ and $0.2$ to assess the impact of different fairness levels on outcomes. The manipulation cost is set to $C_0=\delta+0.01$. We choose $B=7$, which corresponds to the upper bound of observed mortgage interest rates in the dataset (with the maximum being approximately 6.785\%). We assume $\epsilon_t\sim N(0, 0.01)$, and fix $\tau=20, c_\delta=0.01$. 
Our policy is applied with the time horizon $T=10000$, repeated 20 times to record average cumulative regrets. In the exploitation phase, buyers learn the price difference using a neural network and then compare it with the manipulation cost. If the price difference is larger than the manipulation cost, buyers will misreport their group status. 
Deliberately providing false information in a mortgage application may raise legal and ethical issues. In practice, applicants may engage in more indirect forms of strategic behavior that influence how their group identity is perceived by lenders. In the mortgage loan application process, the lender may need to assess the property value of the applicant. During this stage, a minority applicant might ask a friend from the majority group to be present in their place to influence the perceived racial identity. We compare our policy against the benchmark pricing policy without considering buyers' fairness learning. Figure \ref{fig7}  shows that the cumulative regret of the policy without buyers' learning is much larger and grows significantly faster compared to our pricing policy, aligning with our earlier findings in simulated data. Our proposed pricing policy achieves an average reduction of 30.71\% in regret  compared to the benchmark policy at the end of the time horizon.

\section{Summary and Future Directions}\label{sec7}
In this paper, we study the contextual dynamic pricing problem when significant price disparities emerge among specific demographic groups such as gender or race. These disparities not only lead to legal concerns, but also incentivize disadvantaged buyers to strategically manipulate their group identity to obtain lower prices, further complicating the fairness landscape. To tackle these challenges, we propose a fairness-aware contextual dynamic pricing policy, considering scenarios where buyers' group status is private and unobservable by the seller. Our policy addresses both price fairness and strategic behavior, simultaneously.

Looking ahead, there are several promising avenues for future exploration in this area. We can extend our investigation to incorporate additional complexities such as strategic pricing problems with censored demand \citep{qi2022offline}, unobserved confounding \citep{yu2022strategic,  miao2023personalized,qi2023proximal,Shi2022}, offline learning \citep{duan2021risk,duan2024taming} and other fairness constraints \citep{fang2023fairness}. Exploring these avenues will aid in developing more robust and equitable dynamic pricing strategies.

\section*{Acknowledgment}
The authors thank the editor Professor Jianqing Fan, the associate editor and three anonymous reviewers for their valuable comments and suggestions which led to a much improved paper. Will Wei Sun’s research was partially supported by National Science Foundation (Award 2217440). Any opinions, findings, and conclusions expressed in this material are those of the authors and do not reflect the views of the funding agency. The authors report there are no competing interests to declare. 

\baselineskip=18pt
\bibliographystyle{asa}
\bibliography{reference}

\newpage
\appendix 
\baselineskip=24pt
\setcounter{page}{1}
\setcounter{equation}{0}
\setcounter{section}{0}
\renewcommand{\thesection}{S.\arabic{section}}
\renewcommand{\thelemma}{S\arabic{lemma}}
\renewcommand{\theequation}{S\arabic{equation}}

\begin{center}
{\Large\bf Supplementary Materials} \\
\medskip
{\Large\bf ``Fairness-aware Contextual Dynamic Pricing with Strategic Buyers''}  \\
\bigskip
\vspace{0.2in} 
\end{center}
\bigskip

\noindent
In this supplement, we provide additional information including sensitivity tests in Section \ref{st}, estimation errors of neural network and decision
tree in Section \ref{ee}, additional details of real application in Section \ref{ara}, misspecified demand models in Section \ref{md}, extension to multi-group settings in Section \ref{multi}, one-time fixed cost in Section \ref{ot}, discussion on lower bound in Section \ref{do}, and the framework of a Nash equilibrium model in Section \ref{fo}. The detailed proofs are provided in Section \ref{td}. Section \ref{ssec6} includes the supporting technical lemmas.

\appendix
\section{Sensitivity Tests}\label{st}
In this section, we investigate the sensitivity of our policy to hyperparameters $B, \tau, c_\delta$ required in Algorithm \ref{alg1}. Here, $B$ represents an upper bound on the price, and $\tau$ is a constant used in determining the exploration length, and  $c_\delta$ is a parameter in our policy. To assess the sensitivity of our policies, we conduct experiments with different values of these hyperparameters in the setting with $q=0.5$.

First, we examine the sensitivity of $B$. For these simulations, we set $ c_\delta=1$ and $\tau=10$. Figure \ref{fig4} illustrates the regrets of the three policies under three scenarios: $B=3$, $B=4$ and $B=5$.
\begin{figure}[t!]
    \centering
    \begin{tabular}{ccc}  
        \includegraphics[scale = 0.31]{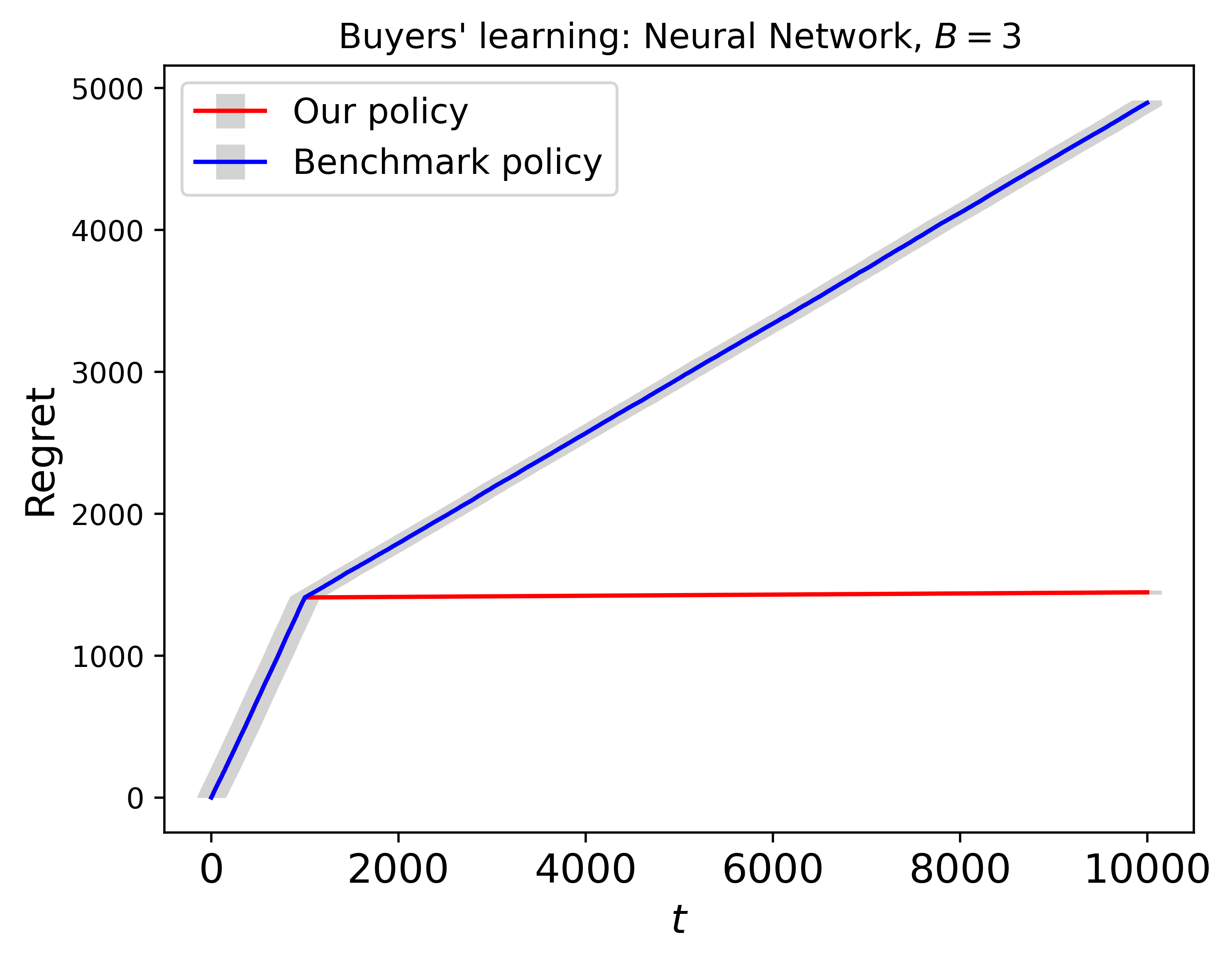}&
        \includegraphics[scale = 0.31]{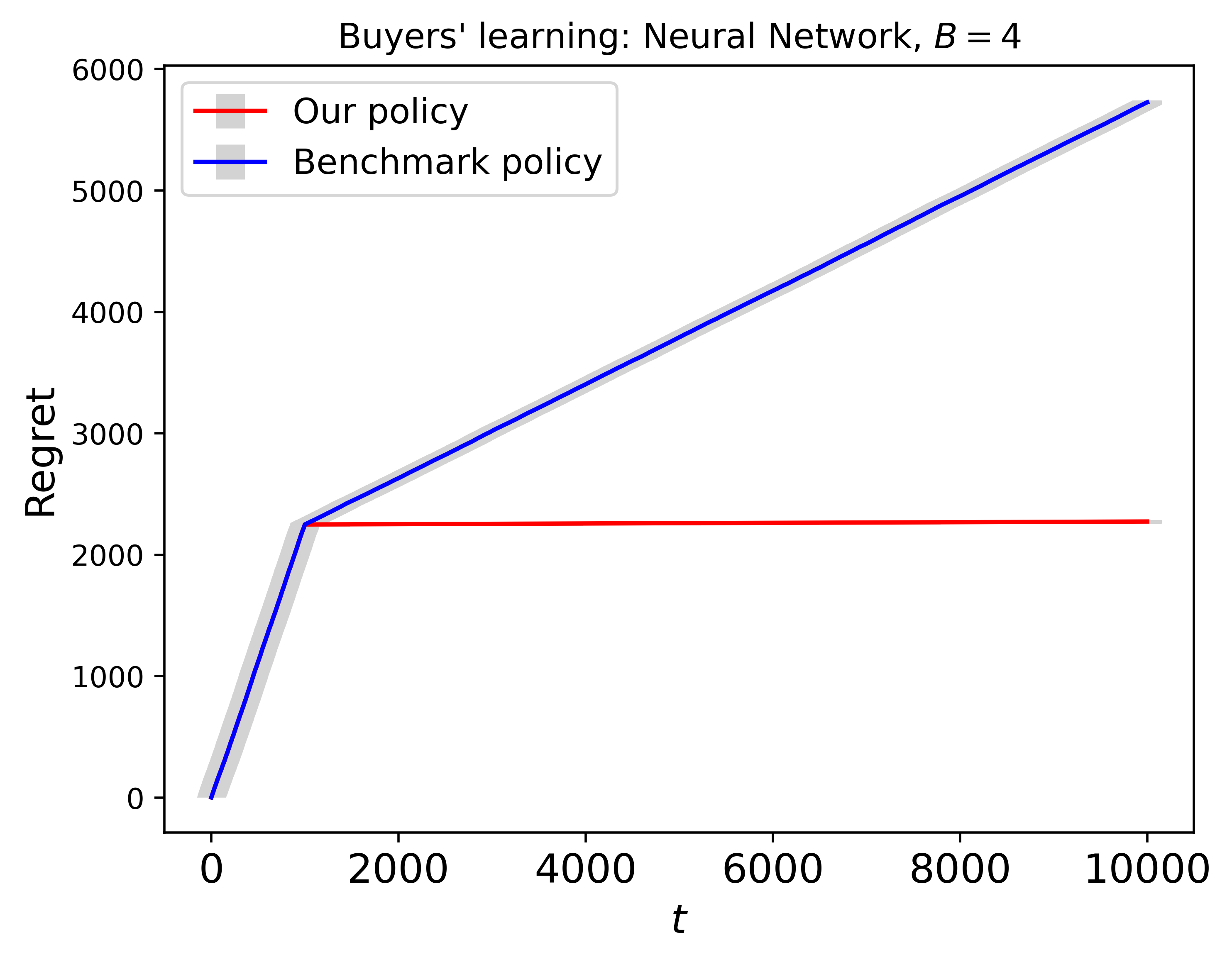}&
        \includegraphics[scale = 0.31]{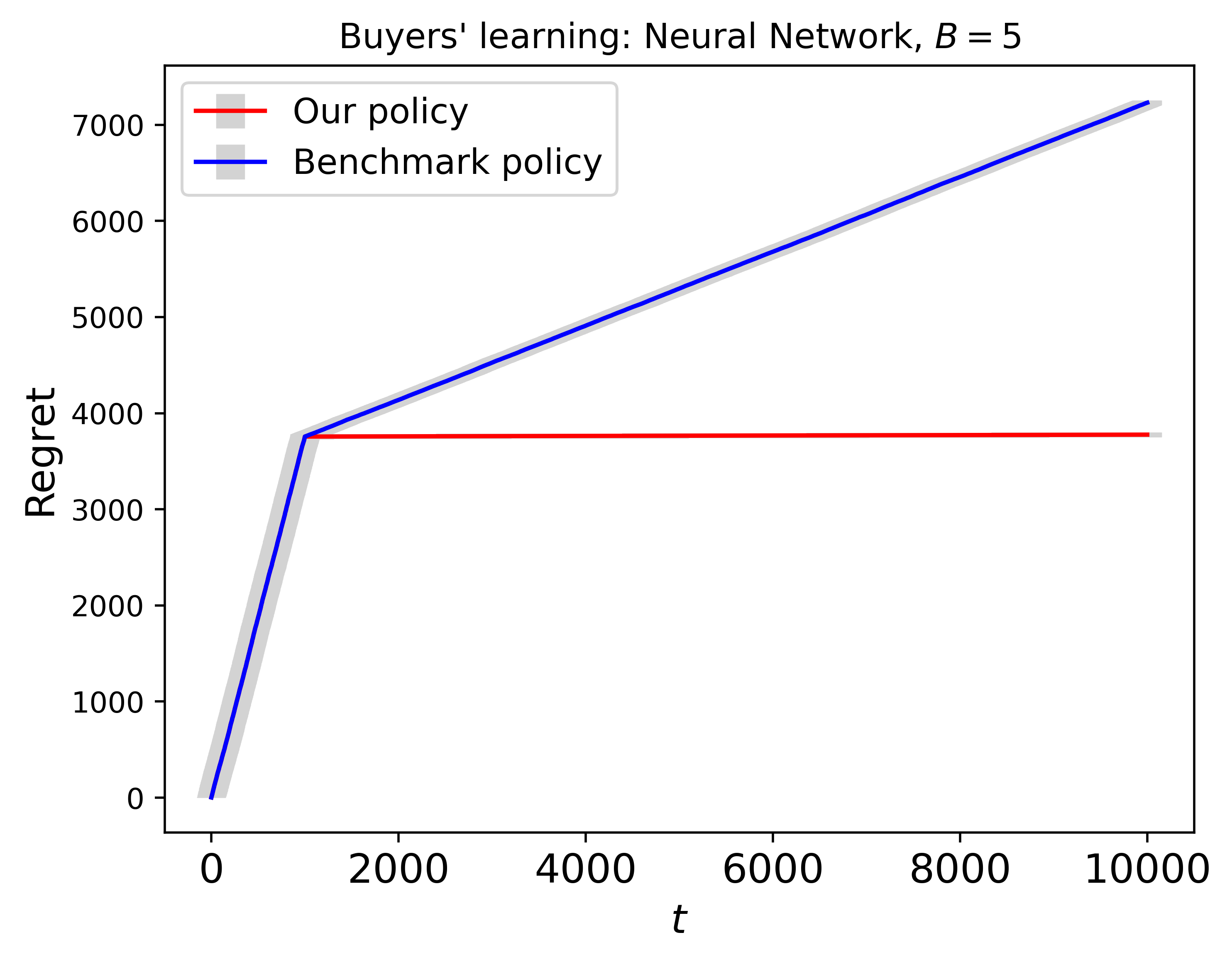}\\
    \end{tabular}
     \caption{ Regret plots for the two policies. The three subplots show the regrets of three different scenarios, $B\in\{3, 4, 5\}$. The remaining caption is the same as Figure \ref{fig6}.}
         \label{fig4}
\end{figure}
Next, we examine the sensitivity of $c_\delta$. For these simulations, we set $B=3$ and $\tau=10$. Figure \ref{fig301} illustrates the regrets of the three policies under three scenarios: $B=3$, $B=4$ and $B=5$.
\begin{figure}[t!]
    \centering
    \begin{tabular}{ccc}  
                \includegraphics[scale = 0.31]{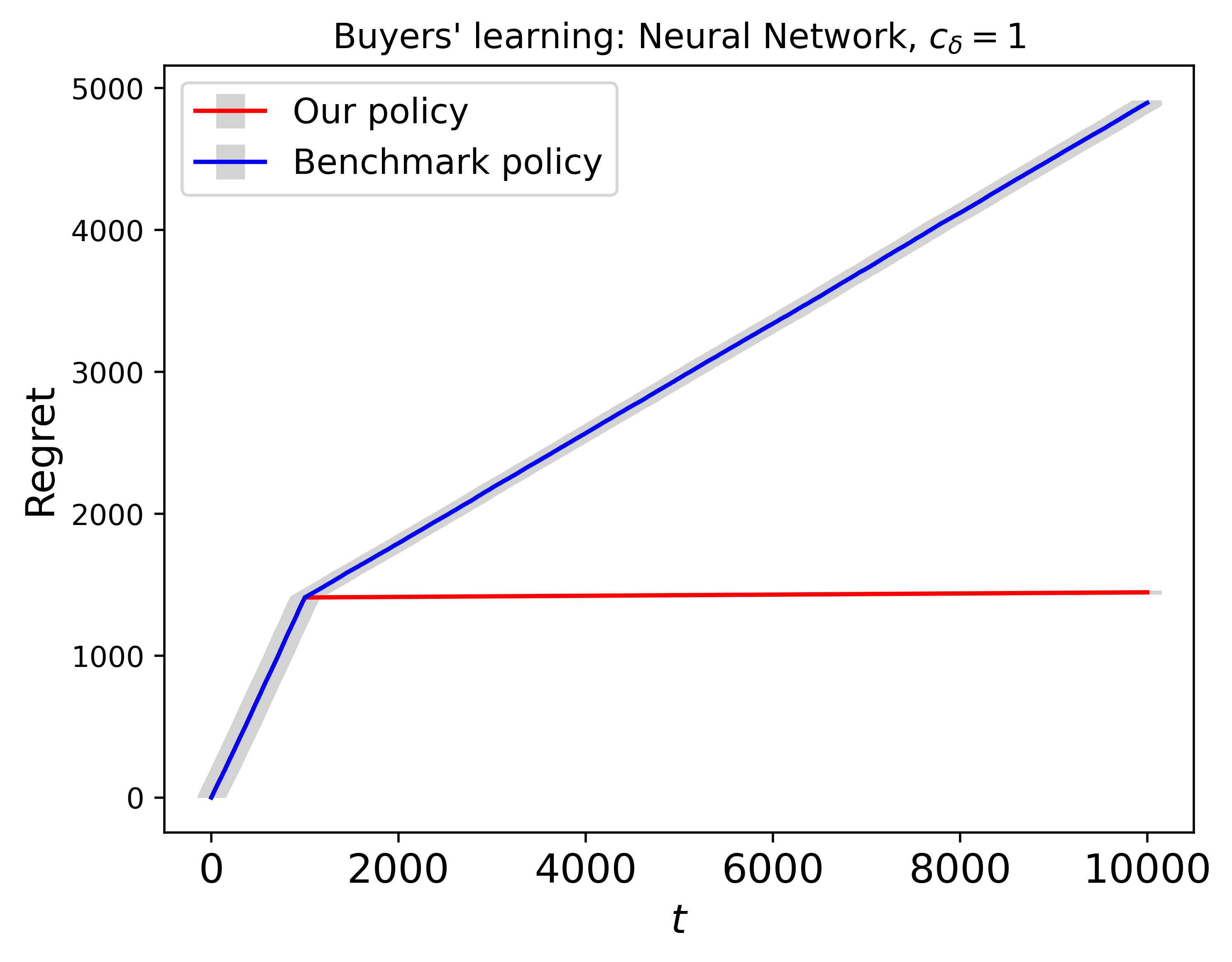}&
        \includegraphics[scale = 0.31]{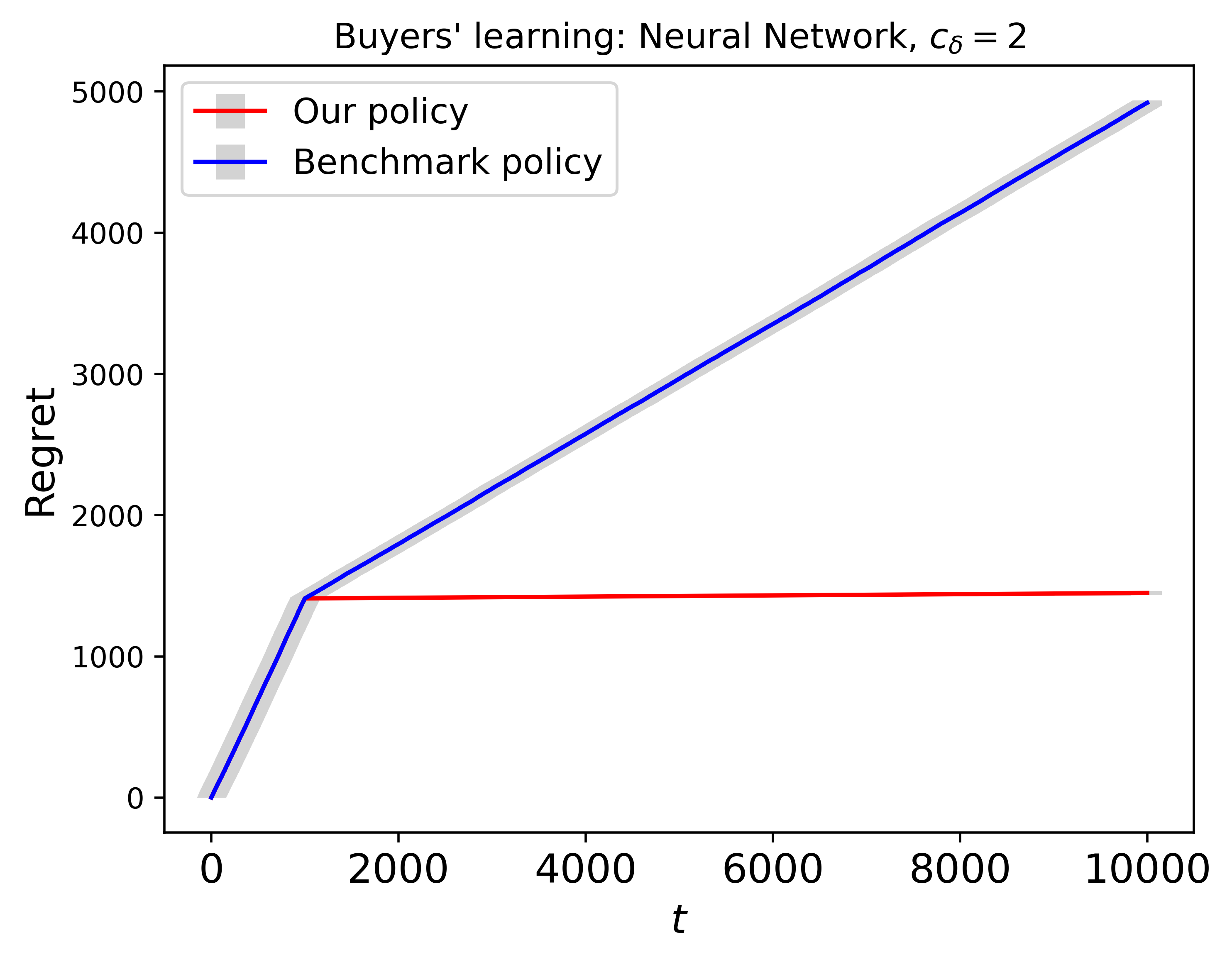}&
        \includegraphics[scale = 0.31]{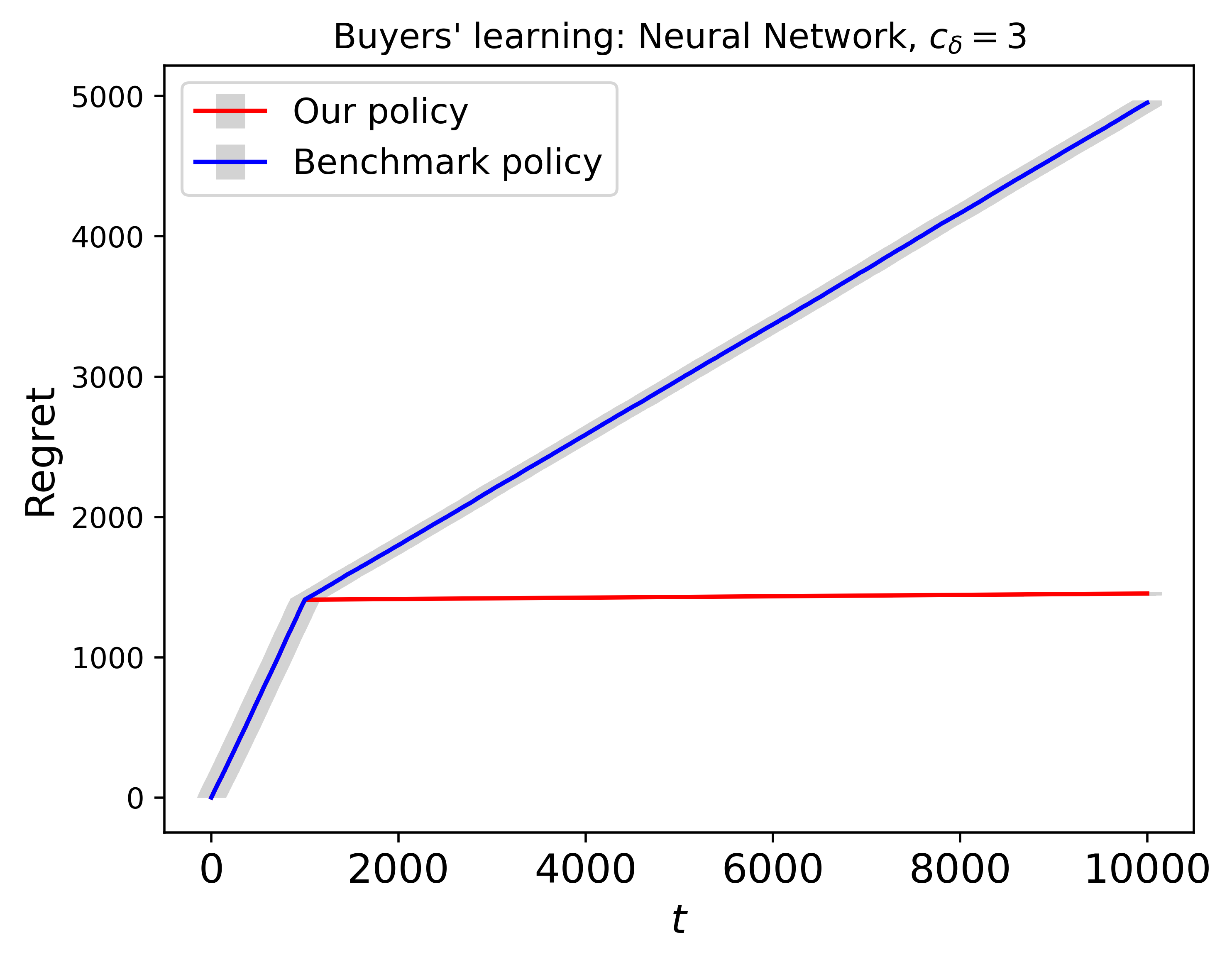}
    \end{tabular}
     \caption{Regret plots for the two policies. The two subplots show the regrets of three different scenarios, $c_\delta\in\{1, 2, 3\}$. The remaining caption is the same as Figure \ref{fig6}.}
         \label{fig301}
\end{figure}
Finally, we assess the sensitivity of $\tau$.  In these simulations, we set $B=3$ and $c_\delta=1$. Figure \ref{fig3} presents the regrets of the three policies for three different scenarios: $\tau=8$, $\tau=10$ and $\tau=12$.

\begin{figure}[t!]
    \centering
    \begin{tabular}{ccc}  
                \includegraphics[scale = 0.31]{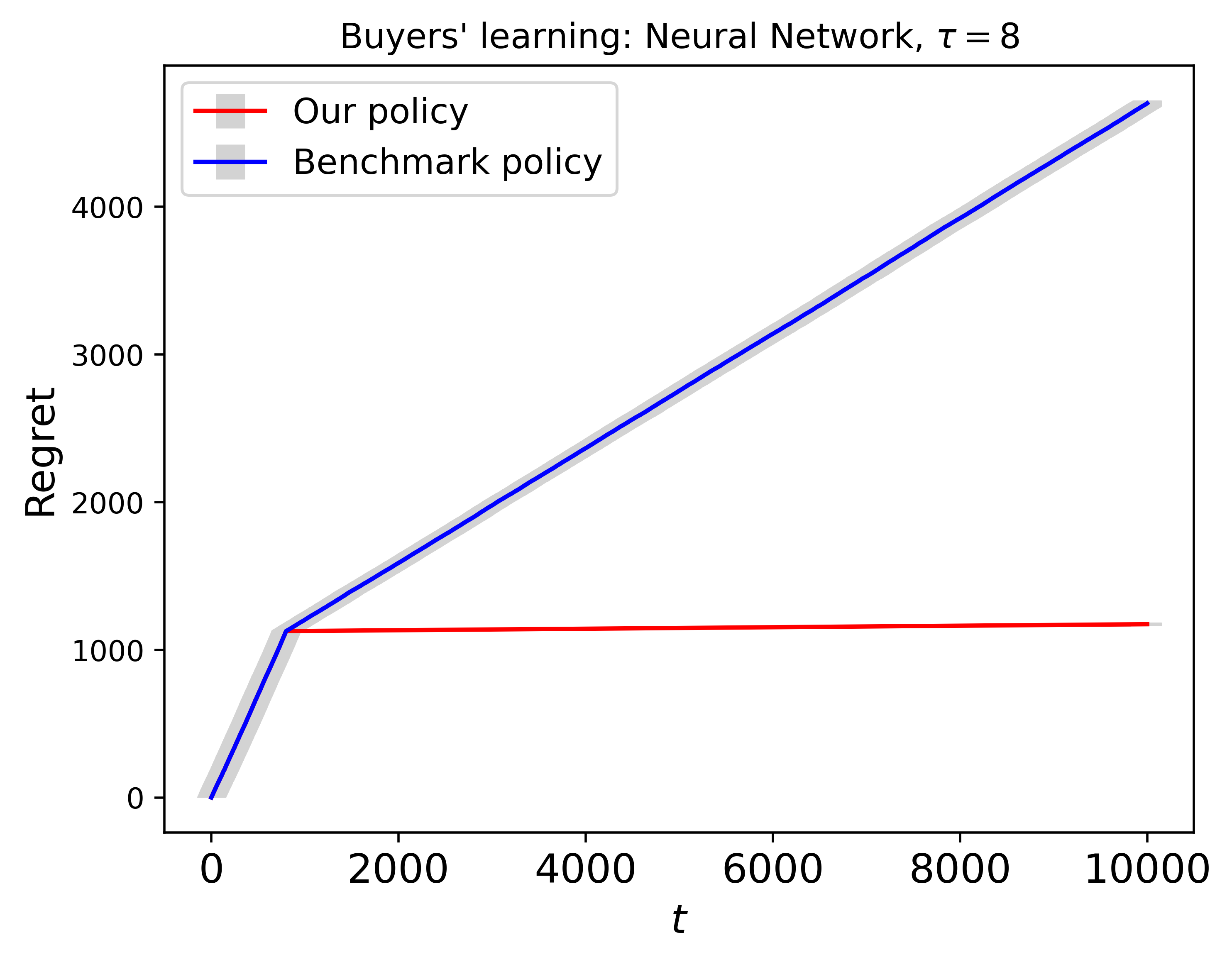}&
        \includegraphics[scale = 0.31]{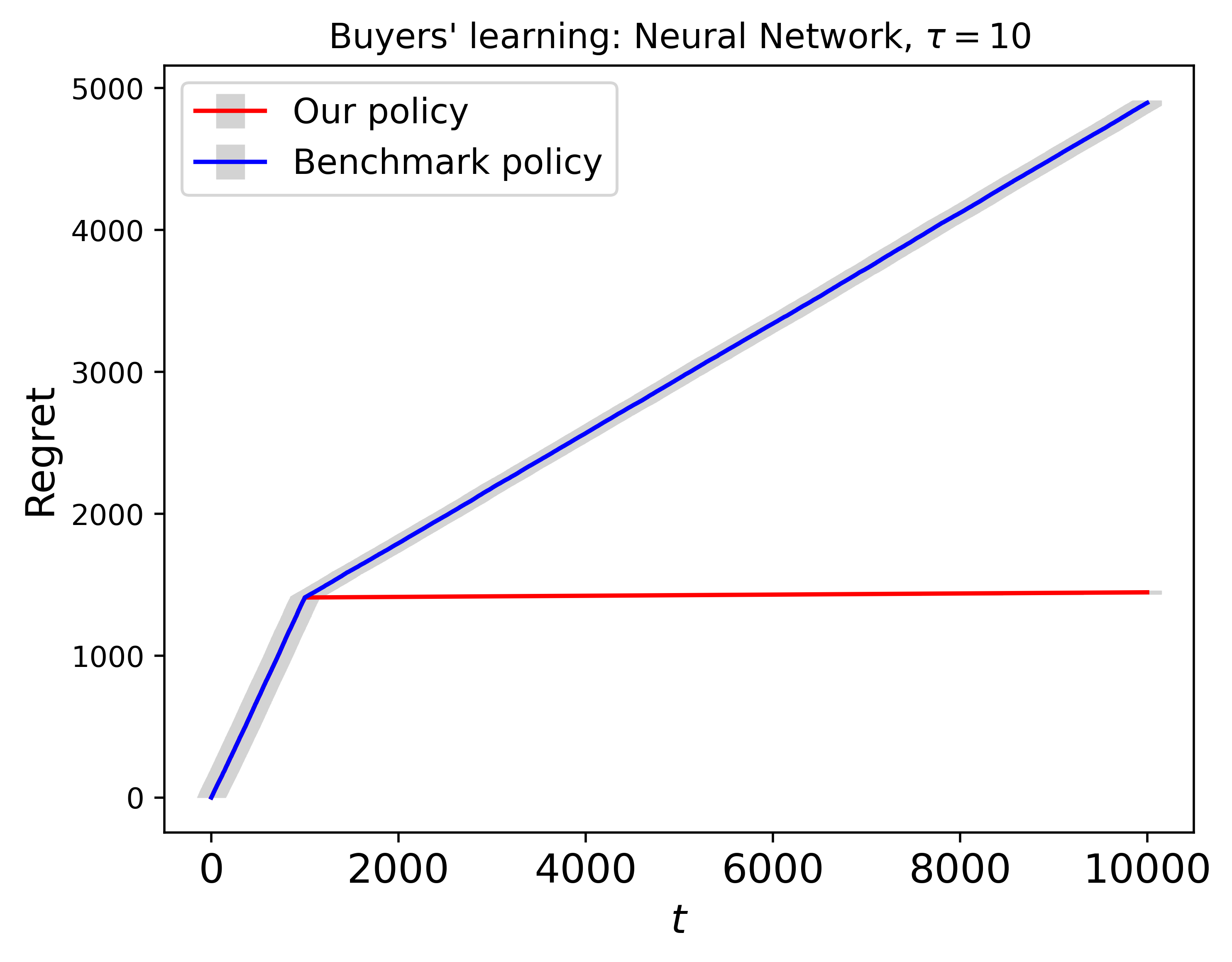}&
        \includegraphics[scale = 0.31]{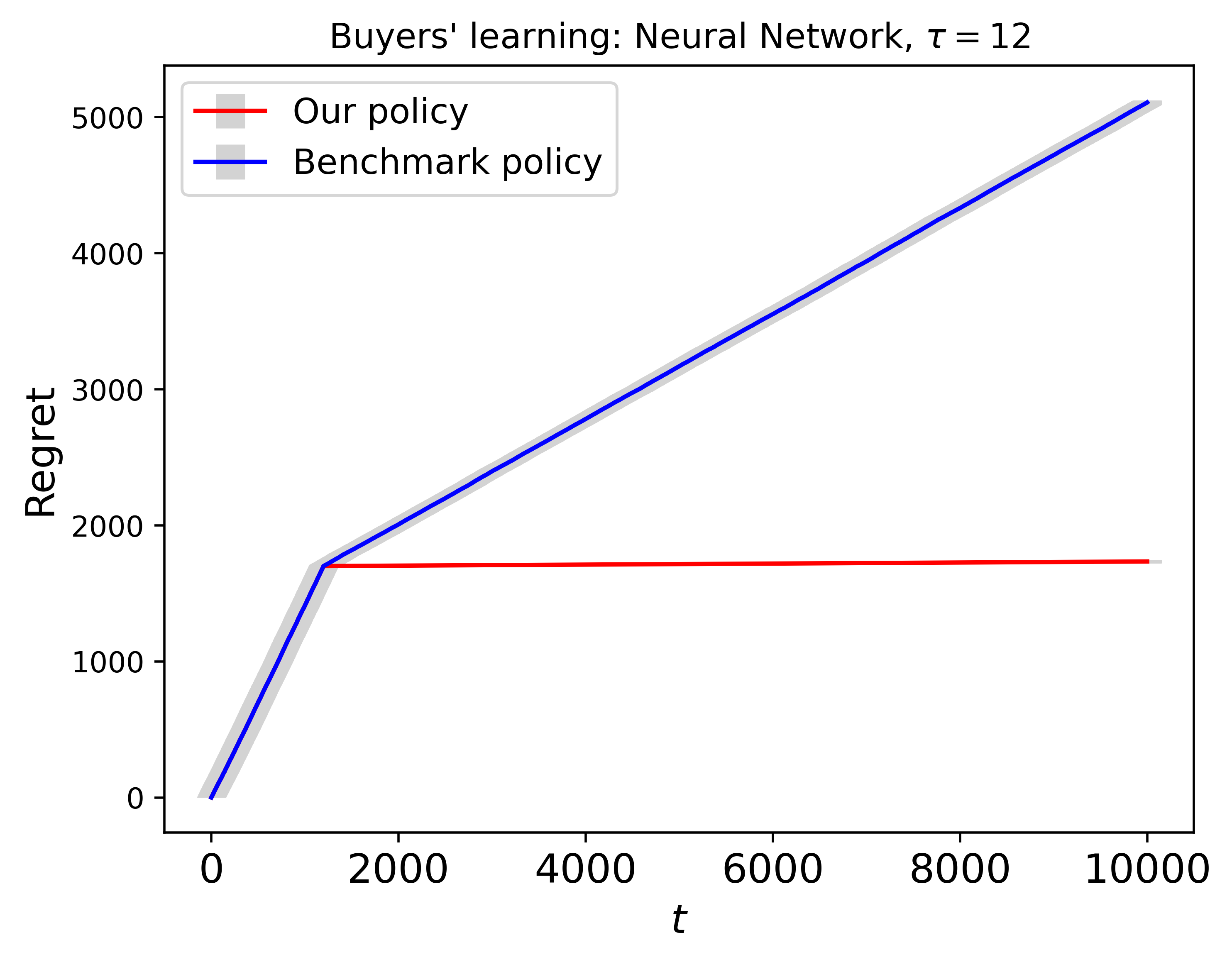}
    \end{tabular}
     \caption{Regret plots for the two policies. The two subplots show the regrets of three different scenarios, $\tau\in\{8, 10, 12\}$. The remaining caption is the same as Figure \ref{fig6}.}
         \label{fig3}
\end{figure}
Overall, our sensitivity analysis indicates that the performance of our policy remains consistent and robust under variations in the hyperparameters $B$, $c_\delta, \tau$, and is always superior over the benchmark policy.

\section{Estimation Errors of Neural Network and Decision Tree}\label{ee}
Figure~\ref{fignn} illustrates the estimation errors of $\delta$ under the neural network and the decision tree, showing that the decision tree yields higher error. These higher errors reduce its effectiveness in discouraging strategic manipulation, which in turn leads to greater regret as shown in Figure 2 of Section 5.1.

    \begin{figure}[t]
    \centering
    \begin{tabular}{ccc}  
                \includegraphics[scale = 0.4]{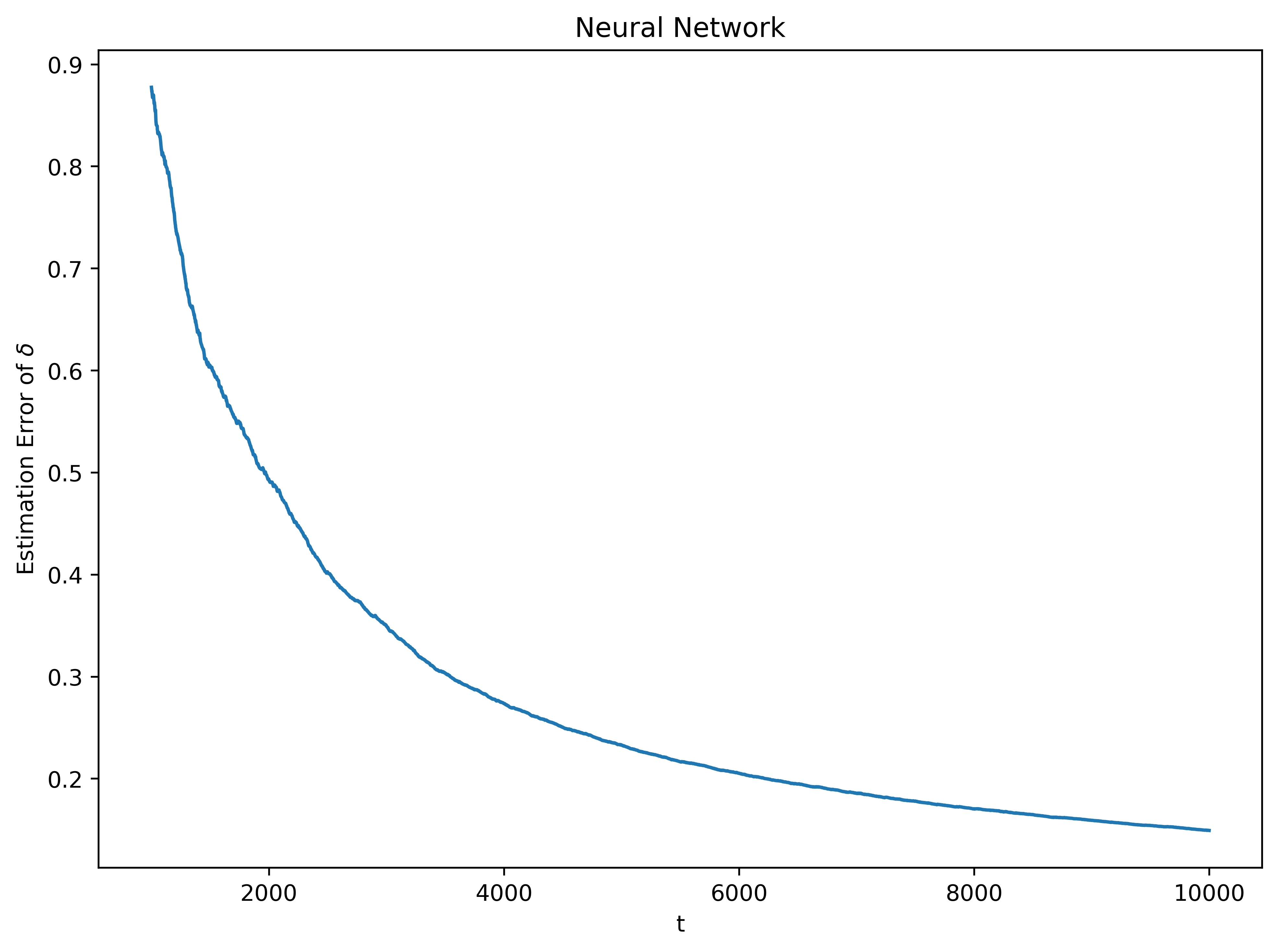}&
        \includegraphics[scale = 0.4]{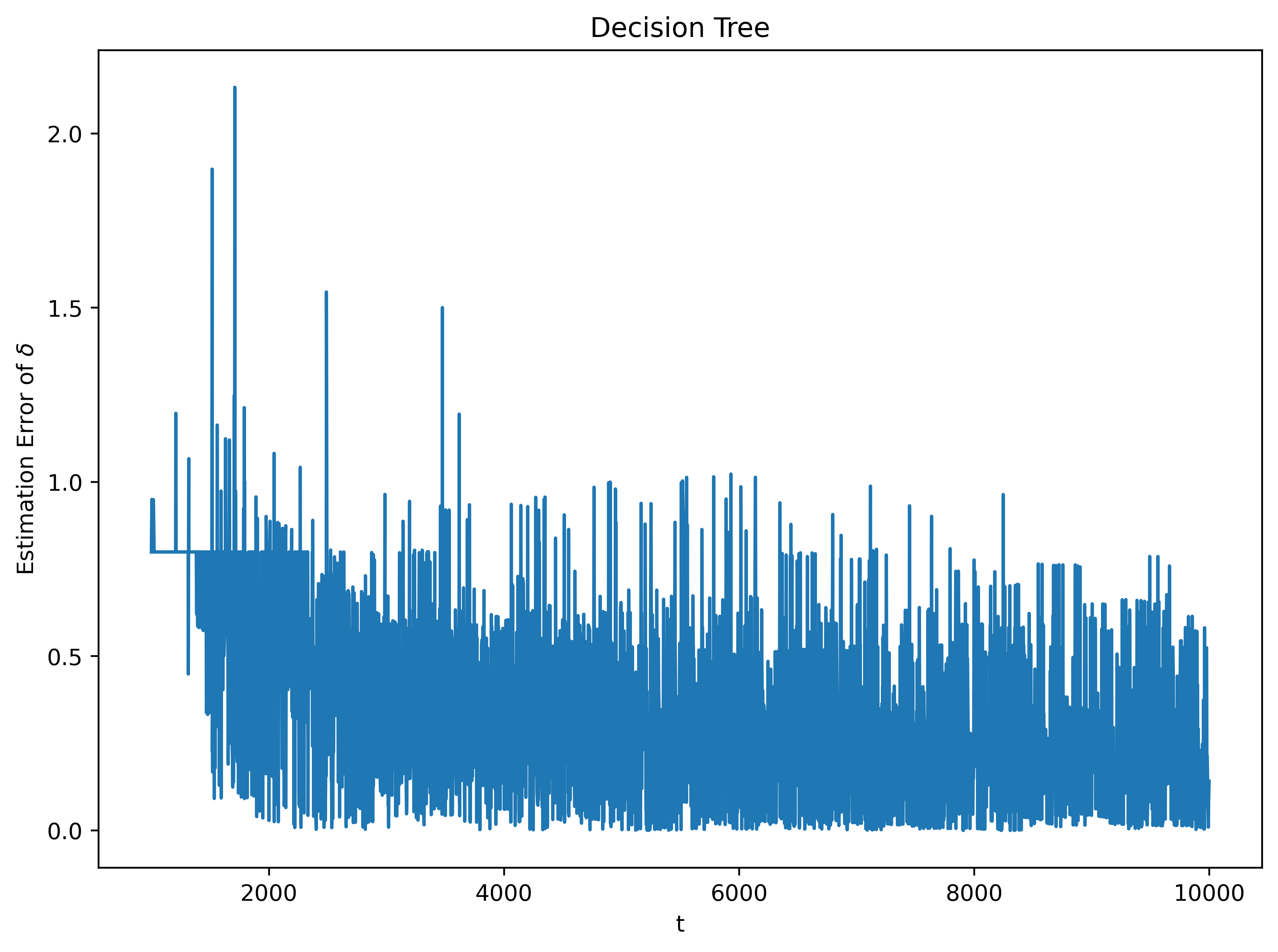}
    \end{tabular}
     \caption{Estimation errors of neural network and decision tree. }
         \label{fignn}
\end{figure}
\section{Additional Details of Real Data Analysis}\label{ara}
We first describe the potential misreporting mechanism in the real application. During the load application process, race can be self-identified at the time of application; however, when the information is missing, lenders are permitted to record race based on visual observation or surname. Our notion of “misreporting” does not assume that the borrower intentionally provides false information at the application stage. Instead, we focus on a potential mismatch that can arise through the appraisal process. In the real application, the borrower already owns the property being appraised and applies for a loan for another purpose. The lender orders an appraisal, and the appraiser conducts an in-person inspection of the property. During this visit, the appraiser typically meets an occupant who provides access to the home and may answer questions about the property. The appraiser does not verify legal ownership or the identity of the borrower. Their task is to assess the physical characteristics and market value of the property. Because the appraiser may rely on visual cues about the person present at the inspection, a mismatch can occur if that person is not the borrower. For example, if the borrower asks a ``friend" of a different race to be present to grant access. In that case, the appraiser’s perception of the “homeowner” can differ from the borrower’s true race. Our model uses this type of appraisal-visit context as a channel through which the perceived race can diverge from the borrower’s true race, thereby generating the misreporting mechanism we analyze.
 
We next use the R package \texttt{MatchIt} \citep{Ho2011} to perform propensity score matching. Propensity scores are estimated via the logistic regression, and we apply 4-to-1 nearest neighbor matching—that is, for each borrower in the non-White group, we match the four closest White borrowers based on propensity scores.

\begin{longtable}{c|cccccc}
\caption{Summary of Covariate Balance Before and After Matching} 
\label{tab2}\\
\hline 
~ & \multicolumn{3}{c}{Before Matching}& \multicolumn{3}{c}{After Matching} \\ \hline
 Covariate & Mean Diff & Var Ratio &  eCDF  & Mean Diff & Var Ratio & eCDF \\ \hline
\endfirsthead

\caption[]{Summary of Covariate Balance Before and After Matching (Continued)} \\
\hline 
~ & \multicolumn{3}{c}{Before Matching}& \multicolumn{3}{c}{After Matching}\\ \hline
 Covariate & Mean Diff & Var Ratio &  eCDF  & Mean Diff & Var Ratio & eCDF \\ \hline
\endhead

\hline
\endfoot
        Income & 0.0364 & 0.9932 & 0.007 & 0.0305 & 0.9988 & 0.0059 \\ 
         Loan term & 0.0444 & 0.8685 & 0.0098 & 0.0379 & 0.8867 & 0.0084 \\ 
        Debt/income & 0.2151 & 0.9925 & 0.0705 & 0.183 & 0.9809 & 0.0597 \\ 
        Loan/value & 0.2303 & 0.8800 & 0.0563 & 0.2006 & 0.8874 & 0.0490 \\ 
        Property value & 0.1328 & 1.0244 & 0.0264 & 0.1083 & 1.0064 & 0.0216 \\
        Discount points & 0.1895 & 1.4118 & 0.0572 & 0.1478 & 1.2029 & 0.0473 \\ 
        Lender credits & 0.1198 & 1.0930 & 0.0380 & 0.0844 & 0.7006 & 0.0335 \\ 
         Applicant age & 0.0178 & 0.9082 & 0.0133 & 0.0142 & 0.9085 & 0.0128 \\ 
        Loan purpose 1 & 0.1634 & - & 0.0732 & 0.1440 & - & 0.0645 \\ 
        Loan purpose 2 & -0.0274 & - & 0.0012 & -0.0248 & - & 0.0011 \\ 
        Loan purpose 3 & -0.1064 & - & 0.0239 & -0.0944 & - & 0.0212 \\ 
        Loan purpose 4 & -0.1155 & - & 0.0480 & -0.1014 & - & 0.0422 \\ 
        State code 1 & -0.0093 & - & 0.0004 & -0.0066 & - & 0.0003 \\ 
        State code 2 & 0.0142 & - & 0.0020 & 0.0201 & - & 0.0028 \\ 
        State code 3 & -0.0514 & - & 0.0033 & -0.0455 & - & 0.0029 \\ 
        State code 4 & -0.0964 & - & 0.0139 & -0.0831 & - & 0.0120 \\ 
        State code 5 & 0.0901 & - & 0.0253 & 0.0827 & - & 0.0233 \\ 
        State code 6 & -0.1402 &-  & 0.0162 & -0.1313 & - & 0.0152 \\ 
        State code 7 & -0.0081 & - & 0.0007 & -0.002 & - & 0.0002 \\ 
        State code 8 & 0.0363 & - & 0.0018 & 0.0254 & - & 0.0013 \\ 
        State code 9 & 0.0250 & - & 0.0019 & 0.0236 & - & 0.0018 \\ 
        State code 10 & -0.0659 & - & 0.0191 & -0.0424 & - & 0.0123 \\ 
        State code 11 & 0.2028 & - & 0.0597 & 0.1521 & - & 0.0448 \\ 
        State code 12 & 0.0128 & - & 0.0002 & 0.0098 & - & 0.0001 \\ 
        State code 13 & 0.0699 & - & 0.0056 & 0.0456 & - & 0.0037 \\ 
        State code 14 & -0.0969 & - & 0.0038 & -0.0918 & - & 0.0036 \\ 
        State code 15 & -0.1714 & - & 0.0059 & -0.1587 & - & 0.0055 \\ 
        State code 16 & 0.0252 & - & 0.0041 & 0.0305 & - & 0.0049 \\ 
        State code 17 & -0.0622 & - & 0.0076 & -0.0507 & - & 0.0062 \\ 
        State code 18 & -0.0583 & - & 0.0031 & -0.0542 & - & 0.0028 \\ 
        State code 19 & -0.0983 & - & 0.0071 & -0.0934 & - & 0.0067 \\ 
        State code 20 & 0.0407 & - & 0.0045 & 0.037 & - & 0.0041 \\ 
        State code 21 & -0.0114 & - & 0.0012 & -0.0044 & - & 0.0005 \\ 
        State code 22 & 0.1446 & - & 0.0317 & 0.1073 & - & 0.0235 \\ 
        State code 23 & -0.1479 & - & 0.0035 & -0.1358 & - & 0.0032 \\ 
        State code 24 & -0.1210 & - & 0.0140 & -0.111 & - & 0.0129 \\ 
        State code 25 & -0.0917 & - & 0.0107 & -0.081 & - & 0.0095 \\ 
        State code 26 & -0.1052 & - & 0.0087 & -0.099 & - & 0.0082 \\ 
        State code 27 & 0.0256 & - & 0.0021 & 0.0248 & - & 0.0020 \\ 
        State code 28 & -0.2259 & - & 0.0049 & -0.205 & - & 0.0045 \\ 
        State code 29 & 0.0721 & - & 0.0160 & 0.0675 & - & 0.0150 \\ 
        State code 30 & -0.0624 & - & 0.0013 & -0.0598 & - & 0.0012 \\ 
        State code 31 & -0.0716 & - & 0.0028 & -0.0684 & - & 0.0027 \\ 
        State code 32 & -0.1522 & - & 0.0039 & -0.1393 & - & 0.0036 \\ 
        State code 33 & 0.0221 & - & 0.0036 & 0.0258 & - & 0.0042 \\ 
        State code 34 & -0.1001 & - & 0.0050 & -0.0957 & - & 0.0048 \\ 
        State code 35 & 0.0412 & - & 0.0058 & 0.0382 & - & 0.0054 \\ 
        State code 36 & 0.0631 & - & 0.0123 & 0.0583 & - & 0.0114 \\ 
        State code 37 & -0.1076 & - & 0.0159 & -0.0933 & - & 0.0137 \\ 
        State code 38 & -0.0311 & - & 0.0031 & -0.0227 & - & 0.0023 \\ 
        State code 39 & -0.1197 & - & 0.0100 & -0.1135 & - & 0.0095 \\ 
        State code 40 & -0.0612 & - & 0.0091 & -0.0470 & - & 0.0070 \\ 
        State code 41 & -0.1108 & - & 0.0025 & -0.1032 & - & 0.0023 \\ 
        State code 42 & -0.0351 & - & 0.0014 & -0.0315 & - & 0.0013 \\ 
        State code 43 & 0.0034 & - & 0.0005 & 0.0117 & - & 0.0017 \\ 
        State code 44 & -0.1022 & - & 0.0021 & -0.0943 & - & 0.0020 \\ 
        State code 45 & -0.0919 & - & 0.0129 & -0.0795 & - & 0.0112 \\ 
        State code 46 & 0.1109 & - & 0.0392 & 0.1077 & - & 0.0381 \\ 
        State code 47 & -0.1815 & - & 0.0120 & -0.1721 & - & 0.0113 \\ 
        State code 48 & 0.0643 & - & 0.0126 & 0.0586 & - & 0.0115 \\ 
        State code 49 & -0.0026 & - & 0 & -0.0023 & - & 0 \\ 
        State code 50 & -0.0905 & - & 0.0009 & -0.0828 & - & 0.0008 \\ 
        State code 51 & -0.0427 & - & 0.0062 & -0.0313 & - & 0.0046 \\ 
        State code 52 & -0.1491 & - & 0.0108 & -0.1432 & - & 0.0104 \\ 
        State code 53 & -0.0751 & - & 0.0026 & -0.0712 & - & 0.0024 \\ 
        State code 54 & -0.1427 & - & 0.0021 & -0.1292 & - & 0.0019 \\ \hline
\end{longtable}
Table~\ref{tab2} summarizes the covariate balance before and after matching. Balance is evaluated using three metrics: standardized mean differences (Mean Diff), variance ratios (Var Ratio), and the mean distance between the two empirical quantile functions (eCDF). The variance ratios for the one-hot encoded variables are not available.
In general, values of Mean Diff and eCDF close to zero, and Var Ratio close to one, indicate good balance. As seen in the “Before Matching” columns, covariates exhibit notable imbalance.
After matching, balance improves substantially for most covariates. 

 The model fits the data well for both groups. The 
$R^2$ value for the non-White group is 0.966, and for the White group, it is 0.960. We provide the estimated regression coefficients for both groups in Tables \ref{tab:nonwhite-regression} and \ref{tab:white-regression}. In the tables, ``Debt/income" represents the debt-to-income ratio, and ``Loan/value" represents the combined loan-to-value ratio. All covariates are statistically significant. The estimated coefficients are also intuitively reasonable. For example, the coefficient on interest rate is negative, implying that higher rates reduce loan demand. The coefficient on applicant age is negative, suggesting that younger individuals are more likely to apply for larger loan amounts.

\begin{longtable}{c|ccccc}
\caption{Regression Estimates for Non-White Group} 
\label{tab:nonwhite-regression}\\
\hline 
Covariate & Coef & Std Err & $t$ & $P>|t|$ & 95\% confidence interval \\ \hline
\endfirsthead
\caption[]{Regression Estimates for Non-White Group (Continued)} \\
\hline 
Covariate & Coef & Std Err & $t$ & $P>|t|$ & 95\% confidence interval \\ \hline
\endhead

\hline
\endfoot
\hline
        \multicolumn{6}{l}{\footnotesize Note: Coef, Std Err and 95\% confidence interval for “State code 49” are scaled by $10^{-15}$.} 
\endlastfoot
        Constant & 24.3732 & 2.096 & 11.628 & 0 & (20.265,  28.481) \\ 
        Interest rate & -0.746 & 0.086 & -8.636 & 0 & (-0.915,  -0.577) \\
        Income & 6.0192 & 0.109 & 55.012 & 0 & (5.805,  6.234) \\ 
        Loan term & -0.5446 & 0.09 & -6.029 & 0 & (-0.722,  -0.368) \\ 
        Debt/income & 1.4636 & 0.091 & 16.055 & 0 & (1.285,  1.642) \\ 
       Loan/value & 62.93 & 0.116 & 543.503 & 0 & (62.703,  63.157) \\ 
        Property value & 132.6207 & 0.122 & 1088.608 & 0 & (132.382,  132.86) \\ 
        Discount points & 1.834 & 0.078 & 23.543 & 0 & (1.681,  1.987) \\ 
        Lender credits & 0.6738 & 0.08 & 8.435 & 0 & (0.517,  0.830) \\ 
        Applicant age & -1.0056 & 0.09 & -11.22 & 0 & (-1.181,  -0.830) \\ 
        Loan purpose 2 & 3.8792 & 1.829 & 2.121 & 0.034 & (0.294,  7.464) \\ 
        Loan purpose 3 & 2.3553 & 0.398 & 5.92 & 0 & (1.576,  3.135) \\ 
        Loan purpose 4 & 7.0774 & 0.245 & 28.882 & 0 & (6.597,  7.558) \\ 
        State code 2 & -9.0071 & 2.128 & -4.232 & 0 & (-13.179,  -4.835) \\ 
        State code 3 & -10.368 & 2.398 & -4.324 & 0 & (-15.068,  -5.668) \\ 
        State code 4 & 2.2939 & 2.123 & 1.08 & 0.28 & (-1.868,  6.456) \\ 
        State code 5 & -2.1351 & 2.073 & -1.03 & 0.303 & (-6.198,  1.927) \\ 
        State code 6 & 5.3019 & 2.165 & 2.448 & 0.014 & (1.058,  9.546) \\ 
        State code 7 & -3.1402 & 2.244 & -1.4 & 0.162 & (-7.538,  1.257) \\ 
        State code 8 & -10.4012 & 2.602 & -3.997 & 0 & (-15.502,  -5.300) \\ 
        State code 9 & -2.5258 & 2.295 & -1.1 & 0.271 & (-7.025,  1.973) \\ 
        State code 10 & -0.7204 & 2.068 & -0.348 & 0.728 & (-4.773,  3.332) \\ 
        State code 11 & -0.9517 & 2.067 & -0.46 & 0.645 & (-5.003,  3.100) \\ 
        State code 12 & 5.5142 & 6.021 & 0.916 & 0.36 & (-6.288,  17.316) \\ 
        State code 13 & -11.5147 & 2.282 & -5.045 & 0 & (-15.988,  -7.041) \\ 
        State code 14 & -8.3882 & 2.883 & -2.91 & 0.004 & (-14.038,  -2.738) \\ 
        State code 15 & -3.2503 & 3.09 & -1.052 & 0.293 & (-9.307,  2.807) \\ 
        State code 16 & -7.8519 & 2.108 & -3.725 & 0 & (-11.984,  -3.720) \\ 
        State code 17 & -3.6947 & 2.152 & -1.717 & 0.086 & (-7.913,  0.523) \\ 
        State code 18 & -6.73 & 2.555 & -2.634 & 0.008 & (-11.737,  -1.723) \\ 
        State code 19 & -5.7833 & 2.33 & -2.482 & 0.013 & (-10.350,  -1.216) \\ 
        State code 20 & -13.7401 & 2.175 & -6.317 & 0 & (-18.003,  -9.477) \\ 
        State code 21 & -0.5792 & 2.182 & -0.265 & 0.791 & (-4.857,  3.698) \\ 
        State code 22 & 0.376 & 2.082 & 0.181 & 0.857 & (-3.704,  4.456) \\ 
        State code 23 & 4.1571 & 3.951 & 1.052 & 0.293 & (-3.587,  11.901) \\ 
        State code 24 & -7.9109 & 2.162 & -3.659 & 0 & (-12.149,  -3.673) \\ 
        State code 25 & 0.1378 & 2.16 & 0.064 & 0.949 & (-4.096,  4.371) \\ 
        State code 26 & -7.0645 & 2.264 & -3.12 & 0.002 & (-11.502,  -2.627) \\ 
        State code 27 & -9.3169 & 2.269 & -4.107 & 0 & (-13.764,  -4.870) \\ 
        State code 28 & -2.5528 & 4.202 & -0.607 & 0.544 & (-10.790, 5.684) \\ 
        State code 29 & -0.2927 & 2.081 & -0.141 & 0.888 & (-4.371, 3.786) \\ 
        State code 30 & -3.8272 & 4.416 & -0.867 & 0.386 & (-12.482, 4.827) \\ 
        State code 31 & -2.9708 & 2.883 & -1.031 & 0.303 & (-8.621, 2.679) \\ 
        State code 32 & 3.8654 & 3.718 & 1.04 & 0.298 & (-3.421, 11.152) \\ 
        State code 33 & -0.046 & 2.109 & -0.022 & 0.983 & (-4.179, 4.087) \\ 
        State code 34 & -2.2117 & 2.605 & -0.849 & 0.396 & (-7.317, 2.893) \\ 
        State code 35 & 1.8927 & 2.127 & 0.89 & 0.374 & (-2.276, 6.061) \\ 
        State code 36 & -5.4352 & 2.092 & -2.598 & 0.009 & (-9.536, -1.334) \\ 
        State code 37 & -6.2738 & 2.12 & -2.959 & 0.003 & (-10.429, -2.118) \\ 
        State code 38 & -10.3223 & 2.2 & -4.691 & 0 & (-14.635, -6.010) \\ 
        State code 39 & 2.2401 & 2.263 & 0.99 & 0.322 & (-2.195, 6.675) \\ 
        State code 40 & -5.3325 & 2.118 & -2.517 & 0.012 & (-9.485, -1.181) \\ 
        State code 41 & -19.0629 & 4.142 & -4.603 & 0 & (-27.181, -10.945) \\ 
        State code 42 & 2.7974 & 2.845 & 0.983 & 0.325 & (-2.779, 8.373) \\ 
        State code 43 & -4.1728 & 2.118 & -1.97 & 0.049 & (-8.325, -0.021) \\ 
        State code 44 & -3.5824 & 4.339 & -0.826 & 0.409 & (-12.087, 4.922) \\ 
        State code 45 & 2.4697 & 2.127 & 1.161 & 0.246 & (-1.698, 6.638) \\ 
        State code 46 & -0.2248 & 2.062 & -0.109 & 0.913 & (-4.267, 3.817) \\ 
        State code 47 & 5.6354 & 2.384 & 2.364 & 0.018 & (0.963, 10.308) \\ 
        State code 48 & -0.8347 & 2.089 & -0.399 & 0.69 & (-4.930, 3.261) \\ 
        State code 49 & -8.10 & 1.96 & -4.135 & 0 & (-11.900, -4.260) \\ 
        State code 50 & -3.0631 & 8.085 & -0.379 & 0.705 & (-18.910, 12.784) \\ 
        State code 51 & 4.6413 & 2.124 & 2.185 & 0.029 & (0.478, 8.804) \\ 
        State code 52 & -10.5734 & 2.326 & -4.545 & 0 & (-15.133, -6.014) \\ 
        State code 53 & -5.8529 & 3.118 & -1.877 & 0.061 & (-11.964, 0.258) \\ 
        State code 54 & -6.2122 & 5.784 & -1.074 & 0.283 & (-17.549, 5.125) \\ 
\end{longtable}

\begin{longtable}{c|ccccc}
\caption{Regression Estimates for White Group} 
\label{tab:white-regression}\\
\hline 
Covariate & Coef & Std Err & $t$ & $P>|t|$ & 95\% confidence interval \\ \hline
\endfirsthead

\caption[]{Regression Estimates for White Group (Continued)} \\
\hline 
\textbf{Covariate} & Coef & Std Err & $t$ & $P>|t|$ & 95\% confidence interval \\ \hline
\endhead

\hline
\endfoot
        Constant & 24.4108 & 1.03 & 23.698 & 0 & (22.392, 26.430) \\ 
        Interest rate & -0.8865 & 0.047 & -18.859 & 0 & (-0.979, -0.794) \\ 
        Income & 5.9661 & 0.06 & 99.847 & 0 & (5.849, 6.083) \\ 
        Loan term & -0.9515 & 0.046 & -20.504 & 0 & (-1.043, -0.861) \\ 
        Debt/income & 1.7647 & 0.05 & 35.603 & 0 & (1.668, 1.862) \\ 
        Loan/value & 58.7266 & 0.06 & 985.327 & 0 & (58.610, 58.843) \\ 
        Property value & 126.8843 & 0.067 & 1897.579 & 0 & (126.753, 127.015) \\ 
        Discount points & 2.5833 & 0.05 & 51.225 & 0 & (2.484, 2.682) \\ 
        Lender credits & 0.7625 & 0.045 & 16.9 & 0 & (0.674, 0.851) \\ 
        Applicant age & -0.8327 & 0.046 & -18.005 & 0 & (-0.923, -0.742) \\ 
        Loan purpose 2 & 2.7612 & 0.784 & 3.522 & 0 & (1.225, 4.298) \\ 
        Loan purpose 3 & 0.5968 & 0.184 & 3.236 & 0.001 & (0.235, 0.958) \\ 
        Loan purpose 4 & 4.0067 & 0.121 & 33.198 & 0 & (3.770, 4.243) \\ 
        State code 2 & -4.2735 & 1.056 & -4.046 & 0 & (-6.344, -2.203) \\ 
        State code 3 & -6.4937 & 1.124 & -5.776 & 0 & (-8.697, -4.290) \\ 
        State code 4 & 1.5555 & 1.031 & 1.509 & 0.131 & (-0.465, 3.576) \\ 
        State code 5 & -0.2199 & 1.021 & -0.215 & 0.829 & (-2.220, 1.781) \\ 
        State code 6 & 4.527 & 1.036 & 4.368 & 0 & (2.496, 6.558) \\ 
        State code 7 & -2.4666 & 1.111 & -2.221 & 0.026 & (-4.643, -0.290) \\ 
        State code 8 & 2.5675 & 1.943 & 1.321 & 0.186 & (-1.241, 6.376) \\ 
        State code 9 & -2.0066 & 1.214 & -1.653 & 0.098 & (-4.385, 0.372) \\ 
        State code 10 & -0.3086 & 1.013 & -0.305 & 0.761 & (-2.294, 1.676) \\ 
        State code 11 & -1.3342 & 1.03 & -1.295 & 0.195 & (-3.353, 0.685) \\ 
        State code 12 & 13.1661 & 10.148 & 1.297 & 0.194 & (-6.723, 33.055) \\ 
        State code 13 & -2.2878 & 1.778 & -1.287 & 0.198 & (-5.773, 1.197) \\ 
        State code 14 & -6.9136 & 1.167 & -5.924 & 0 & (-9.201, -4.626) \\ 
        State code 15 & -1.2245 & 1.13 & -1.084 & 0.278 & (-3.439, 0.990) \\ 
        State code 16 & -5.4549 & 1.045 & -5.221 & 0 & (-7.503, -3.407) \\ 
        State code 17 & -4.641 & 1.046 & -4.438 & 0 & (-6.690, -2.592) \\ 
        State code 18 & -7.3294 & 1.156 & -6.342 & 0 & (-9.595, -5.064) \\ 
        State code 19 & -5.6874 & 1.079 & -5.273 & 0 & (-7.801, -3.573) \\ 
        State code 20 & -8.3807 & 1.12 & -7.484 & 0 & (-10.576, -6.186) \\ 
        State code 21 & -1.5276 & 1.076 & -1.42 & 0.156 & (-3.636, 0.581) \\ 
        State code 22 & 0.7522 & 1.053 & 0.715 & 0.475 & (-1.311, 2.816) \\ 
        State code 23 & -2.4104 & 1.215 & -1.984 & 0.047 & (-4.792, -0.029) \\ 
        State code 24 & -5.3938 & 1.038 & -5.195 & 0 & (-7.429, -3.359) \\ 
        State code 25 & -0.1304 & 1.042 & -0.125 & 0.9 & (-2.172, 1.911) \\ 
        State code 26 & -4.4909 & 1.063 & -4.224 & 0 & (-6.574, -2.407) \\ 
        State code 27 & -7.8866 & 1.189 & -6.633 & 0 & (-10.217, -5.556) \\ 
        State code 28 & -2.7618 & 1.167 & -2.367 & 0.018 & (-5.048, -0.475) \\ 
        State code 29 & -1.2667 & 1.03 & -1.23 & 0.219 & (-3.286, 0.753) \\ 
        State code 30 & -7.084 & 1.459 & -4.855 & 0 & (-9.944, -4.224) \\ 
        State code 31 & -3.3698 & 1.201 & -2.806 & 0.005 & (-5.724, -1.016) \\ 
        State code 32 & 0.4483 & 1.193 & 0.376 & 0.707 & (-1.890,  2.786) \\ 
        State code 33 & -0.0129 & 1.044 & -0.012 & 0.99 & (-2.060, 2.034) \\ 
        State code 34 & -3.5158 & 1.124 & -3.128 & 0.002 & (-5.719, -1.313) \\ 
        State code 35 & 2.5048 & 1.068 & 2.346 & 0.019 & (0.412, 4.597) \\ 
        State code 36 & -3.6208 & 1.038 & -3.488 & 0 & (-5.656, -1.586) \\ 
        State code 37 & -7.2239 & 1.029 & -7.019 & 0 & (-9.241, -5.207) \\ 
        State code 38 & -7.4766 & 1.074 & -6.963 & 0 & (-9.581, -5.372) \\ 
        State code 39 & 1.2463 & 1.058 & 1.177 & 0.239 & (-0.828, 3.321) \\ 
        State code 40 & -5.2714 & 1.034 & -5.1 & 0 & (-7.297, -3.246) \\ 
        State code 41 & -17.3413 & 1.286 & -13.483 & 0 & (-19.862, -14.82) \\ 
        State code 42 & 0.6704 & 1.276 & 0.525 & 0.599 & (-1.831, 3.171) \\ 
        State code 43 & -1.1719 & 1.046 & -1.12 & 0.263 & (-3.222, 0.878) \\ 
        State code 44 & -1.4613 & 1.322 & -1.106 & 0.269 & (-4.052, 1.129) \\ 
        State code 45 & -0.1416 & 1.032 & -0.137 & 0.891 & (-2.165, 1.882) \\ 
        State code 46 & -3.5493 & 1.013 & -3.503 & 0 & (-5.535, -1.564) \\ 
        State code 47 & 2.8721 & 1.062 & 2.705 & 0.007 & (0.791, 4.953) \\ 
        State code 48 & 0.8482 & 1.038 & 0.817 & 0.414 & (-1.187, 2.883) \\ 
        State code 49 & 1.3431 & 18.918 & 0.071 & 0.943 & (-35.735, 38.421) \\ 
        State code 50 & -2.0441 & 1.689 & -1.21 & 0.226 & (-5.355, 1.267) \\ 
        State code 51 & 1.3177 & 1.038 & 1.269 & 0.204 & (-0.717, 3.352) \\ 
        State code 52 & -4.926 & 1.061 & -4.641 & 0 & (-7.006, -2.846) \\ 
        State code 53 & -9.3132 & 1.233 & -7.554 & 0 & (-11.730, -6.897) \\ 
        State code 54 & -2.0485 & 1.351 & -1.517 & 0.129 & (-4.696, 0.599) \\
\end{longtable}

To assess robustness, we conduct a sensitivity analysis using a milder trimming rule that removes only the upper and lower 1\% of selected covariates. The resulting regret curves (Figure~\ref{fig70}) are qualitatively similar to those obtained under 5\% trimming across all three scenarios ($\delta\in{0.05,0.1,0.2}$), indicating that our main conclusions do not depend on the specific trimming threshold.
\begin{figure}[t!]
    \centering
    \begin{tabular}{ccc}  
                \includegraphics[scale = 0.31]{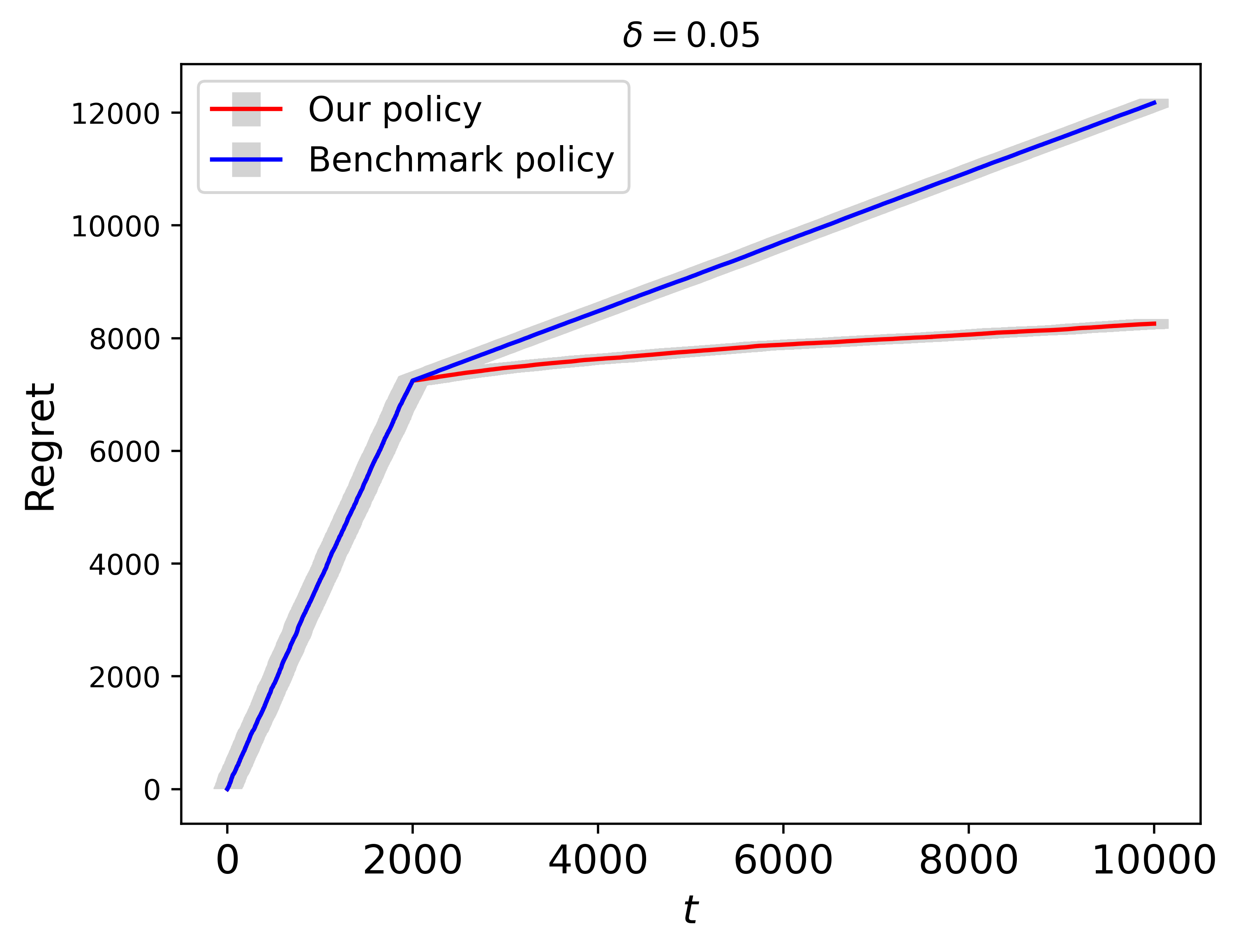}&
        \includegraphics[scale = 0.31]{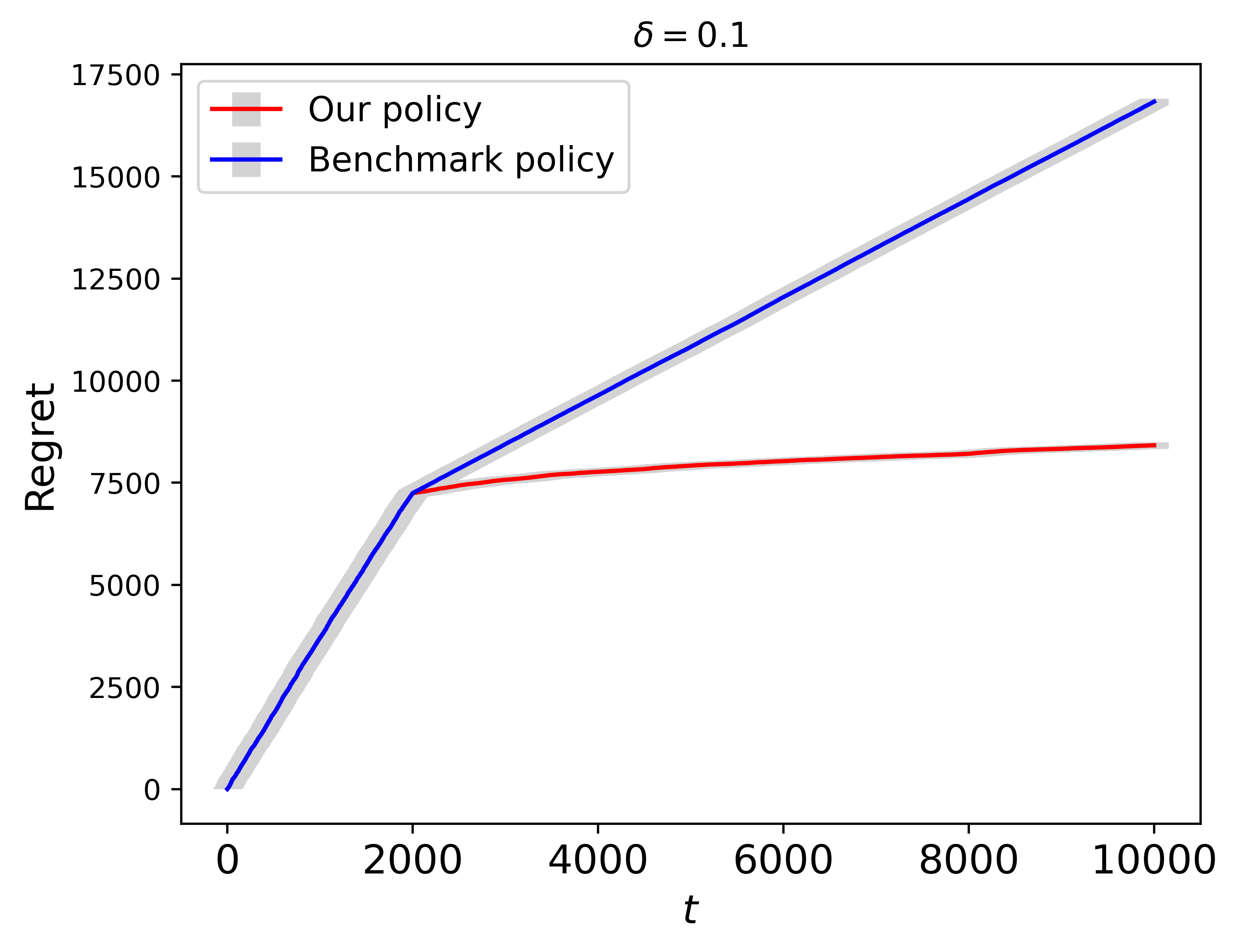}&
        \includegraphics[scale = 0.31]{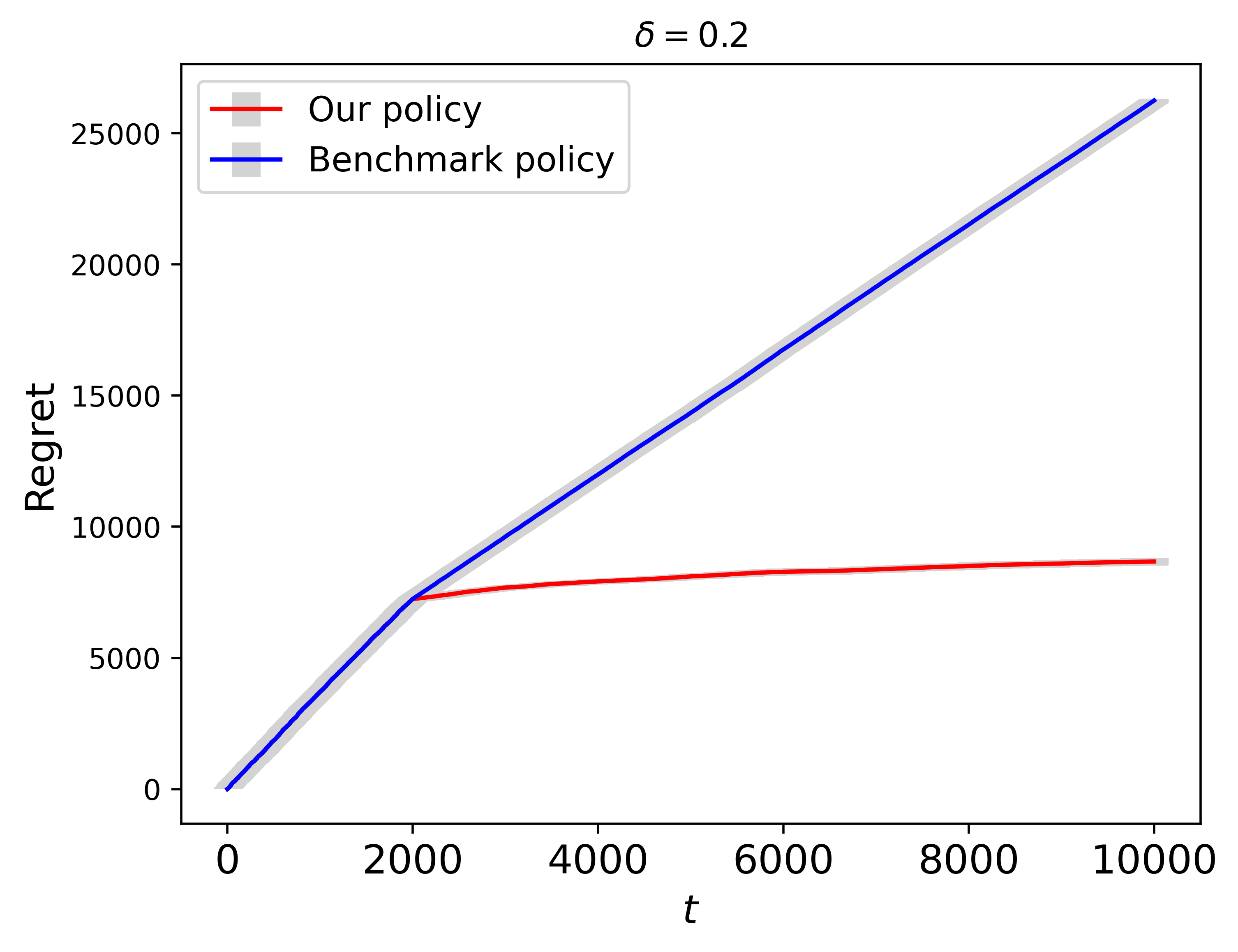}
    \end{tabular}
     \caption{Regret plots for the two policies (1\% tails removed). The three subplots show the regrets of three different scenarios, $\delta\in\{0.05, 0.1, 0.2\}$.} 
          \label{fig70}
\end{figure}

\section{Misspecified Demand Models}\label{md}
In this section, we conducted additional simulation experiments to evaluate the performance of our proposed pricing policy under misspecified demand models. The experimental setup follows the same framework as in Section \ref{sec5}, with modifications to the true underlying demand structure.
We consider the following two cases of model misspecification:
    \begin{equation*}
    \text{Case 1}:
 y_{j}=\alpha_jp+\boldsymbol{\beta}_{j0} +\boldsymbol{\beta}_{j1} \boldsymbol{x}_{1}^2+\boldsymbol{\beta}_{j2}\boldsymbol{x}_{2}+\boldsymbol{\beta}_{j3}\boldsymbol{x}_{3}+\epsilon.
\end{equation*}
and 
    \begin{equation*}
   \text{Case 2}: y_{j}=\alpha_jp+\boldsymbol{\beta}_{j0} +\boldsymbol{\beta}_{j1} \boldsymbol{x}_{1}^2+\boldsymbol{\beta}_{j2}\boldsymbol{x}_{1}\boldsymbol{x}_{2}+\boldsymbol{\beta}_{j3}\boldsymbol{x}_{3}+\epsilon.
\end{equation*}
In both settings, the seller incorrectly assumes a linear specification and learns the model using $\hat{y}_{j}=\hat{\alpha}_jp+\hat{\boldsymbol{\beta}}_{j0}+\hat{\boldsymbol{\beta}}_{j1} \boldsymbol{x}_{1}+\hat{\boldsymbol{\beta}}_{j2}\boldsymbol{x}_{2}+\hat{\boldsymbol{\beta}}_{j3}\boldsymbol{x}_{3}$. Under the misspecification, we conduct simulations to compare the performance of our proposed policy with the benchmark policy. The results are presented in Figure~\ref{fig301miss} and  Figure~\ref{fig301miss1}. As expected, the regret under the misspecified setting is higher than in the correctly specified case. The regret in Case 2 is higher than that in Case 1 because the degree of model misspecification is more severe in Case 2. This increased mismatch between the true demand model and the model used by the policy leads to less accurate parameter estimation and higher cumulative regret.  However, our policy still significantly outperforms the benchmark in all scenarios.
    \begin{figure}[t!]
    \centering
    \begin{tabular}{ccc}  
                \includegraphics[scale = 0.31]{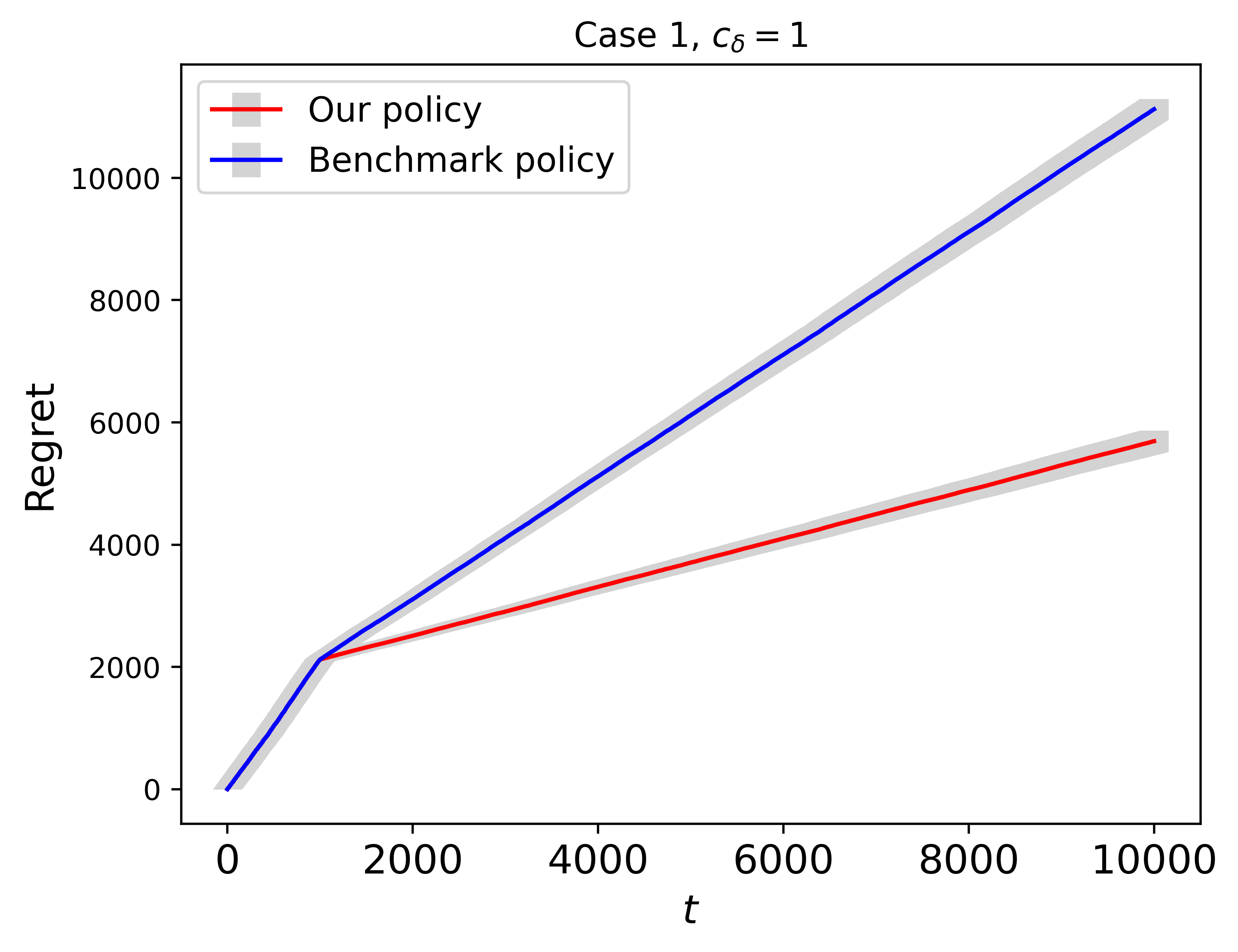}&
        \includegraphics[scale = 0.31]{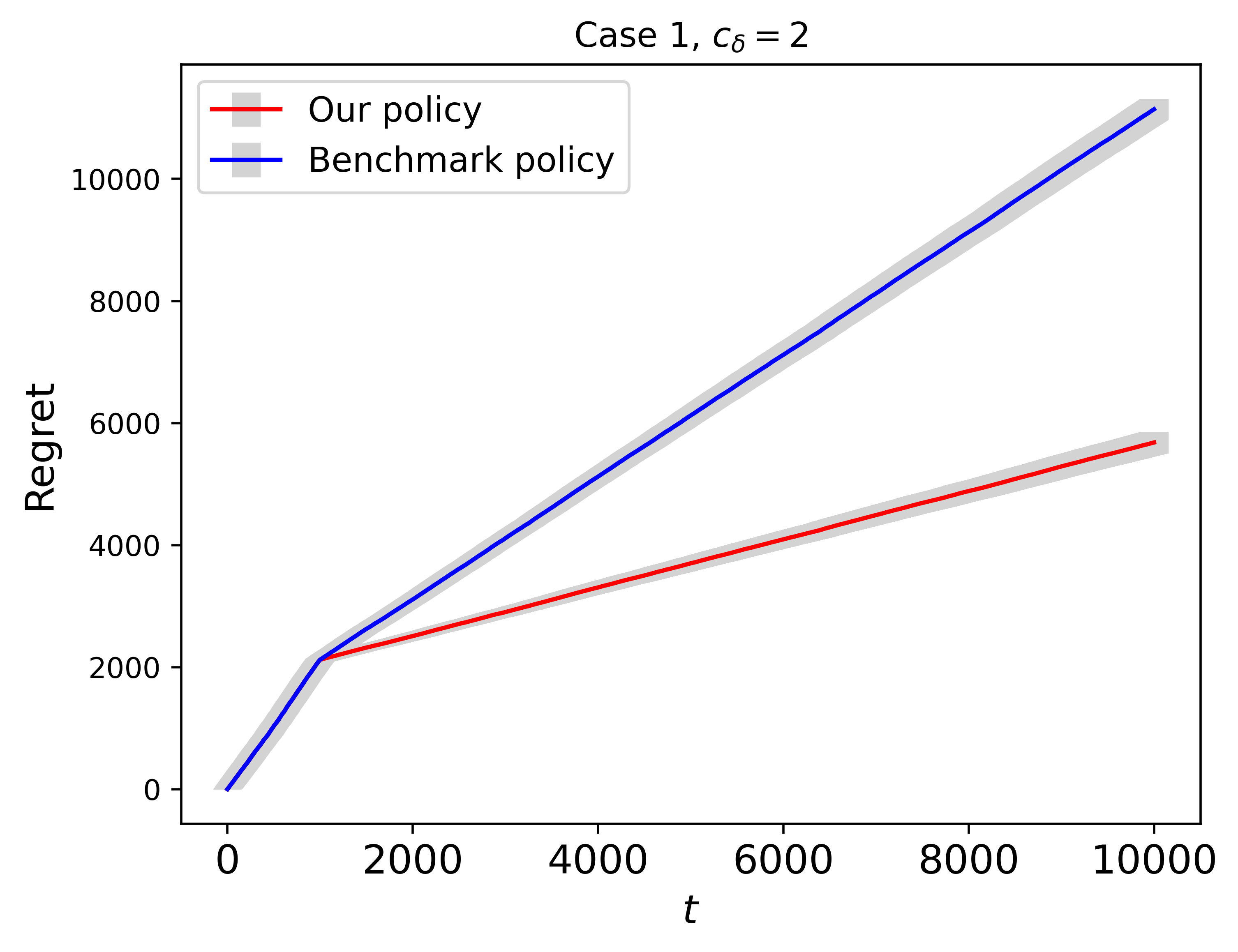}&
        \includegraphics[scale = 0.31]{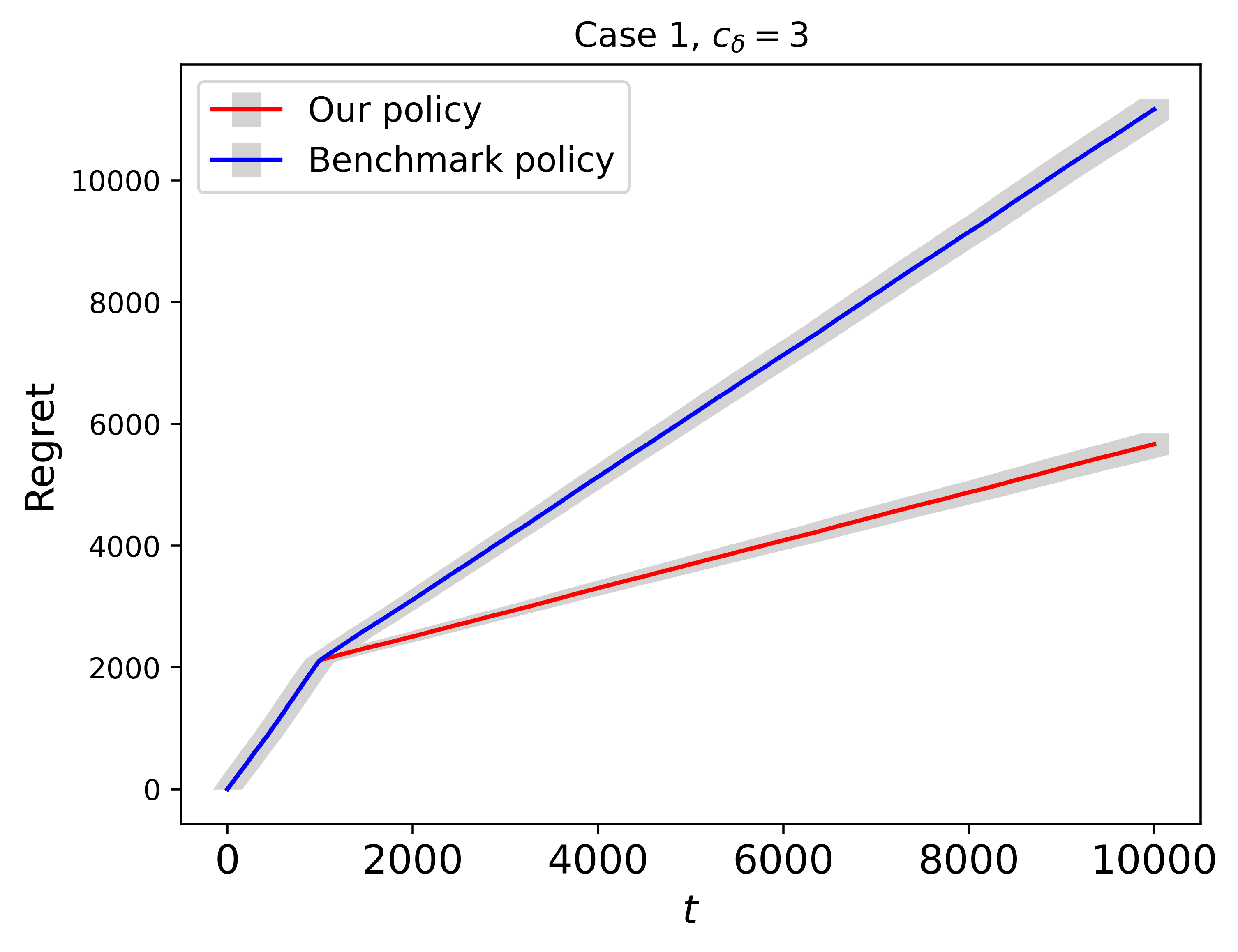}
    \end{tabular}
     \caption{Regret plots for the two policies under the misspecified demand model. The three subplots show the regrets of three different scenarios, $c_\delta\in\{1, 2, 3\}$. }
         \label{fig301miss}
\end{figure}
    \begin{figure}[t!]
    \centering
    \begin{tabular}{ccc}  
                \includegraphics[scale = 0.31]{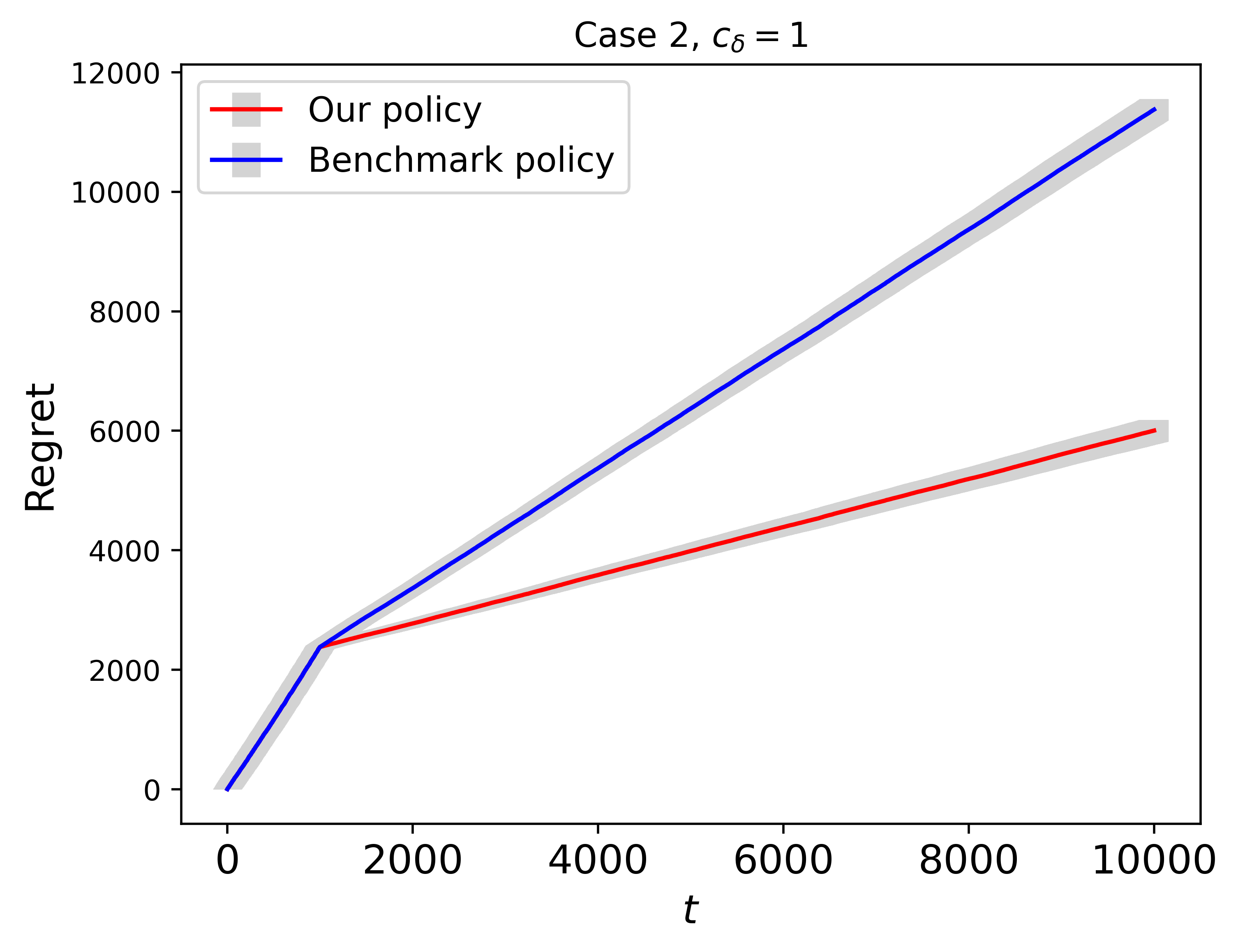}&
        \includegraphics[scale = 0.31]{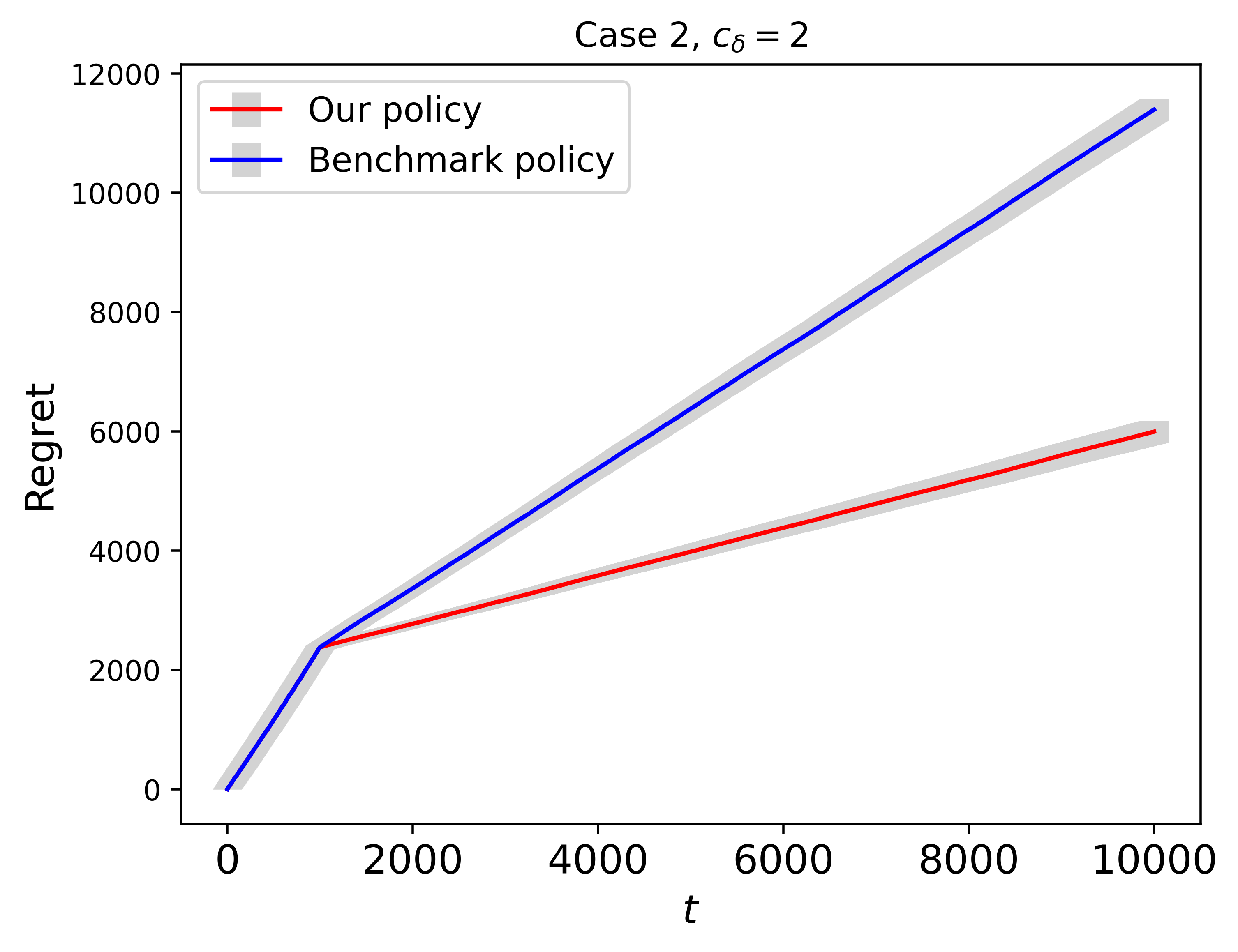}&
        \includegraphics[scale = 0.31]{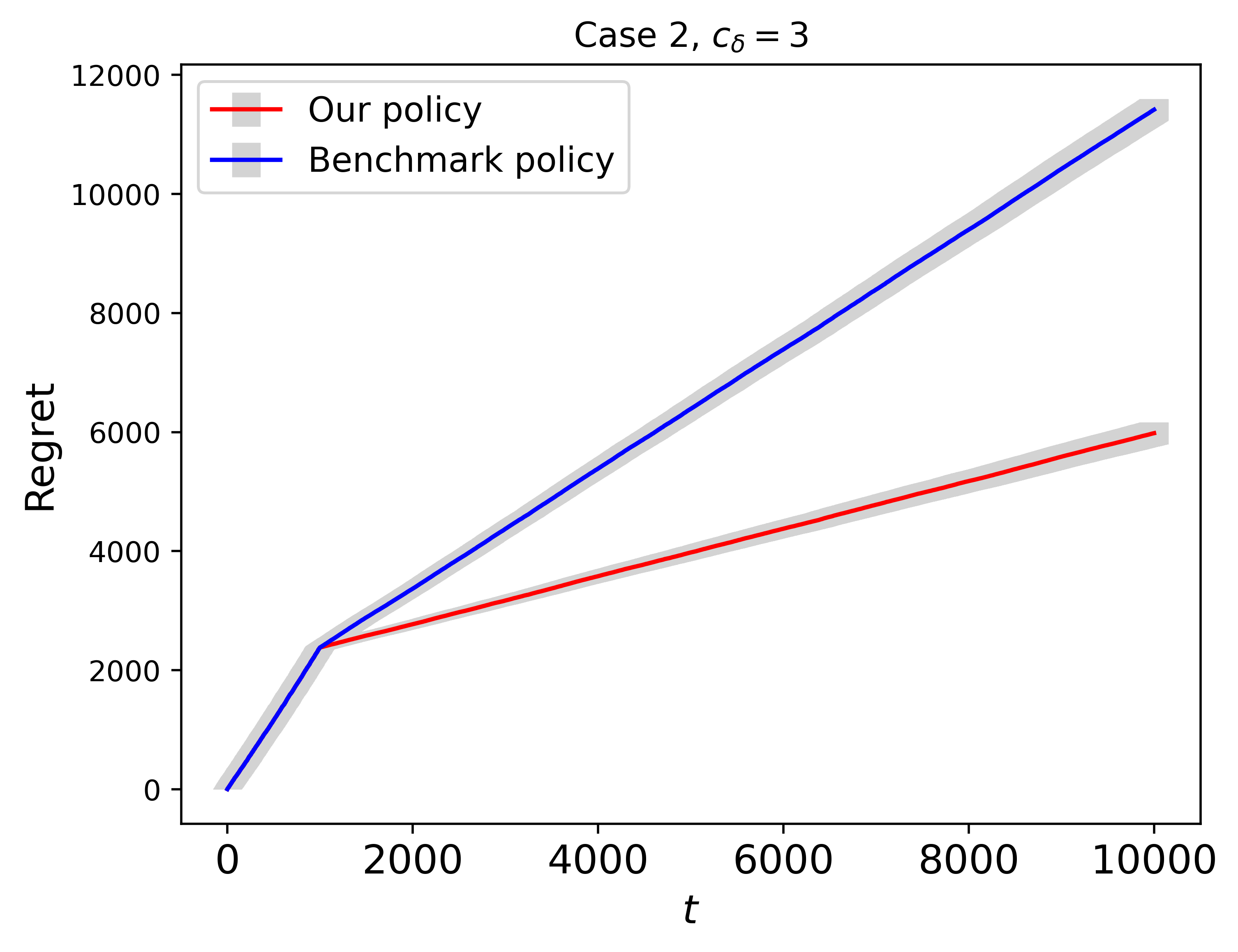}
    \end{tabular}
     \caption{Regret plots for the two policies under the misspecified demand model. The three subplots show the regrets of three different scenarios, $c_\delta\in\{1, 2, 3\}$. }
         \label{fig301miss1}
\end{figure}
\section{Extension to Multi-group Settings}\label{multi}
We outline how the proposed framework can be extended to a multi-group setting. Assume there are $N\geq 2$ groups.  At time $t$, the demand of a buyer with feature $\boldsymbol{x}_t$ and group status $G_t=j\in\{1,\cdots,N\}$ is $y_{jt}=\alpha_jp_{t}+\boldsymbol{\beta}_j^{\top} \tilde{\boldsymbol{x}}_{t}+\epsilon_t$,
where $\alpha_j$ and $\boldsymbol{\beta}_j$ are group-specific unknown parameters, $\tilde{\boldsymbol{x}}_{t}=(1, \boldsymbol{x}_t^\top)^\top$, and $\epsilon_t$ is an $i.i.d.$ sub-Gaussian variable with $\mathbb{E}(\epsilon_t|\boldsymbol{x}_t,p_{t})=0$.  Let $q_j$ denote the proportion of buyers from group $j$. At each time $t$, the fairness-aware clairvoyant seller maximizes the weighted revenue by solving the following constrained optimization problem,
\begin{equation}\label{pricer}
\begin{aligned}
\mathop{\max}_{p_1,\cdots,p_N}&\ \sum_{j=1}^Nq_jR_j(p_{j},\boldsymbol{x}_t)\\
&s.t. \ \lvert p_{i}-p_{j}\rvert \leq\delta, \forall i\neq j.
\end{aligned}    
\end{equation}
where $R_j(p, \boldsymbol{x}_t)=p(\alpha_jp+\boldsymbol{\beta}_j^{\top}\tilde{\boldsymbol{x}}_t)$ is the expected revenue. The fairness constraint used here is similar to that proposed in \cite{Cohen2021}. We denote $\boldsymbol{\theta}_j=(\alpha_j, \boldsymbol{\beta}_j^\top)^\top$ and $\boldsymbol{\theta}=(\boldsymbol{\theta}_1^\top,\cdots,\boldsymbol{\theta}_N^\top)$.
Solving the optimization problem \eqref{pricer} yields the optimal price for group $j$: $p^*_{j}(\boldsymbol{\theta},\boldsymbol{x}_t,\delta)$ for $j=1,\cdots, N$. As in the two-group case, the pricing policy consists of an exploration phase and an exploitation phase. During exploration, the seller uses a uniform pricing policy to collect data for estimating the demand parameters $\boldsymbol{\theta}$,  resulting in estimates $\hat{\boldsymbol{\theta}}$. In the exploitation phase, the seller applies the fairness-aware pricing policy: $p_{j}(\hat{\boldsymbol{\theta}},\boldsymbol{x}_t,\delta_j)$ for $j=1,\cdots, N$, where $\delta_j$ reflects the fairness constraint associated with group $j$. The pricing policy $p_{j}(\hat{\boldsymbol{\theta}},\boldsymbol{x}_t,\delta_j)$ can be viewed as a plug-in estimator of $p^*_{j}(\boldsymbol{\theta},\boldsymbol{x}_t,\delta)$, where $\hat{\boldsymbol{\theta}}_j$ is obtained from the exploration phase. The fairness level $\delta_j$ should typically be chosen close to the global fairness parameter $\delta$, but its exact value must be carefully selected to effectively deter strategic misreporting. Our strategy to select the fairness level in the algorithm can be extended to this multi-group setting.

\section{One-time Fixed Cost}\label{ot}
If we treat the manipulation cost $C_0$ as a one-time fixed amount, the model can be adjusted accordingly. Following our original model, we consider that the group 0 is discriminated. For a buyer from group 0, the expected benefit of manipulation is 
$$p_0(\alpha_0p_0+\boldsymbol{\beta}_0^{\top}\tilde{\boldsymbol{x}}_t)-p_1(\alpha_0p_1+\boldsymbol{\beta}_0^{\top}\tilde{\boldsymbol{x}}_t)=\alpha_0(p_0^2-p_1^2)+(p_0-p_1)\boldsymbol{\beta}_0^{\top}\tilde{\boldsymbol{x}}_t.$$
To deter misreporting, the seller must ensure this gain does not exceed the fixed manipulation cost $C_0$. Accordingly, the fairness-constrained pricing problem becomes
\begin{equation}\label{pricecost}
\begin{aligned}
\mathop{\min}_{p_0,p_1}&\ -qR_0(p_{0},\boldsymbol{x}_t)-(1-q)R_1(p_{1},\boldsymbol{x}_t)\\
&s.t. \ \alpha_0(p_0^2-p_1^2)+(p_0-p_1)\boldsymbol{\beta}_0^{\top}\tilde{\boldsymbol{x}}_t\leq \delta,
\end{aligned}    
\end{equation}
where $R_j(p, \boldsymbol{x}_t)=p(\alpha_jp+\boldsymbol{\beta}_j^{\top}\tilde{\boldsymbol{x}}_t)$ is the expected revenue and $\delta<C_0$. We denote $\boldsymbol{\theta}_j=(\alpha_j, \boldsymbol{\beta}_j^\top)^\top$ and $\boldsymbol{\theta}=(\boldsymbol{\theta}_0^\top,\boldsymbol{\theta}_1^\top)$.
Solving the optimization problem \eqref{pricecost} yields the optimal price for group $j$: $p^*_{j}(\boldsymbol{\theta},\boldsymbol{x}_t,\delta)$ for $j=1,2$. The pricing policy consists of an exploration phase and an exploitation phase. During exploration, the seller uses a uniform pricing policy to collect data for estimating the demand parameters $\boldsymbol{\theta}$,  resulting in estimates $\hat{\boldsymbol{\theta}}$. In the exploitation phase, the seller applies the fairness-aware pricing policy: $p_{j}(\hat{\boldsymbol{\theta}},\boldsymbol{x}_t,\delta')$ for $j=0, 1$, where $\delta'$ reflects the fairness constraint used in the algorithm. The pricing policy $p_{j}(\hat{\boldsymbol{\theta}},\boldsymbol{x}_t,\delta')$ can be viewed as a plug-in estimator of $p^*_{j}(\boldsymbol{\theta},\boldsymbol{x}_t,\delta)$, where $\hat{\boldsymbol{\theta}}$ is obtained from the exploration phase. The fairness level $\delta'$ used in the algorithm is typically chosen close to the global fairness parameter $\delta$.

We conduct new experiments to evaluate the performance of the proposed pricing policy under the one-time manipulation cost setting. We set $C_0=0.8, \delta=0.799$, and the length of the exploration phase as $T_0=\lceil 10\sqrt{T}\rceil$, with a total time horizon of $T=10000$. Following Algorithm \ref{alg1}, we set $\delta'=\delta-c_\delta \sqrt{\frac{\log T_0}{T_0}}$, and vary the adjustment parameter $c_\delta=1,2,3$. Figure~\ref{fig301cost} presents the regret plots for our policy compared to the benchmark. In all cases, our proposed policy achieves substantially lower regret than the benchmark, demonstrating its effectiveness even when the cost of strategic manipulation is one-time fixed.
\begin{figure}[t!]
    \centering
    \begin{tabular}{ccc}  
                \includegraphics[scale = 0.31]{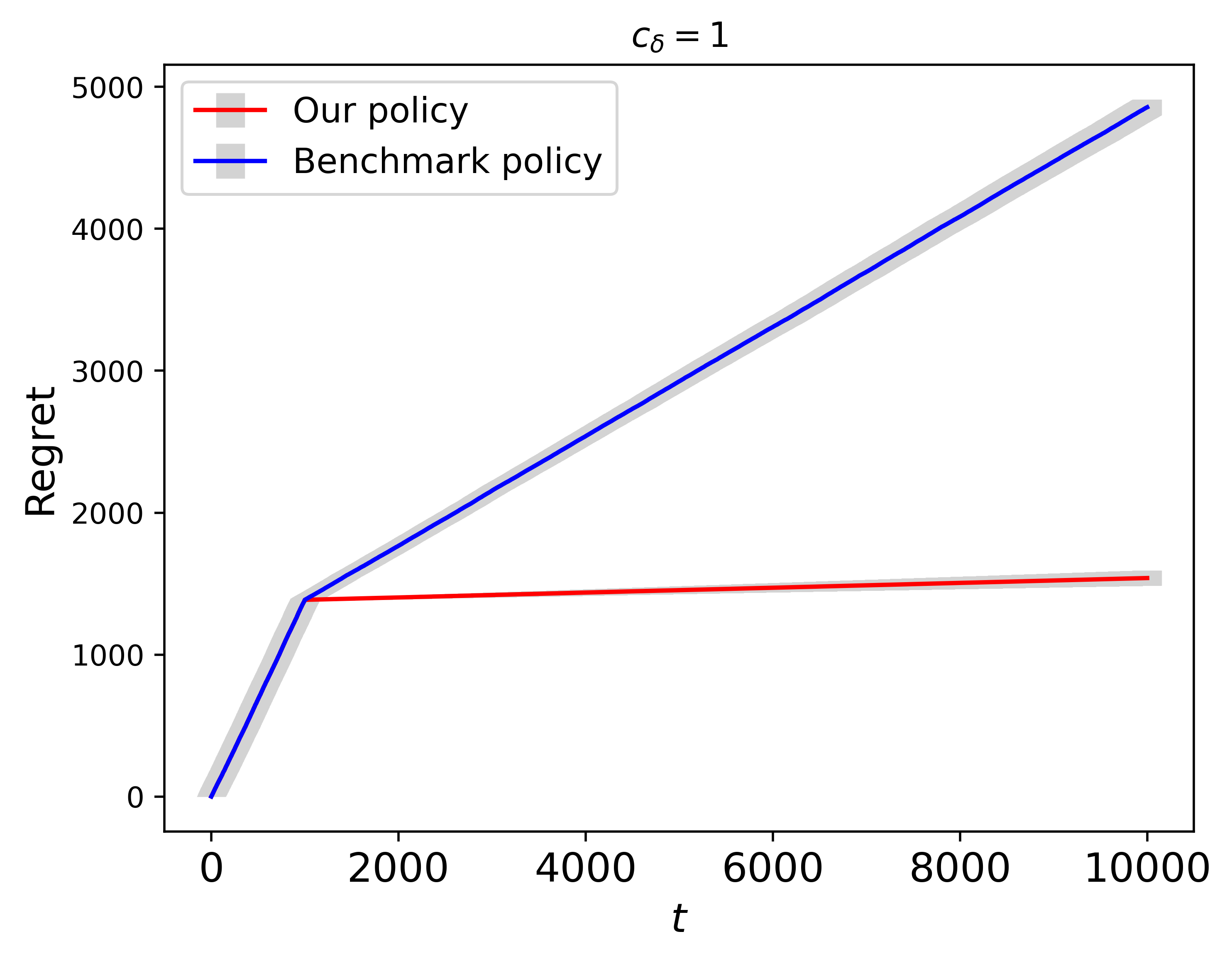}&
        \includegraphics[scale = 0.31]{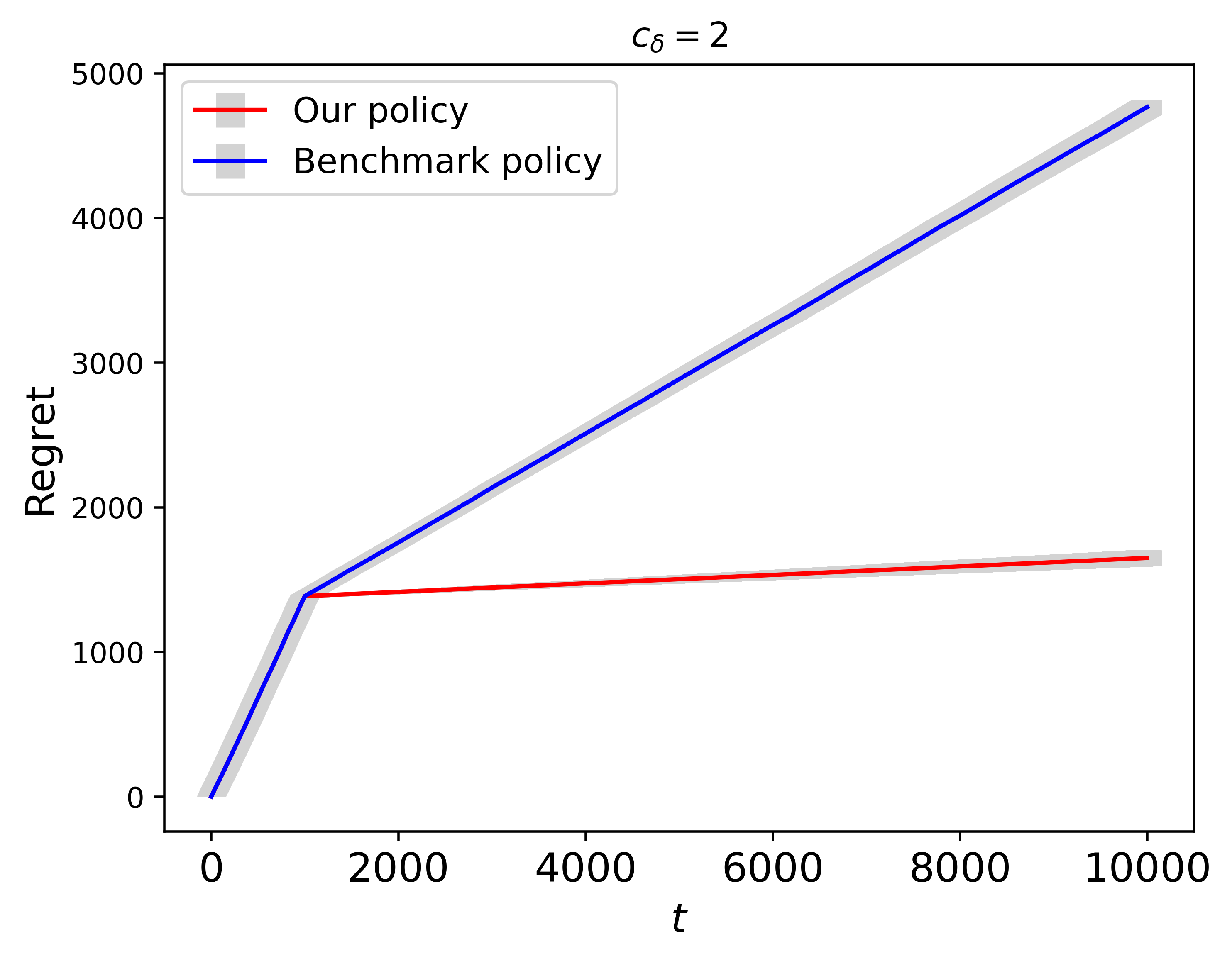}&
        \includegraphics[scale = 0.31]{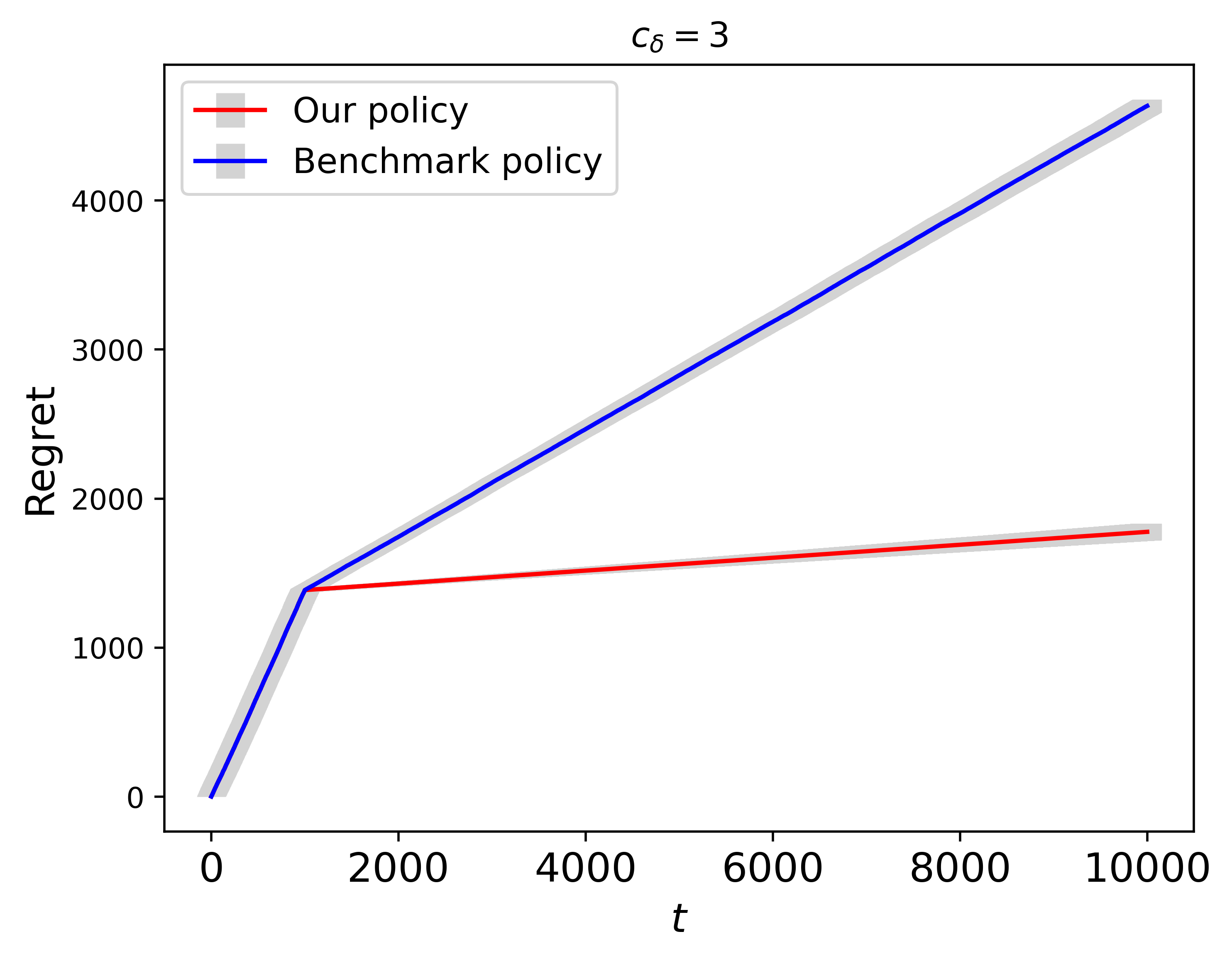}
    \end{tabular}
     \caption{Regret plots for the two policies. The three subplots show the regrets of three different scenarios, $c_\delta\in\{1, 2, 3\}$. }
         \label{fig301cost}
\end{figure}
\section{Discussion on Lower Bound}\label{do}
We derived a lower bound under the non-contextual pricing setting in Theorem \ref{thm3}. In this section, we present a lower bound in the contextual pricing setting under stronger assumptions by following \citet{chai2024localized}. We construct a special instance that satisfies Theorem 3 in \citet{chai2024localized}, which allows us to invoke their result establishing a dimension-free $\Omega(\sqrt{T})$ regret lower bound in contextual pricing. 
We include their assumption below for clarity:

\begin{assumption}[\cite{chai2024localized}]\label{assc}
\leavevmode
\begin{enumerate}
    \item The distribution $\mu$ of feature vectors $\boldsymbol{x} \in \mathbb{R}^d$ satisfies an anti-concentration property: there exists a constant $c_{\mathrm{ac}} > 0$ such that for any symmetric matrix $\boldsymbol{A} \in \mathbb{S}^d$ with $\mathrm{rank}(\boldsymbol{A}) \leq 4$,
    $
    \mathbb{E}_{\boldsymbol{x} \sim \mu} \left[ (\boldsymbol{x}^\top A \boldsymbol{x})^2 \right] \geq c_{\mathrm{ac}} \cdot \|\boldsymbol{A}\|_F^2.
    $

    \item The demand model is non-degenerate, meaning there exists a constant $0<c <1$ such that
    $
    \|\boldsymbol{\alpha}_j\|_2 \geq c, \quad \|\boldsymbol{\beta}_j\|_2 \geq c, \quad \text{and} \quad \left| \boldsymbol{\alpha}_j^\top \boldsymbol{\beta}_j \right| \leq (1 - c) \|\boldsymbol{\alpha}_j\|_2 \cdot \|\boldsymbol{\beta}_j\|_2$ for $j=0, 1$.  
\end{enumerate}
\end{assumption}

In our construction, we assume the following demand model:
\[
y_{jt} = \boldsymbol{\beta}_j^\top \tilde{\boldsymbol{x}}_t + p_t \cdot (\tilde{\boldsymbol{x}}_t^\top \boldsymbol{\alpha}_j) + \epsilon_t,
\]
where $\tilde{\boldsymbol{x}}_t = (1, \boldsymbol{x}_t^\top)^\top \in \mathbb{R}^{d+1}$ and $\epsilon_t$ is from the normal distribution $N(0, \sigma_\epsilon^2)$. We set $\boldsymbol{\alpha}_j = (\alpha_j, 0, \dots, 0)^\top \in \mathbb{R}^{d+1}$.
Suppose that buyers can perfectly learn the price difference, and Assumption \ref{assc} holds.
We further assume that the fairness constraint is not binding.
Under this setup, the problem instance aligns with the setting of Theorem 3 in \cite{chai2024localized}, and therefore a dimension-free regret lower bound of $\Omega(\sqrt{T})$ applies.

\section{Framework of a Nash Equilibrium Model}\label{fo}
In this section, we present the framework by considering the Nash equilibrium. We consider a simplified model to illustrate it by assuming buyers know the values of $\delta$, and the seller knows the demand parameters.
 Inspired by \citet{rabin1993incorporating} and \citet{nelson2001incorporating}, we outline a utility-based framework where both material and fairness payoffs influence decision-making. The full utility consists of material payoffs and fairness payoffs.  The seller’s material payoff is the expected revenue, while the fairness payoff $f_s(\delta)$ captures reputational benefits, compliance with fairness regulations, or increased willingness to pay from fairness-aware buyers. \cite{rabin1993incorporating} indicates that people not only pursue material self-interest but also care about social goals. 
    
    We begin by analyzing a simple static model involving one seller and two types of buyers: those from group 0 and those from group 1. We assume that group 0 faces discrimination. The demand for a buyer from group $j\in\{0, 1\}$ is $y_j=\alpha_jp+\boldsymbol{\beta}_j^\top \tilde{\boldsymbol{x}}+\epsilon$ with $\tilde{\boldsymbol{x}}=(1,\boldsymbol{x}^\top)^\top$, where $\alpha_j$ and $\boldsymbol{\beta_j}$ are demand parameters, $\boldsymbol{x}$ represents observable features, and $\epsilon$ is is a mean-zero noise term. The expected revenue from a buyer in group $j$ is $R_j(p, \boldsymbol{x})=p(\alpha_jp+\boldsymbol{\beta}_j^\top \tilde{\boldsymbol{x}})$. The seller’s material payoff is defined as the expected revenue. It
    depends on the buyer’s reported group status $G'$ and is defined as 
\begin{equation*}
R(p_0, p_1, G',\boldsymbol{x})=\left\{
\begin{aligned}
&R_0(p_0, \boldsymbol{x}), & \text{if}\ \  G'=0,\\
&R_1(p_1, \boldsymbol{x}), & \text{if}\ \  G'=1.
\end{aligned}
\right.
\end{equation*}
Let $\delta$ be the fairness level set by the seller. The seller’s total utility is then $U_s(p_0,p_1, G', \boldsymbol{x}, \delta)=R(p_0, p_1, G',    \boldsymbol{x})+f_s(\delta)$,  subject to the fairness constraint $p_0-p_1\leq \delta$. The seller needs to decide $p_0, p_1$ and $\delta$. The seller’s optimization problem is
    \begin{equation}\label{e1}
\begin{aligned}
\mathop{\max}_{p_0, p_1,\delta}&\ U_s(p_0, p_1, G', \boldsymbol{x},\delta)\\
&s.t. \ p_{0}-p_{1}\leq\delta.
\end{aligned}    
\end{equation}
For a buyer from group 0, the utility depends on the declared group identity $G'\in\{0, 1\}$ and a fairness-related utility $f_b(C_0, \delta)$.  The utility is given by $U_0(G', p_0, p_1, \delta)=-[(1-G')p_0+jp_1]+G'C_0+f_b(C_0, \delta)$, and the buyer chooses $G'$  to maximize utility
\begin{equation}\label{e2}
    \mathop{\max}_{G'}U_0(G', p_0, p_1, \delta).
\end{equation}
Since buyers from group 1 always report their true group status, we do not model their decision-making explicitly. 
Solving the joint game defined by \eqref{e1} and \eqref{e2} would yield a Nash equilibrium.

\section{Proofs}\label{td}
\subsection{Derivation of \eqref{eq6}}\label{derkkt}
The optimization problem \eqref{price} is equivalent to 
\begin{equation*}
\begin{aligned}
\mathop{\min}_{p_0,p_1}&\ -qR_0(p_{0},\boldsymbol{x}_t)-(1-q)R_1(p_{1},\boldsymbol{x}_t)\\
&s.t. \ p_{0}-p_{1}-\delta\leq 0.
\end{aligned}    
\end{equation*}
The Lagrangian function is
\begin{align*}
\mathcal{L}(p_0, p_1, \lambda)&=-qR_0(p_{0},\boldsymbol{x}_t)-(1-q)R_1(p_{1},\boldsymbol{x}_t)+\lambda (p_0-p_1-\delta)\\
&=-q[p_0(\alpha_0p_0+\boldsymbol{\beta}_0^{\top}\tilde{\boldsymbol{x}}_t)]-(1-q)[p_1(\alpha_1p_1+\boldsymbol{\beta}_1^{\top}\tilde{\boldsymbol{x}}_t)]+\lambda (p_0-p_1-\delta).
\end{align*}
By the Karush–Kuhn–Tucker condition, we have
\begin{equation}\label{kkt}
 \begin{aligned}
&\frac{\partial \mathcal{L}(p_0, p_1, \lambda)}{\partial p_0}=-2q\alpha_0p_0-q\boldsymbol{\beta}_0^{\top}\tilde{\boldsymbol{x}}_t+\lambda=0,\\
&\frac{\partial \mathcal{L}(p_0, p_1, \lambda)}{\partial p_1}=-2(1-q)\alpha_1p_1-(1-q)\boldsymbol{\beta}_1^{\top}\tilde{\boldsymbol{x}}_t-\lambda=0,\\
&\lambda(p_0-p_1-\delta)=0,\\
&p_0-p_1-\delta\leq 0,\\
&\lambda \geq 0.
\end{aligned}   
\end{equation}
By solving \eqref{kkt}, we obtain
\begin{equation*}
p_{j}^*(\boldsymbol{x}_t)=\left\{
\begin{aligned}
&-\frac{\boldsymbol{\beta}_j^\top \tilde{\boldsymbol{x}}_t}{2\alpha_j}, & \text{if}\ \  \frac{\boldsymbol{\beta}_1^\top \tilde{\boldsymbol{x}}_t}{2\alpha_1}-\frac{\boldsymbol{\beta}_0^\top \tilde{\boldsymbol{x}}_t}{2\alpha_0}\leq \delta,\\
&\boldsymbol{\gamma}_1^\top \tilde{\boldsymbol{x}}_t-j\cdot \delta+\gamma_2, & \text{if}\  \ \frac{\boldsymbol{\beta}_1^\top \tilde{\boldsymbol{x}}_t}{2\alpha_1}-\frac{\boldsymbol{\beta}_0^\top \tilde{\boldsymbol{x}}_t}{2\alpha_0}> \delta,
\end{aligned}
\right.
\end{equation*}
where 
\begin{equation*}
 \boldsymbol{\gamma}_1=-\frac{q\boldsymbol{\beta}_0+(1-q)\boldsymbol{\beta}_1}{2q\alpha_0+2(1-q)\alpha_1},\ \gamma_2=\frac{(1-q)\alpha_1\delta}{q\alpha_0+(1-q)\alpha_1}.  
\end{equation*}

\subsection{Proof of Theorem \ref{lem0}}\label{ssec1}
To begin the proof, we first construct an instance.
We assume that the dimension of the features is $d=1$, and the expected demands of group 0 and group are $\mathbb{E}y_{0t}=2+x_t-p_{t}$ and $\mathbb{E}y_{1t}=2+x_{t}-2p_{t}$, respectively. We assume $x_t\sim$ Unif(-1/2, 1/2) and set $q=1/2, \delta=1/4$. By \eqref{eq6}, the optimal prices for group 0 and group 1 under the fairness constraint are
$p_{0}^*(x_t)=x_t/3+5/6$ and $ p_{1}^*(x_t)=x_t/3+7/12$, respectively.
We denote $p_{jt}^*=p_{j}^*(\boldsymbol{x}_t)$ and $p_{jt}=p_{j}(\boldsymbol{x}_t)$ for $j=0,1$. The expected revenue of the optimal pricing policy at time $t$ under the fairness constraint is
\begin{equation}\label{neq1}
 \begin{aligned}
\frac{1}{2}[R_0(p_{0t}^*, x_t)+R_1(p_{1t}^*, x_t)]&=\frac{1}{2}[p_{0t}^*(2+x_t-p_{0t}^*)+p_{1t}^*(2+x_t-2p_{1t}^*)]\\
&=\frac{1}{2}\left(\frac{x_t}{3}+\frac{5}{6}\right)\left(\frac{2x_t}{3}+\frac{7}{6}\right)+\left(\frac{x_t}{3}+\frac{7}{12}\right)\left(\frac{x_t}{3}+\frac{5}{6}\right)\\
&=\frac{x_t^2}{6}+\frac{17x_t}{24}+\frac{35}{48}.
\end{aligned}   
\end{equation}
Without the price fairness constraint, the optimal prices for group 0 and group 1 are $p_0^\#(x_t)=(2+x_t)/2$ and  $p_1^\#(x_t)=(2+x_t)/4$, respectively. We set the manipulation cost as $C_0=5/16$. We have $p_0^\#(x_t)-p_1^\#(x_t)>C_0$.
Since the buyers cannot perceive the price fairness, the buyers from group 0 consistently misreport their group status, leading to a payment of $p_1(x_t)$. Therefore, the revenue of any other pricing policy without fairness learning at time $t$ is
\begin{equation}\label{neq2}
\frac{1}{2}[R_0(p_{1t},x_t)+ R_1(p_{1t},x_t)]=\frac{1}{2}[p_{1t}(2+x_t-p_{1t})+p_{1t}(2+x_t-2p_{1t})]\leq \frac{(2+x_t)^2}{6}.
\end{equation}
 By \eqref{neq1} and \eqref{neq2}, the expected cumulative regret of any other pricing policy without fairness learning at time $T$ is
 \begin{align*}
 Regret_T&=\frac{1}{2}\sum_{t=1}^T\mathbb{E}\{[R_0(p_{0t}^*, x_t)+R_1(p_{1t}^*, x_t)]-[R_0(p_{1t},x_t)+R_1(p_{1t},x_t)]\}\\
 &\geq\sum_{t=1}^T\mathbb{E}\left[\frac{x_t^2}{6}+\frac{17x_t}{24}+\frac{35}{48}-\frac{(2+x_t)^2}{6}\right]\\
 &\geq \frac{T}{16}.
 \end{align*}
 The proof is completed.

\subsection{Proof of Lemma \ref{lem1}} \label{ssec2}
Since the proofs for the estimation errors of $\boldsymbol{\theta}_0$ and $\boldsymbol{\theta}_1$ are identical, we omit the group symbol in the following proof, assuming that all variables are from the same group.
To facilitate the presentation of the proof, we introduce some new notations. Let $\tilde{p}_t=(1\ p_t)^\top\in \mathbb{R}^2$ and  $z_t=(\tilde{p}_t^\top\ \boldsymbol{x}_t^\top)^\top\in \mathbb{R}^{d+2}$. The prices and features from time 1 to $t$ are formulated as  $Z_{t}=(z_1,\cdots, z_{t})^\top \in \mathbb{R}^{t\times(d+2)}$, the demand is $Y_{t}=(y_1,\cdots,y_{t})^\top\in \mathbb{R}^{t}$, and the noise sequence is $\boldsymbol{\epsilon}_{t}=(\epsilon_1,\cdots,\epsilon_{t})^\top\in\mathbb{R}^{t}$. Therefore, the demand can be expressed as $Y_{t}=Z_{t}\boldsymbol{\theta}+\boldsymbol{\epsilon}_{t}$. \par The OLS estimator of $\boldsymbol{\theta}$ is denoted as 
$\hat{\boldsymbol{\theta}}=(Z_{t}^\top Z_{t})^{-1}Z_{t}^\top Y_{t}$. 
The estimation error of the parameter $\boldsymbol{\theta}$ is given by
\begin{equation}\label{theta1}
\|\boldsymbol{\theta}-\hat{\boldsymbol{\theta}}\|_2^2=\|(Z_{t}^\top Z_{t})^{-1}Z_{t}^\top\boldsymbol{\epsilon}_{t}\|_2^2=\boldsymbol{\epsilon}_{t} ^\top Z_{t}(Z_{t}^\top Z_{t})^{-2}Z_{t}^\top \boldsymbol{\epsilon}_{t}\leq \frac{\boldsymbol{\epsilon}_{t}^\top Z_{t}Z_{t}^\top \boldsymbol{\epsilon}_{t}}{\lambda^2_{min}(Z_{t}^\top Z_{t})},
\end{equation}
where $\lambda_{min}(Z_{t}^\top Z_{t})$ is the minimum eigenvalue of the matrix $Z_{t}^\top Z_{t}$. To prove Lemma \ref{lem1}, we first establish an upper bound of $\lambda_{min}(Z_{t}^\top Z_{t})$ using the following lemma.
\begin{lemma} \label{mineigen}
 Let Assumptions \ref{ass0} and \ref{ass1} hold. Assume that $p_1,\cdots, p_{t}$ are $i.i.d.$ from a uniform distribution $U(0, B)$, $\boldsymbol{x}_1,\cdots,\boldsymbol{x}_{t}$ are $i.i.d.$ samples from an unknown distribution with $\mathbb{E}\boldsymbol{x}_i=\boldsymbol{0}$, and $p_i$ is independent from $\boldsymbol{x}_i$. Then, with a probability of at least  $1-(d+2)(\frac{e}{2})^{-\frac{\lambda_0t}{2L}}$, the minimum eigenvalue of $Z_{t}^\top Z_{t}=\sum_{i=1}^{t}z_iz_i^\top$ is given by
$$\lambda_{min}(Z_{t}^\top Z_{t})\geq \frac{\lambda_0 t}{2},$$
where 
$$\lambda_0=\mathop{\min}\bigg\{\frac{B^2+3-\sqrt{B^4+3B^2+9}}{6},\lambda_{min}(\Sigma_x)\bigg\} \textit{and}\ L=B^2+1+x_{max}^2.$$
\end{lemma}
\begin{proof} 
Note that $p_i$ is independent of $\boldsymbol{x}_i$. This together with the fact that $\mathbb{E}\boldsymbol{x}_i=\boldsymbol{0}$ leads to the conclusion $\mathbb{E}(p_i\boldsymbol{x}_i)=\mathbb{E}p_i\mathbb{E}(\boldsymbol{x}_i)=\boldsymbol{0}$.
Using $z_i=(\tilde{p}_i^\top\ \boldsymbol{x}_i^\top)^\top$, we obtain 
\begin{align*}
\mathbb{E}z_iz_i^\top=   \begin{pmatrix}
\mathbb{E}\tilde{p}_i\tilde{p}_i^\top & \mathbb{E}\tilde{p}_i\boldsymbol{x}_i^\top \\
\mathbb{E}\boldsymbol{x}_i\tilde{p}_i^\top  & \mathbb{E}\boldsymbol{x}_i\boldsymbol{x}_i^\top
\end{pmatrix}=\begin{pmatrix}
\mathbb{E}\tilde{p}_i\tilde{p}_i^\top & \boldsymbol{0}^\top \\
\boldsymbol{0}  & \Sigma_x
\end{pmatrix}. 
\end{align*}
Recall that $\tilde{p}_t=(1\ p_t)^\top$. Given that $p_i$ follows a uniform distribution $U(0, B)$, we have
$$\mathbb{E}\tilde{p}_i\tilde{p}_i^\top=\begin{pmatrix}
1 & \mathbb{E}p_i \\
 \mathbb{E}p_i &  \mathbb{E}p_i^2
\end{pmatrix}=\begin{pmatrix}
1 & B/2 \\
B/2 & B^2/3
\end{pmatrix}.$$
By simple calculation, the minimum eigenvalue of $ \mathbb{E}\tilde{p}_i\tilde{p}_i^\top$ is $\lambda_{min}( \mathbb{E}\tilde{p}_i\tilde{p}_i^\top)=\frac{B^2+3-\sqrt{B^4+3B^2+9}}{6}>0$. Consequently, the minimum eigenvalue of $\mathbb{E}z_iz_i^\top$ can be expressed as
\begin{align*}
\lambda_{min}(\mathbb{E}z_iz_i^\top)&=\mathop{\min}\{\lambda_{min}( \mathbb{E}\tilde{p}_i\tilde{p}_i^\top), \lambda_{min}(\Sigma_x)\}=\mathop{\min}\bigg\{\frac{B^2+3-\sqrt{B^4+3B^2+9}}{6},\lambda_{min}(\Sigma_x)\bigg\}.   
\end{align*}
Under Assumption \ref{ass1}, we know $\lambda_{min}(\Sigma_x)>0$. Thus, $\lambda_0>0$.
Because $\{z_iz_i^\top\}_{i=1}^{t}$ are $i.i.d.$, we have
\begin{align*}
\lambda_{min}\left(\sum_{i=1}^{t}\mathbb{E}z_iz_i^\top\right)= \lambda_{min}(t\mathbb{E}z_iz_i^\top)=t\lambda_{min}(\mathbb{E}z_iz_i^\top)=\lambda_0t.
\end{align*}
According to Assumption \ref{ass0}, the maximum eigenvalue of $z_iz_i^\top$ is
\begin{align*}
\lambda_{max}(z_iz_i^\top)&= tr(z_i^\top z_i)=1+p_i^2+\|\boldsymbol{x}_i\|_2^2\leq B^2+1+x_{max}^2:=L.
\end{align*}
Obviously, $\{z_iz_i^\top\}_{i=1}^{t}$ are independent, random, self-adjoint matrices with dimension $d+2$. According to Lemma \ref{tropp} and the fact that each matrix $z_iz_i^\top$ is positive semi-definite, with $\zeta=1/2$, we have
$$\mathbb{P}\bigg\{\lambda_{min}\bigg(\sum_{i=1}^{t}z_iz_i^\top\bigg)\leq \frac{\lambda_0t}{2}\bigg\}\leq (d+2)\bigg(\frac{e}{2}\bigg)^{-\frac{\lambda_0t}{2L}}.$$ This completes the proof.
\end{proof}
We now return to the proof of Lemma \ref{lem1}. During the exploration phase, the price $p_i$ is $i.i.d.$ from $U(0, B)$, and the feature $\boldsymbol{x}_i$ is $i.i.d.$, with $p_i$ being independent from $\boldsymbol{x}_i$. Utilizing Lemma \ref{mineigen}, we observe that
\begin{equation}\label{Zmax}
\mathbb{P}\bigg(\lambda_{min}(Z_{t}^\top Z_{t})> \frac{\lambda_0t}{2}\bigg)\geq 1-(d+2)\bigg(\frac{e}{2}\bigg)^{-\frac{\lambda_0t}{2L}}.
\end{equation}
Next, we establish an upper bound for $\boldsymbol{\epsilon}_{t}^\top Z_{t}Z_{t}^\top \boldsymbol{\epsilon}_{t}$. By Assumption \ref{ass0}, we have $\|\boldsymbol{\theta}-\hat{\boldsymbol{\theta}}\|_2^2\leq 2(a_{max}^2+b_{max}^2)$.
Using (\ref{theta1}) and (\ref{Zmax}), the expectation of the estimation error of $\boldsymbol{\theta}$ is given by
\begin{equation}\label{norm1}
\begin{aligned}
\mathbb{E}(\|\boldsymbol{\theta}-\hat{\boldsymbol{\theta}}\|_2^2)&= \mathbb{E}[\|\boldsymbol{\theta}-\hat{\boldsymbol{\theta}}\|_2^2I(\lambda_{min}(Z_{t}^\top Z_{t})\geq \frac{\lambda_0t}{2})]+\mathbb{E}[\|\boldsymbol{\theta}-\hat{\boldsymbol{\theta}}\|_2^2I(\lambda_{min}(Z_{t}^\top Z_{t})< \frac{\lambda_0t}{2})]\\
&\leq \mathbb{E}\bigg[\frac{\|Z_{t}^\top \boldsymbol{\epsilon}_{t}\|_2^2I(\lambda_{min}(Z_{t}^\top Z_{t})\geq \lambda_0t/2)}{\lambda_{min}^2(Z_{t}^\top Z_{t})}\bigg]+2(a_{max}^2+b_{max}^2)\mathbb{P}\bigg(\lambda_{min}(Z_{t}^\top Z_{t})< \frac{\lambda_0t}{2}\bigg)\\
&\leq \frac{4\mathbb{E}\|Z_{t}^\top \boldsymbol{\epsilon}_{t}\|_2^2}{\lambda_0^2 t^2}+2(a_{max}^2+b_{max}^2)(d+2)\bigg(\frac{e}{2}\bigg)^{-\frac{\lambda_0t}{2L}}.
\end{aligned}    
\end{equation}
Considering that $\epsilon_i$ is independent of $\epsilon_j$, $z_i$ and $z_j$ for $i\neq j$, and $\mathbb{E}\epsilon_i =0$, we have $\mathbb{E}(z_i^\top z_j\epsilon_i\epsilon_j)=0$.  Therefore,
\begin{align}\label{norm2}
 \mathbb{E}\|Z_{t}^\top \boldsymbol{\epsilon}_{t}\|_2^2=\mathbb{E}\bigg(\sum_{i=1}^{t}z_i^\top z_i\epsilon_i^2\bigg)  =\sigma_{\epsilon}^2\mathbb{E}\bigg(\sum_{i=1}^{t} z_i^\top z_i\bigg)\leq t\sigma_{\epsilon}^2(B^2+1+x_{max}^2).
\end{align}
Combining (\ref{norm1}) and (\ref{norm2}), we have
\begin{align*}
 \mathbb{E}(\|\boldsymbol{\theta}-\hat{\boldsymbol{\theta}}\|_2^2)&\leq \frac{4\sigma_{\epsilon}^2(B^2+1+x_{max}^2)}{\lambda_0^2 t}+2(a_{max}^2+b_{max}^2)(d+2)\left(\frac{e}{2}\right)^{-\frac{\lambda_0t}{2L}}.
\end{align*}    
Noting that when $t\geq 6$, $(\frac{e}{2})^{-t}<\frac{1}{t}$. We have when $t\geq \frac{12L}{\lambda_0}$,
\begin{equation}\label{error1}
 \mathbb{E}(\|\hat{\boldsymbol{\theta}}-\boldsymbol{\theta}\|_2^2)\leq  \frac{4L[\sigma_{\epsilon}^2+\lambda_0(a_{max}^2+b_{max}^2)(d+2)]}{\lambda_0^2t}.
\end{equation}   
Since the length of the exploration phase is $T_0$ and the proportion of buyers from group 0 is $q$, the numbers of samples used to estimate  $\boldsymbol{\theta}_{0}$ and $\boldsymbol{\theta}_{1}$ are $qT_0$ and $(1-q)T_0$, respectively. We substitute $t=qT_0$ and $t=(1-q)T_0$ into (\ref{error1}) to obtain the estimation errors of $\boldsymbol{\theta}_{0}$ and $\boldsymbol{\theta}_{1}$. This concludes the proof.

\subsection{Proof of Lemma \ref{lem2}} \label{ssec3}
Let $p_0(\cdot)$ and $p_1(\cdot)$ be the pricing functions in \eqref{p3}. When $ \frac{\hat{\boldsymbol{\beta}}_{1}^\top \tilde{\boldsymbol{x}}_t}{2\hat{\alpha}_{1}}-\frac{\hat{\boldsymbol{\beta}}_{0}^\top \tilde{\boldsymbol{x}}_t}{2\hat{\alpha}_{0}}\leq  \delta-c_\delta\sqrt{\frac{\log T_0}{T_0}}$, the price difference is
\begin{align*}
 \hat{\delta}&=\hat{p}_0(\boldsymbol{x})  -\hat{p}_1(\boldsymbol{x})\\
 &=\hat{p}_0(\boldsymbol{x})-p_0(\boldsymbol{x}) +p_1(\boldsymbol{x}) -\hat{p}_1(\boldsymbol{x})+p_0(\boldsymbol{x})-p_1(\boldsymbol{x})\\
 &=\hat{p}_0(\boldsymbol{x})-p_0(\boldsymbol{x}) +p_1(\boldsymbol{x}) -\hat{p}_1(\boldsymbol{x})+\frac{\hat{\boldsymbol{\beta}}_{1}^\top \tilde{\boldsymbol{x}}_t}{2\hat{\alpha}_{1}}-\frac{\hat{\boldsymbol{\beta}}_{0}^\top \tilde{\boldsymbol{x}}_t}{2\hat{\alpha}_{0}}\\
 &\leq \hat{p}_0(\boldsymbol{x})-p_0(\boldsymbol{x}) +p_1(\boldsymbol{x}) -\hat{p}_1(\boldsymbol{x})-c_\delta\sqrt{\frac{\log T_0}{T_0}}+\delta.
\end{align*}
Therefore,
\begin{equation}\label{12401}
\hat{\delta}-C_0\leq\hat{p}_0(\boldsymbol{x})-p_0(\boldsymbol{x}) +p_1(\boldsymbol{x}) -\hat{p}_1(\boldsymbol{x})-c_\delta\sqrt{\frac{\log T_0}{T_0}}+\delta-C_0.
\end{equation}
When $ \frac{\hat{\boldsymbol{\beta}}_{1}^\top \tilde{\boldsymbol{x}}_t}{2\hat{\alpha}_{1}}-\frac{\hat{\boldsymbol{\beta}}_{0}^\top \tilde{\boldsymbol{x}}_t}{2\hat{\alpha}_{0}}>  \delta-c_\delta\sqrt{\frac{\log T_0}{T_0}}$, the price difference is
\begin{align*}
 \hat{\delta}&=\hat{p}_0(\boldsymbol{x})  -\hat{p}_1(\boldsymbol{x})\\
 &=\hat{p}_0(\boldsymbol{x})-p_0(\boldsymbol{x}) +p_1(\boldsymbol{x}) -\hat{p}_1(\boldsymbol{x})+p_0(\boldsymbol{x})-p_1(\boldsymbol{x})\\
 &=\hat{p}_0(\boldsymbol{x})-p_0(\boldsymbol{x}) +p_1(\boldsymbol{x}) -\hat{p}_1(\boldsymbol{x})+\delta.
\end{align*}
Therefore,
\begin{equation}\label{12402}
\hat{\delta}-C_0=\hat{p}_0(\boldsymbol{x})-p_0(\boldsymbol{x}) +p_1(\boldsymbol{x}) -\hat{p}_1(\boldsymbol{x})+\delta-C_0.
\end{equation}
By \eqref{12401} and \eqref{12402}, we have
\begin{align*}
\hat{\delta}-C_0\leq\hat{p}_0(\boldsymbol{x})-p_0(\boldsymbol{x}) +p_1(\boldsymbol{x}) -\hat{p}_1(\boldsymbol{x})+\delta-C_0.
\end{align*}
At time $t$, buyers learn the price difference $\hat{\delta}$ using $t-1$ samples.
By Assumption \ref{ass2}, with at least probability $1-2\eta_{t-1}$, we have $|\hat{p}_0(\boldsymbol{x})-p_0(\boldsymbol{x})|\leq  \mathcal{E}_{\mathcal{P},\eta_{t-1}}(t-1)$ and $|p_1(\boldsymbol{x}) -\hat{p}_1(\boldsymbol{x})|\leq  \mathcal{E}_{\mathcal{P},\eta_{t-1}}(t-1)$. Therefore, with at least probability $1-2\eta_{t-1}$,
\begin{align*}
\hat{\delta}-C_0\leq 2\mathcal{E}_{\mathcal{P},\eta_{t-1}}(t-1)+\delta-C_0
\end{align*}
Since $\mathcal{E}_{\mathcal{P},\eta_{t-1}}(t-1)$ decreases to 0 as $t$ increases, and $\Delta=\delta-C_0<0$ is a constant, when $t>c$ for some positive constant $c$,  with at least probability $1-2\eta_{t-1}$, we have $\hat{\delta}\leq C_0$.

\subsection{Proof of Theorem \ref{thm2}}\label{ssec4}
The time period is segmented into the exploration phase and the exploitation phase.  The seller’s revenue at time $t$ is $R_j(p_t)=R_j(p_t, \boldsymbol{x}_t)$ for $j=0, 1$. Let
$$reg_t=qR_0(p^*_{0t})+(1-q)R_1(p^*_{1t})-qR_0(p_{0t})-(1-q)R_1(p_{1t})$$ be the regret under Algorithm \ref{alg1} at time period $t$. by Assumption \ref{ass0}, we have
\begin{equation}\label{e2026}
    qR_0(p^*_{0t})+(1-q)R_1(p^*_{1t})-qR_0(p_{0t})-(1-q)R_1(p_{1t})\leq  2B(a_{max}B+b_{max}x_{max}).
\end{equation}

Therefore, the regret at time $t$ in the exploration phase is 
\begin{equation}\label{regret1}
\mathbb{E}(reg_t)\leq 2B(a_{max}B+b_{max}x_{max}).
\end{equation}
Now, we focus on the analysis of the regret during the exploitation phase. During the exploitation phase, The pricing function \eqref{p3} is equivalent to 
\begin{equation}\label{p3_1}
p_t=\left\{
\begin{aligned}
&-\frac{\hat{\boldsymbol{\beta}}_{G_t'}^\top \tilde{\boldsymbol{x}}_t}{2\hat{\alpha}_{G_t'}}, &&  \text{if}\ \  \frac{\hat{\boldsymbol{\beta}}_{1}^\top \tilde{\boldsymbol{x}}_t}{2\hat{\alpha}_{1}}-\frac{\hat{\boldsymbol{\beta}}_{0}^\top \tilde{\boldsymbol{x}}_t}{2\hat{\alpha}_{0}}\leq  \delta-c_\delta\sqrt{\frac{\log T_0}{T_0}},  \\
&\hat{\boldsymbol{\gamma}}_{1}^\top \tilde{\boldsymbol{x}}_t-\delta \cdot G_t'+\hat{\gamma}_{2}, && \text{if}\  \ \frac{\hat{\boldsymbol{\beta}}_{1}^\top \tilde{\boldsymbol{x}}_t}{2\hat{\alpha}_{1}}-\frac{\hat{\boldsymbol{\beta}}_{0}^\top \tilde{\boldsymbol{x}}_t}{2\hat{\alpha}_{0}}\geq \delta+c_\delta\sqrt{\frac{\log T_0}{T_0}}, \\
&\hat{\boldsymbol{\gamma}}_{1}^\top \tilde{\boldsymbol{x}}_t-\delta \cdot G_t'+\hat{\gamma}_{2}, && \text{if}\  \ \delta-c_\delta\sqrt{\frac{\log T_0}{T_0}}<\frac{\hat{\boldsymbol{\beta}}_{1}^\top \tilde{\boldsymbol{x}}_t}{2\hat{\alpha}_{1}}-\frac{\hat{\boldsymbol{\beta}}_{0}^\top \tilde{\boldsymbol{x}}_t}{2\hat{\alpha}_{0}}< \delta+c_\delta\sqrt{\frac{\log T_0}{T_0}}.
\end{aligned}
\right.
\end{equation}
We now show that the probability $\mathbb{P}\left(\delta-c_\delta\sqrt{\frac{\log T_0}{T_0}}<\frac{\hat{\boldsymbol{\beta}}_{1}^\top \tilde{\boldsymbol{x}}_t}{2\hat{\alpha}_{1}}-\frac{\hat{\boldsymbol{\beta}}_{0}^\top \tilde{\boldsymbol{x}}_t}{2\hat{\alpha}_{0}}< \delta+c_\delta\sqrt{\frac{\log T_0}{T_0}}\right)$ is small using the following lemma.
\begin{lemma}\label{lemmas8}
There exists some positive constant $c_1$, such that when $T_0>c_1$,
    $$\mathbb{P}\left(\delta-c_\delta\sqrt{\frac{\log T_0}{T_0}}<\frac{\hat{\boldsymbol{\beta}}_{1}^\top \tilde{\boldsymbol{x}}_t}{2\hat{\alpha}_{1}}-\frac{\hat{\boldsymbol{\beta}}_{0}^\top \tilde{\boldsymbol{x}}_t}{2\hat{\alpha}_{0}}< \delta+c_\delta\sqrt{\frac{\log T_0}{T_0}}\right)\leq \frac{4}{T_0}+(d+2)[(\frac{e}{2})^{-\frac{\lambda_0qT_0}{2L}}+(\frac{e}{2})^{-\frac{\lambda_0(1-q)T_0}{2L}}].$$
\end{lemma}
\begin{proof}
We denote $\Delta_{\boldsymbol{x}_t}=\frac{\boldsymbol{\beta}_{0}^\top \tilde{\boldsymbol{x}}_t}{2\alpha_{0}}-\frac{\boldsymbol{\beta}_{1}^\top \tilde{\boldsymbol{x}}_t}{2\alpha_{1}}+\delta$. Then
\begin{equation}\label{e2025}
\begin{aligned}
&~~~~\mathbb{P}\left(\delta-c_\delta\sqrt{\frac{\log T_0}{T_0}}<\frac{\hat{\boldsymbol{\beta}}_{1}^\top \tilde{\boldsymbol{x}}_t}{2\hat{\alpha}_{1}}-\frac{\hat{\boldsymbol{\beta}}_{0}^\top \tilde{\boldsymbol{x}}_t}{2\hat{\alpha}_{0}}< \delta+c_\delta\sqrt{\frac{\log T_0}{T_0}}\right)\\
&=\mathbb{P}\left(\delta-c_\delta\sqrt{\frac{\log T_0}{T_0}}<\frac{\hat{\boldsymbol{\beta}}_{1}^\top \tilde{\boldsymbol{x}}_t}{2\hat{\alpha}_{1}}-\frac{\boldsymbol{\beta}_{1}^\top \tilde{\boldsymbol{x}}_t}{2\alpha_{1}}+\frac{\boldsymbol{\beta}_{0}^\top \tilde{\boldsymbol{x}}_t}{2\alpha_{0}}-\frac{\hat{\boldsymbol{\beta}}_{0}^\top \tilde{\boldsymbol{x}}_t}{2\hat{\alpha}_{0}}+\frac{\boldsymbol{\beta}_{1}^\top \tilde{\boldsymbol{x}}_t}{2\alpha_{1}}-\frac{\boldsymbol{\beta}_{0}^\top \tilde{\boldsymbol{x}}_t}{2\alpha_{0}}< \delta+c_\delta\sqrt{\frac{\log T_0}{T_0}}\right)\\
&=\mathbb{P}\left(\Delta_{\boldsymbol{x}_t}-c_\delta\sqrt{\frac{\log T_0}{T_0}}<\frac{\hat{\boldsymbol{\beta}}_{1}^\top \tilde{\boldsymbol{x}}_t}{2\hat{\alpha}_{1}}-\frac{\boldsymbol{\beta}_{1}^\top \tilde{\boldsymbol{x}}_t}{2\alpha_{1}}+\frac{\boldsymbol{\beta}_{0}^\top \tilde{\boldsymbol{x}}_t}{2\alpha_{0}}-\frac{\hat{\boldsymbol{\beta}}_{0}^\top \tilde{\boldsymbol{x}}_t}{2\hat{\alpha}_{0}}< \Delta_{\boldsymbol{x}_t}+c_\delta\sqrt{\frac{\log T_0}{T_0}}\right).
\end{aligned}
\end{equation}
There exists some positive constant $c_1$ such that $\Delta_{\boldsymbol{x}_t}>2c_\delta\sqrt{\frac{\log T_0}{T_0}}$ when $T_0>c_1$. For any $T_0>c_1$, by \eqref{e2025}, we have 
\begin{align*}
&~~~~\mathbb{P}\left(\delta-c_\delta\sqrt{\frac{\log T_0}{T_0}}<\frac{\hat{\boldsymbol{\beta}}_{1}^\top \tilde{\boldsymbol{x}}_t}{2\hat{\alpha}_{1}}-\frac{\hat{\boldsymbol{\beta}}_{0}^\top \tilde{\boldsymbol{x}}_t}{2\hat{\alpha}_{0}}< \delta+c_\delta\sqrt{\frac{\log T_0}{T_0}}\right)\\
&\leq \mathbb{P}\left(\frac{\hat{\boldsymbol{\beta}}_{1}^\top \tilde{x}_t}{2\hat{\alpha}_{1}}-\frac{\boldsymbol{\beta}_{1}^\top \tilde{\boldsymbol{x}}_t}{2\alpha_{1}}+\frac{\boldsymbol{\beta}_{0}^\top \tilde{\boldsymbol{x}}_t}{2\alpha_{0}}-\frac{\hat{\boldsymbol{\beta}}_{0}^\top \tilde{\boldsymbol{x}}_t}{2\hat{\alpha}_{0}}>\Delta_{\boldsymbol{x}_t}-c_\delta\sqrt{\frac{\log T_0}{T_0}}\right)\\
&\leq \mathbb{P}\left(\left(\frac{\hat{\boldsymbol{\beta}}_{1}^\top \tilde{\boldsymbol{x}}_t}{2\hat{\alpha}_{1}}-\frac{\boldsymbol{\beta}_{1}^\top \tilde{\boldsymbol{x}}_t}{2\alpha_{1}}+\frac{\boldsymbol{\beta}_{0}^\top \tilde{\boldsymbol{x}}_t}{2\alpha_{0}}-\frac{\hat{\boldsymbol{\beta}}_{0}^\top \tilde{\boldsymbol{x}}_t}{2\hat{\alpha}_{0}}\right)^2>\left(\Delta_{\boldsymbol{x}_t}-c_\delta\sqrt{\frac{\log T_0}{T_0}}\right)^2\right)\\
&\leq \mathbb{P}\left(\left(\frac{\hat{\boldsymbol{\beta}}_{1}^\top \tilde{\boldsymbol{x}}_t}{2\hat{\alpha}_{1}}-\frac{\boldsymbol{\beta}_{1}^\top \tilde{\boldsymbol{x}}_t}{2\alpha_{1}}+\frac{\boldsymbol{\beta}_{0}^\top \tilde{\boldsymbol{x}}_t}{2\alpha_{0}}-\frac{\hat{\boldsymbol{\beta}}_{0}^\top \tilde{\boldsymbol{x}}_t}{2\hat{\alpha}_{0}}\right)^2>\frac{c^2_\delta\log T_0}{T_0}\right).
\end{align*}
Next, we have
\begin{equation}\label{new2}
  \begin{aligned}
\left(\frac{\hat{\boldsymbol{\beta}}_{1}^\top \tilde{\boldsymbol{x}}_t}{2\hat{\alpha}_{1}}-\frac{\boldsymbol{\beta}_{1}^\top \tilde{\boldsymbol{x}}_t}{2\alpha_{1}}+\frac{\boldsymbol{\beta}_{0}^\top \tilde{\boldsymbol{x}}_t}{2\alpha_{0}}-\frac{\hat{\boldsymbol{\beta}}_{0}^\top \tilde{\boldsymbol{x}}_t}{2\hat{\alpha}_{0}}\right)^2&\leq 2\left(\frac{\hat{\boldsymbol{\beta}}_{0}^\top \tilde{\boldsymbol{x}}_t}{2\hat{\alpha}_{0}}-\frac{\boldsymbol{\beta}_0^\top \tilde{\boldsymbol{x}}_t}{2\alpha_0}\right)^2+2\left(\frac{\hat{\boldsymbol{\beta}}_{1}^\top \tilde{\boldsymbol{x}}_t}{2\hat{\alpha}_{1}}-\frac{\boldsymbol{\beta}_1^\top \tilde{\boldsymbol{x}}_t}{2\alpha_1}\right)^2\\
&\leq \frac{\mathop{\max}\{a_{max}^2, b_{max}^2\}x_{max}^2\sum_{j=0}^1\|\hat{\boldsymbol{\theta}}_{j}-\boldsymbol{\theta}_j\|_2^2}{a_{min}^4},
\end{aligned}  
\end{equation}
where the second inequality follows from 
\begin{align*}
 \left(\frac{\hat{\boldsymbol{\beta}}_{j}^\top \tilde{\boldsymbol{x}}_t}{2\hat{\alpha}_{j}}-\frac{\boldsymbol{\beta}_j^\top \tilde{\boldsymbol{x}}_t}{2\alpha_j}\right)^2
&=\left(\frac{\alpha_j\hat{\boldsymbol{\beta}}_{j}^\top \tilde{\boldsymbol{x}}_t-\hat{\alpha}_{j}\boldsymbol{\beta}_j^\top \tilde{\boldsymbol{x}}_t+\alpha_j\boldsymbol{\beta}_j^\top \tilde{\boldsymbol{x}}_t-\alpha_j\boldsymbol{\beta}_j^\top \tilde{\boldsymbol{x}}_t}{2\alpha_j\hat{\alpha}_{j}}\right)^2\\
&=\left(\frac{\alpha_j\tilde{x}_t^\top (\hat{\boldsymbol{\beta}}_{j}-\boldsymbol{\beta}_j)+(\alpha_j-\hat{\alpha}_{j})\boldsymbol{\beta}_j^\top \tilde{\boldsymbol{x}}_t}{2\alpha_j\hat{\alpha}_{j}}\right)^2\\
&\leq \frac{2a_{max}^2x_{max}^2\|\hat{\boldsymbol{\beta}}_{j}-\boldsymbol{\beta}_j\|_2^2+2b_{max}^2x_{max}^2\|\alpha_j-\hat{\alpha}_{j}\|_2^2}{4a_{min}^4}\\
&\leq \frac{\mathop{\max}\{a_{max}^2, b_{max}^2\}x_{max}^2\|\hat{\boldsymbol{\theta}}_{j}-\boldsymbol{\theta}_j\|_2^2}{2a_{min}^4}
\end{align*}
by leveraging Assumption \ref{ass0} for $j=0,1$. 
\begin{lemma}\label{lemthetap}
Assume $0<\eta<1$. Let $T_0$ be the length of the exploration phase and $q$ be the proportion of buyers in group 0. $\hat{\boldsymbol{\theta}}_{0}$ and $\hat{\boldsymbol{\theta}}_{1}$ are obtained by (\ref{price}). Under Assumptions \ref{ass0} and \ref{ass1}, We have with probability at least $1-\eta-(d+2)(e/2)^{-\frac{\lambda_0qT_0}{2L}}$, 
\begin{align*}
\|\hat{\boldsymbol{\theta}}_{0}-\boldsymbol{\theta}_0\|_2^2\leq \frac{8(d+2)z_{max}^2\sigma_{\epsilon}^2\log(2/\eta)}{\lambda_0^2qT_0},
\end{align*}    
and with probability at least $1-\eta-(d+2)(e/2)^{-\frac{\lambda_0(1-q)T_0}{2L}}$, 
\begin{align*}
\|\hat{\boldsymbol{\theta}}_{1}-\boldsymbol{\theta}_1\|_2^2\leq \frac{8(d+2)z_{max}^2\sigma_{\epsilon}^2\log(2/\eta)}{\lambda_0^2(1-q)T_0},
\end{align*}    
where $z_{max}=\mathop{\max}\{1,x_{max},B\}$, $\lambda_0=\mathop{\min}\{(B^2+3-\sqrt{B^4+3B^2+9})/6,\lambda_{min}(\Sigma_x)\}$, and $L=B^2+1+x_{max}^2$.
\end{lemma}
\begin{proof}
 Recall that $Z_t=(z_1,\cdots, z_t)^\top \in \mathbb{R}^{t\times(d+2)}$ with $z_t=(1\ p_t\ \boldsymbol{x}_t^\top)^\top\in \mathbb{R}^{d+2}$ from the exploration phase. We slightly abuse the notation and let $z_{ij}$ be the $(i, j)$-th elements of $Z_t$. Under Assumption \ref{ass0}, we have $|z_{ij}|\leq \mathop{\max}\{1,x_{max},B\}:=z_{max}$. Noting that $\epsilon_i$ is bounded independent $\sigma_\epsilon^2$-sub-Gaussian variable with mean zero, and $z_{ij}$ is independent from $\epsilon_i$, we know that $z_{ij}\epsilon_i$ is zero-mean bounded random variable with variance at most $z_{max}^2\sigma_{\epsilon}^2$. By Lemma \ref{Hoeffding}, for $0<w_j<1$, we have
\begin{equation}\label{subgau}
\begin{aligned}
 &~~~~\mathbb{P}\bigg(\bigg|\sum_{i=1}^tz_{ij}\epsilon_i\bigg|
 < z_{max}\sigma_\epsilon\sqrt{2t\log(2/w_j)}\bigg)\\
 &=1-\mathbb{P}\bigg(\bigg|\sum_{i=1}^tz_{ij}\epsilon_i\bigg|
 \geq z_{max}\sigma_\epsilon\sqrt{2t\log(2/w_j)}\bigg)\\
 &=1-\mathbb{P}\bigg(\sum_{i=1}^tz_{ij}\epsilon_i\geq z_{max}\sigma_\epsilon\sqrt{2t\log(2/w_j)}\bigg)-\mathbb{P}\bigg(-\sum_{i=1}^tz_{ij}\epsilon_i\geq z_{max}\sigma_\epsilon\sqrt{2t\log(2/w_j)}\bigg)\\
 &\geq 1-w_j.   
\end{aligned}    
\end{equation}
  Let $0<w<1$ and $w_j=w/(d+2)$ for $j=1,\cdots,(d+2)$. By (\ref{subgau}), with probability at least $1-w$, we have
\begin{equation}\label{zz}
    \begin{aligned}
\boldsymbol{\epsilon}_t^\top Z_tZ_t^\top \boldsymbol{\epsilon}_t&=\sum_{j=1}^{d+2}\bigg(\sum_{i=1}^tz_{ij}\epsilon_i\bigg)^2\leq 2(d+2)z_{max}^2\sigma_{\epsilon}^2t\log(2/w).
    \end{aligned}
\end{equation}
By (\ref{theta1}), (\ref{Zmax}) and (\ref{zz}), we see
\begin{equation}\label{thetap}
\begin{aligned}
\|\hat{\boldsymbol{\theta}}-\boldsymbol{\theta}\|_2^2\leq \frac{\boldsymbol{\epsilon}_t^\top Z_tZ_t^\top \boldsymbol{\epsilon}_t}{\lambda^2_{min}(Z_t^\top Z_t)}\leq \frac{8(d+2)z_{max}^2\sigma_{\epsilon}^2\log(2/w)}{\lambda_0^2t},
\end{aligned}    
\end{equation}
with probability at least $1-w-(d+2)(e/2)^{-\frac{\lambda_0t}{2L}}$. Since the length of the exploration phase is $T_0$, and the proportion of buyers from group 0 is $q$, the numbers of samples used to estimate  $\boldsymbol{\theta}_{0}$ and $\boldsymbol{\theta}_{1}$ are $qT_0$ and $(1-q)T_0$, respectively. Plugging $t=qT_0$ and $t=(1-q)T_0$ into (\ref{thetap}), respectively, we can obtain the estimation errors of $\boldsymbol{\theta}_{0}$ and $\boldsymbol{\theta}_{1}$.
\end{proof}
By Lemma \ref{lemthetap} with $w=2/T_0$, we have with probability
at most $\frac{4}{T_0}+(d+2)[(e/2)^{-\frac{\lambda_0qT_0}{2L}}+(e/2)^{-\frac{\lambda_0(1-q)T_0}{2L}}]$,
\begin{equation}\label{new3}
  \frac{\mathop{\max}\{a_{max}^2, b_{max}^2\}x_{max}^2\sum_{j=0}^1\|\hat{\boldsymbol{\theta}}_{j}-\boldsymbol{\theta}_j\|_2^2}{a_{min}^4}>\frac{8(d+2)z_{max}^2\sigma_{\epsilon}^2\log(T_0)\mathop{\max}\{a_{max}^2, b_{max}^2\}x_{max}^2}{a_{min}^4\lambda_0^2q(1-q)T_0}.  
\end{equation}
We denote $c_\delta=\sqrt{\frac{8(d+2)z_{max}^2\sigma_{\epsilon}^2\mathop{\max}\{a_{max}^2, b_{max}^2\}x_{max}^2}{a_{min}^4\lambda_0^2q(1-q)}}$. Finally, we obtain
\begin{align*}
\mathbb{P} \left(\delta-c_\delta\sqrt{\frac{\log T_0}{T_0}}<\frac{\hat{\boldsymbol{\beta}}_{1}^\top \tilde{\boldsymbol{x}}_t}{2\hat{\alpha}_{1}}-\frac{\hat{\boldsymbol{\beta}}_{0}^\top \tilde{\boldsymbol{x}}_t}{2\hat{\alpha}_{0}}< \delta+c_\delta\sqrt{\frac{\log T_0}{T_0}}\right)\leq \frac{4}{T_0}+(d+2)[(\frac{e}{2})^{-\frac{\lambda_0qT_0}{2L}}+(\frac{e}{2})^{-\frac{\lambda_0(1-q)T_0}{2L}}].
\end{align*}
\end{proof}
By Assumption \ref{ass0}, we have
\begin{equation}\label{new5}
 \hat{\boldsymbol{\gamma}}_{1}^\top \tilde{\boldsymbol{x}}_t-\delta \cdot G_t'+\hat{\gamma}_{2}\leq \frac{b_{max}x_{max}+2a_{max}\delta}{2a_{min}}.   
\end{equation}
Combining Lemma \ref{lemmas8} and \eqref{e2026}, when $\delta-c_\delta\sqrt{\frac{\log T_0}{T_0}}<\frac{\hat{\boldsymbol{\beta}}_{1}^\top \tilde{\boldsymbol{x}}_t}{2\hat{\alpha}_{1}}-\frac{\hat{\boldsymbol{\beta}}_{0}^\top \tilde{\boldsymbol{x}}_t}{2\hat{\alpha}_{0}}< \delta+c_\delta\sqrt{\frac{\log T_0}{T_0}}$, the expected regret during the exploitation phase at time $t$ is given by 
\begin{equation}\label{j(0)}
\begin{aligned}
J_0:=\mathbb{E}(reg_t)&\leq 2B(a_{max}B+b_{max}x_{max})\left\{\frac{4}{T_0}+(d+2)[(\frac{e}{2})^{-\frac{\lambda_0qT_0}{2L}}+(\frac{e}{2})^{-\frac{\lambda_0(1-q)T_0}{2L}}]\right\}\\
&\leq  \frac{4B(a_{max}B+b_{max}x_{max})[(d+2)L+2\lambda_0q(1-q)]}{\lambda_0q(1-q)T_0},
\end{aligned}
\end{equation}
where the inequality follows from the fact that when $T_0\geq 6, (e/2)^{-T_0}<1/T_0$.

Now, we start to analyze the cases where $\frac{\hat{\boldsymbol{\beta}}_{1}^\top \tilde{\boldsymbol{x}}_t}{2\hat{\alpha}_{1}}-\frac{\hat{\boldsymbol{\beta}}_{0}^\top \tilde{\boldsymbol{x}}_t}{2\hat{\alpha}_{0}}\leq  \delta-c_\delta\sqrt{\frac{\log T_0}{T_0}}$ and $\frac{\hat{\boldsymbol{\beta}}_{1}^\top \tilde{\boldsymbol{x}}_t}{2\hat{\alpha}_{1}}-\frac{\hat{\boldsymbol{\beta}}_{0}^\top \tilde{\boldsymbol{x}}_t}{2\hat{\alpha}_{0}}\geq   \delta+c_\delta\sqrt{\frac{\log T_0}{T_0}}$.
Using the benchmark policy, the price for a buyer with features $\boldsymbol{x}_t$ from group $G_t=j\in\{0, 1\}$ at time $t$ is determined by \eqref{eq6}. By \eqref{p3_1}, our policy is related to the disclosed group status $G_t'$. Given the strategic nature of buyers, there exists the possibility of them revealing a false group status.

Given that the buyers from the group 1 (advantage group) do not manipulate and reveal the true group type, the price for these buyers under our policy is defined as
\begin{equation*}
p_{1t}=\left\{
\begin{aligned}
&-\frac{\hat{\boldsymbol{\beta}}_{1}^\top \tilde{\boldsymbol{x}}_t}{2\hat{\alpha}_{1}}, & \text{if}\ \  \frac{\hat{\boldsymbol{\beta}}_{1}^\top \tilde{\boldsymbol{x}}_t}{2\hat{\alpha}_{1}}-\frac{\hat{\boldsymbol{\beta}}_{0}^\top \tilde{\boldsymbol{x}}_t}{2\hat{\alpha}_{0}}\leq \delta-c_\delta\sqrt{\frac{\log T_0}{T_0}},\\
&\hat{\boldsymbol{\gamma}}_{1}^\top \tilde{\boldsymbol{x}}_t-\delta+\hat{\gamma}_{2}, & \text{if}\  \ \frac{\hat{\boldsymbol{\beta}}_{1}^\top \tilde{\boldsymbol{x}}_t}{2\hat{\alpha}_{1}}-\frac{\hat{\boldsymbol{\beta}}_{0}^\top \tilde{\boldsymbol{x}}_t}{2\hat{\alpha}_{0}}> \delta+c_\delta\sqrt{\frac{\log T_0}{T_0}}.
\end{aligned}
\right.
\end{equation*}
The price for buyers from group 0 is contingent on the group status they reveal. Let $\hat{\delta}$ be the price difference that the buyers in group 0 estimated based on the public data. If $\hat{\delta}>C_0$, the buyers from group 0 reveal a manipulated group type. Conversely, if $\hat{\delta}\leq C_0$, they disclose their true group type. Consequently, under our policy, the price for buyers in group 0 is given by
\begin{equation*}
p_{0t}=\left\{\begin{array}{lr}
p_{0t}', \ \text{if}\  \hat{\delta}\leq C_0,\\
p_{1t}, \ \text{if}\  \hat{\delta}> C_0,
\end{array}
\right.    
\end{equation*}
where
\begin{equation*}
p_{0t}'=\left\{\begin{aligned}
    &-\frac{\hat{\boldsymbol{\beta}}_{0}^\top \tilde{\boldsymbol{x}}_t}{2\hat{\alpha}_{0}},\ & \text{if}\  \ \frac{\hat{\boldsymbol{\beta}}_{1}^\top \tilde{\boldsymbol{x}}_t}{2\hat{\alpha}_{1}}-\frac{\hat{\boldsymbol{\beta}}_{0}^\top \tilde{\boldsymbol{x}}_t}{2\hat{\alpha}_{0}}\leq  \delta-c_\delta\sqrt{\frac{\log T_0}{T_0}},\\
     &\hat{\boldsymbol{\gamma}}_{1}^\top \tilde{\boldsymbol{x}}_t+\hat{\gamma}_{2},\  & \text{if}\  \ \frac{\hat{\boldsymbol{\beta}}_{1}^\top \tilde{\boldsymbol{x}}_t}{2\hat{\alpha}_{1}}-\frac{\hat{\boldsymbol{\beta}}_{0}^\top \tilde{\boldsymbol{x}}_t}{2\hat{\alpha}_{0}}> \delta+c_\delta\sqrt{\frac{\log T_0}{T_0}}.
\end{aligned}
\right.
\end{equation*}
The regret of our policy depends on the probability $\mathbb{P}(\hat{\delta}\leq C_0)$. We define the historical information up to time $t$ as $\mathcal{H}_t=\{\boldsymbol{x}_1,\cdots,\boldsymbol{x}_t, G_1',\cdots,G_t',  y_1,\cdots,y_{t},p_1,\cdots,p_t\}$. We also define $\tilde{\mathcal{H}}_t=\mathcal{H}_t\cup \{\boldsymbol{x}_{t+1}, G_{t+1}\}$ as the filtration including $\boldsymbol{x}_{t+1}$ and $G_{t+1}$. The expected regret at time $t$ in the exploitation phase is 
\begin{equation}\label{regt}
\begin{aligned}
\mathbb{E}(reg_t|\tilde{\mathcal{H}}_{t-1})&=\mathbb{E}\{[qR_0(p^*_{0t})+(1-q)R_1(p^*_{1t})-qR_0(p_{0t})-(1-q)R_1(p_{1t})]|\tilde{\mathcal{H}}_{t-1}\}\mathbb{I}(G_t=1)\\
&~~~+\mathbb{E}\{[qR_0(p^*_{0t})+(1-q)R_1(p^*_{1t})-qR_0(p_{0t})-(1-q)R_1(p_{1t})]|\tilde{\mathcal{H}}_{t-1}\}\mathbb{I}(G_t=0)\\
&=(1-q)\mathbb{E}\{[qR_0(p^*_{0t})+(1-q)R_1(p^*_{1t})-qR_0(p_{0t}')-(1-q)R_1(p_{1t})]|\tilde{\mathcal{H}}_{t-1}\}\\
&~~~+q\mathbb{E}\{q[R_0(p^*_{0t})-R_0(p_{0t}')]+(1-q)[R_1(p^*_{1t})-R_1(p_{1t})]|\tilde{\mathcal{H}}_{t-1}\}\mathbb{P}(\hat{\delta}\leq C_0)\\
&~~~+q\mathbb{E}\{q[R_0(p^*_{0t})-R_0(p_{1t})]+(1-q)[R_1(p^*_{1t})-R_1(p_{1t})]|\tilde{\mathcal{H}}_{t-1}\}\mathbb{P}(\hat{\delta}> C_0)\\
&\leq \underbrace{\mathbb{E}\{q[R_0(p^*_{0t})-R_0(p_{0t}')]+(1-q)[R_1(p^*_{1t})-R_1(p_{1t})]|\tilde{\mathcal{H}}_{t-1}\}}_{J_1}\\
&~~~+\underbrace{q\mathbb{E}\{q[R_0(p^*_{0t})-R_0(p_{1t})]+(1-q)[R_1(p^*_{1t})-R_1(p_{1t})]|\tilde{\mathcal{H}}_{t-1}\}\mathbb{P}(\hat{\delta}> C_0)}_{J_2}.
\end{aligned}
\end{equation}

We first analyze $J_1$. For $J_1$, buyers report their true group status. We can rewrite $J_1$ as 
\begin{equation}\label{j1total}
 \begin{aligned}
J_1&=\underbrace{J_1\mathbb{I}\left(\frac{\boldsymbol{\beta}_{1}^\top \tilde{\boldsymbol{x}}_t}{2\alpha_{1}}-\frac{\boldsymbol{\beta}_{0}^\top \tilde{\boldsymbol{x}}_t}{2\alpha_{0}}\geq \delta, \ \frac{\hat{\boldsymbol{\beta}}_{1}^\top \tilde{\boldsymbol{x}}_t}{2\hat{\alpha}_{1}}-\frac{\hat{\boldsymbol{\beta}}_{0}^\top \tilde{\boldsymbol{x}}_t}{2\hat{\alpha}_{0}}\geq \delta-c_\delta\sqrt{\frac{\log T_0}{T_0}}\right)}_{j_1^{(1)}}\\
&~~~~+\underbrace{J_1\mathbb{I}\left(\frac{\boldsymbol{\beta}_{1}^\top \tilde{\boldsymbol{x}}_t}{2\alpha_{1}}-\frac{\boldsymbol{\beta}_{0}^\top \tilde{\boldsymbol{x}}_t}{2\alpha_{0}}\leq \delta, \ \frac{\hat{\boldsymbol{\beta}}_{1}^\top \tilde{\boldsymbol{x}}_t}{2\hat{\alpha}_{1}}-\frac{\hat{\boldsymbol{\beta}}_{0}^\top \tilde{\boldsymbol{x}}_t}{2\hat{\alpha}_{0}}\leq \delta+c_\delta\sqrt{\frac{\log T_0}{T_0}}\right)}_{J_1^{(2)}}\\
&~~~~+\underbrace{J_1\mathbb{I}\left(\frac{\boldsymbol{\beta}_{1}^\top \tilde{\boldsymbol{x}}_t}{2\alpha_{1}}-\frac{\boldsymbol{\beta}_{0}^\top \tilde{\boldsymbol{x}}_t}{2\alpha_{0}}\geq \delta,\  \frac{\hat{\boldsymbol{\beta}}_{1}^\top \tilde{\boldsymbol{x}}_t}{2\hat{\alpha}_{1}}-\frac{\hat{\boldsymbol{\beta}}_{0}^\top \tilde{\boldsymbol{x}}_t}{2\hat{\alpha}_{0}}\leq \delta-c_\delta\sqrt{\frac{\log T_0}{T_0}}\right)}_{J_1^{(3)}}\\
&~~~~+\underbrace{J_1\mathbb{I}\left(\frac{\boldsymbol{\beta}_{1}^\top \tilde{\boldsymbol{x}}_t}{2\alpha_{1}}-\frac{\boldsymbol{\beta}_{0}^\top \tilde{\boldsymbol{x}}_t}{2\alpha_{0}}\leq \delta, \ \frac{\hat{\boldsymbol{\beta}}_{1}^\top \tilde{\boldsymbol{x}}_t}{2\hat{\alpha}_{1}}-\frac{\hat{\boldsymbol{\beta}}_{0}^\top \tilde{\boldsymbol{x}}_t}{2\hat{\alpha}_{0}}\geq \delta+c_\delta\sqrt{\frac{\log T_0}{T_0}}\right)}_{J_1^{(4)}}.
\end{aligned}   
\end{equation}
Now, we analyze $J_1$ in four cases.

\textbf{Case 1}. When $\frac{\boldsymbol{\beta}_{1}^\top \tilde{\boldsymbol{x}}_t}{2\alpha_{1}}-\frac{\boldsymbol{\beta}_{0}^\top \tilde{\boldsymbol{x}}_t}{2\alpha_{0}}\geq \delta$ and $\frac{\hat{\boldsymbol{\beta}}_{1}^\top \tilde{\boldsymbol{x}}_t}{2\hat{\alpha}_{1}}-\frac{\hat{\boldsymbol{\beta}}_{0}^\top \tilde{\boldsymbol{x}}_t}{2\hat{\alpha}_{0}}\geq \delta-c_\delta\sqrt{\frac{\log T_0}{T_0}}$, the price for buyers from group 0 is $p_{0t}'=\hat{\gamma}_1x_t+\hat{\gamma}_2=p_{1t}+\delta$. Therefore,
\begin{equation}\label{j1}
 \begin{aligned}
J_1^{(1)}&=  \mathbb{E}\{q[R_0(p^*_{0t})-R_0(p_{0t}')]\}+\mathbb{E}\{(1-q)[R_1(p^*_{1t})-R_1(p_{1t})]|\tilde{\mathcal{H}}_{t-1}\} \\
&=\mathbb{E}\{q[p^*_{0t}(\alpha_0p^*_{0t}+\boldsymbol{\beta}_0^\top\tilde{\boldsymbol{x}}_t)-p_{0t}'(\alpha_0p_{0t}'+\boldsymbol{\beta}_0^\top \tilde{\boldsymbol{x}})]+(1-q)[p^*_{1t}(\alpha_1p^*_{1t}+\boldsymbol{\beta}_1^\top\tilde{\boldsymbol{x}}_t)-p_{1t}(\alpha_1p_{1t}+\boldsymbol{\beta}_1^\top\tilde{\boldsymbol{x}}_t)]\}\\
&=\mathbb{E}\{q\alpha_0(p_{0t}^{*2}-p_{0t}'^2)+(1-q)\alpha_1(p_{1t}^{*2}-p_{1t}^2)+q\boldsymbol{\beta}_0^\top \tilde{\boldsymbol{x}}_t(p_{0t}^*-p_{0t}')+(1-q)\boldsymbol{\beta}_1^\top \tilde{\boldsymbol{x}}_t(p_{1t}^*-p_{1t})\}\\
&=\mathbb{E}\{[q\alpha_0+(1-q)\alpha_1](p_{0t}^{*2}-p_{0t}'^2)-2(1-q)\alpha_1\delta (p_{0t}^*- p_{0t}')+(q\boldsymbol{\beta}_0^\top \tilde{\boldsymbol{x}}_t+(1-q)\boldsymbol{\beta}_1^\top \tilde{\boldsymbol{x}}_t)(p_{0t}^*-p_{0t}')\}\\
&=\mathbb{E}\{[(q\alpha_0+(1-q)\alpha_1)(p_{0t}^{*}+p_{0t}')-2(1-q)\alpha_1\delta+(q\boldsymbol{\beta}_0+(1-q)\boldsymbol{\beta}_1 )^\top \tilde{\boldsymbol{x}}_t](p_{0t}^*-p_{0t}')\}\\
&=\mathbb{E}\{[(q\alpha_0+(1-q)\alpha_1)(p_{0t}^{*}+p_{0t}')-2(q\alpha_0+(1-q)\alpha_1)p_{0t}^*](p_{0t}^*-p_{0t}')\\
&=-[q\alpha_0+(1-q)\alpha_1]\mathbb{E}(p_{0t}^*-p_{0t}')^2\\
&\leq a_{max}\mathbb{E}(p_{0t}^*-p_{0t}')^2,
\end{aligned}   
\end{equation}
where the fourth equality is from $p_{1t}^*=p_{0t}^*-\delta$ and  $p_{1t}=p_{0t}'-\delta$, and the last second equality is due to $2(q\alpha_0+(1-q)\alpha_1)p_{0t}^*=2(1-q)\alpha_1\delta-(q\boldsymbol{\beta}_0+(1-q)\boldsymbol{\beta}_1)^\top \tilde{\boldsymbol{x}}_t$ from \eqref{eq6}. We now upper bound the price difference between the optimal policy and our policy.
By \eqref{gam0}, we rewrite the pricing parameters as
\begin{equation}\label{gamma2}
\boldsymbol{\gamma}=\begin{pmatrix}
    \frac{2(1-q)\alpha_1\delta-q\boldsymbol{\beta}_{0[1]}-(1-q)\boldsymbol{\beta}_{1[1]}}{2q\alpha_0+2(1-q)\alpha_1}& -\frac{q\boldsymbol{\beta}_{0[2:d]}+(1-q)\boldsymbol{\beta}_{1[2:d]}}{2q\alpha_0+2(1-q)\alpha_1}
\end{pmatrix},    
\end{equation}
where $\boldsymbol{\beta}_{j[1]}$ is the first component of $\boldsymbol{\beta}_{j}$, and $\boldsymbol{\beta}_{j[2:d]}$ is the second to $d$-th components of $\boldsymbol{\beta}_{j}$ for $j=0, 1$.  We denote $\hat{\boldsymbol{\gamma}}$ as the plug-in estimator of $\boldsymbol{\gamma}$. By \eqref{gamma2}, we can express the prices as $p_{0t}^*=\boldsymbol{\gamma}^\top \tilde{\boldsymbol{x}}_t$ and $p_{0t}'=\hat{\boldsymbol{\gamma}}^\top \tilde{\boldsymbol{x}}_t$. Then, the square of the difference between $p^*_{0t}$ and $p_{0t}'$ is 
\begin{equation}\label{p2}
|p^*_{0t}-p_{0t}'|^2=|(\boldsymbol{\gamma}-\hat{\boldsymbol{\gamma}})^\top\tilde{\boldsymbol{x}}_t|^2\leq  \|\boldsymbol{\gamma}-\hat{\boldsymbol{\gamma}}\|_2^2\|\tilde{\boldsymbol{x}}_t\|_2^2\leq (1+x_{max}^2)\|\boldsymbol{\gamma}-\hat{\boldsymbol{\gamma}}\|_2^2.    
\end{equation}

By (\ref{gamma2}), the estimation error of $\boldsymbol{\gamma}$ can be expressed as
\begin{equation}\label{gamma}
\begin{aligned} 
&~~~~\|\hat{\boldsymbol{\gamma}}-\boldsymbol{\gamma}\|_2^2\\
&=\bigg\|\begin{pmatrix}
\frac{2(1-q)\hat{\alpha}_1\delta-q\hat{\boldsymbol{\beta}}_{0[1]}-(1-q)\hat{\boldsymbol{\beta}}_{1[1]}}{2q\hat{\alpha}_0+2(1-q)\hat{\alpha}_1}-\frac{2(1-q)\alpha_1\delta-q\boldsymbol{\beta}_{0[1]}-(1-q)\boldsymbol{\beta}_{1[1]}}{2q\alpha_0+2(1-q)\alpha_1}& \frac{q\boldsymbol{\beta}_{0[2:d+1]}+(1-q)\boldsymbol{\beta}_{1[2:d+1]}}{2q\alpha_0+2(1-q)\alpha_1}-\frac{q\hat{\boldsymbol{\beta}}_{0[2:d+1]}+(1-q)\hat{\boldsymbol{\beta}}_{1[2:d+1]}}{2q\hat{\alpha}_0+2(1-q)\hat{\alpha}_1}\end{pmatrix}\bigg\|_2^2\\
&\leq 2\delta^2 (1-q)^2\bigg |\frac{\hat{\alpha}_{1}}{q\hat{\alpha}_{0}+(1-q)\hat{\alpha}_{1}}-\frac{\alpha_{1}}{q\alpha_{0}+(1-q)\alpha_{1}}\bigg|^2+\frac{1}{2}\bigg\|\frac{q\hat{\boldsymbol{\beta}}_{0}+(1-q)\hat{\boldsymbol{\beta}}_{1}}{q\hat{\alpha}_{0}+(1-q)\hat{\alpha}_{1}}-\frac{q\boldsymbol{\beta}_{0}+(1-q)\boldsymbol{\beta}_{1}}{q\alpha_{0}+(1-q)\alpha_{1}}\bigg\|_2^2.       
\end{aligned}
\end{equation}
To bound the first and second terms in \eqref{gamma}, respectively, we proceed as follows. Start with the first term,
\begin{equation}\label{alpha}
 \begin{aligned}
\bigg |\frac{\hat{\alpha}_{1}}{q\hat{\alpha}_{0}+(1-q)\hat{\alpha}_{1}}-\frac{\alpha_{1}}{q\alpha_{0}+(1-q)\alpha_{1}}\bigg|^2
&=\bigg|\frac{\hat{\alpha}_{1}[q\alpha_{0}+(1-q)\alpha_{1}]-\alpha_{1}[q\hat{\alpha}_{0}+(1-q)\hat{\alpha}_{1}]}{[q\hat{\alpha}_{0}+(1-q)\hat{\alpha}_{1}][q\alpha_{0}+(1-q)\alpha_{1}]}\bigg|^2\\
&=\bigg|\frac{q(\hat{\alpha}_{1}\alpha_{0}-\alpha_{1}\hat{\alpha}_{0}+\alpha_{0}\alpha_{1}-\alpha_{0}\alpha_{1})}{[q\hat{\alpha}_{0}+(1-q)\hat{\alpha}_{1}][q\alpha_{0}+(1-q)\alpha_{1}]}\bigg|^2\\
&=\bigg|\frac{q[\alpha_{0}(\hat{\alpha}_{1}-\alpha_{1})-\alpha_{1}(\hat{\alpha}_{0}-\alpha_{0})}{[q\hat{\alpha}_{0}+(1-q)\hat{\alpha}_{1}][q\alpha_{0}+(1-q)\alpha_{1}]}\bigg|^2\\
&\leq \frac{2q^2[\alpha_{0}^2(\hat{\alpha}_{1}-\alpha_{1})^2+\alpha_{1}^2(\hat{\alpha}_{0}-\alpha_{0})^2]}{[q\hat{\alpha}_{0}+(1-q)\hat{\alpha}_{1}]^2[q\alpha_{0}+(1-q)\alpha_{1}]^2}\\
&\leq \frac{2a_{max}^2q^2(|\hat{\alpha}_{0}-\alpha_0|^2+|\hat{\alpha}_{1}-\alpha_1|^2)}{a_{min}^4}.
\end{aligned}   
\end{equation}
Now, consider the second term,
\begin{equation}\label{beta}
 \begin{aligned}
&~~~~\bigg\|\frac{q\hat{\boldsymbol{\beta}}_{0}+(1-q)\hat{\boldsymbol{\beta}}_{1}}{q\hat{\alpha}_{0}+(1-q)\hat{\alpha}_{1}}-\frac{q\boldsymbol{\beta}_{0}+(1-q)\boldsymbol{\beta}_{1}}{q\alpha_{0}+(1-q)\alpha_{1}}\bigg\|_2^2\\
& =\bigg\|\frac{[q\hat{\boldsymbol{\beta}}_{0}+(1-q)\hat{\boldsymbol{\beta}}_{1}][q\alpha_{0}+(1-q)\alpha_{1}]-[q\boldsymbol{\beta}_{0}+(1-q)\boldsymbol{\beta}_{1}][q\hat{\alpha}_{0}+(1-q)\hat{\alpha}_{1}]}{[q\hat{\alpha}_{0}+(1-q)\hat{\alpha}_{1}][q\alpha_{0}+(1-q)\alpha_{1}]}\bigg\|_2^2\\
&\leq \frac{2a_{max}^2\|(1-q)(\hat{\boldsymbol{\beta}}_{1}-\boldsymbol{\beta}_{1})+q(\hat{\boldsymbol{\beta}}_{0}-\boldsymbol{\beta}_{0})\|_2^2+2b_{max}^2[q(\hat{\alpha}_{0}-\alpha_{0})+(1-q)(\hat{\alpha}_{1}-\alpha_{1})]^2}{a_{min}^4}\\
&\leq \frac{4a_{max}^2[(1-q)^2\|\hat{\boldsymbol{\beta}}_{1}-\boldsymbol{\beta}_{1}\|_2^2+q^2\|\hat{\boldsymbol{\beta}}_{0}-\boldsymbol{\beta}_{0}\|_2^2]+4b_{max}^2[q^2|\hat{\alpha}_{0}-\alpha_{0}|^2+(1-q)^2|\hat{\alpha}_{1}-\alpha_{1}|^2]}{a_{min}^4}.
\end{aligned}   
\end{equation}
Substituting \eqref{alpha} and \eqref{beta} into \eqref{gamma}, we obtain
\begin{equation}\label{gammaerror}
 \begin{aligned}
&~~~\|\hat{\boldsymbol{\gamma}}-\boldsymbol{\gamma}\|_2^2\\
&\leq \frac{4\delta^2q^2a_{max}^2(1-q)^2(|\hat{\alpha}_{0}-\alpha_0|^2+|\hat{\alpha}_{1}-\alpha_1|^2)}{a_{min}^4}\\
&~~~~+\frac{4a_{max}^2((1-q)^2\|\hat{\boldsymbol{\beta}}_{1}-\boldsymbol{\beta}_{1}\|_2^2+q^2\|\hat{\boldsymbol{\beta}}_{0}-\boldsymbol{\beta}_{0}\|_2^2)+4b_{max}^2[q^2|\hat{\alpha}_{0}-\alpha_{0}|^2+(1-q)^2|\hat{\alpha}_{1}-\alpha_{1}|^2]}{a_{min}^4}\\
&\leq\frac{4q^2\mathop{\max}(\delta^2(1-q)^2 a_{max}^2+b_{max}^2, a_{max}^2)\|\hat{\boldsymbol{\theta}}_{0}-\boldsymbol{\theta}_0\|_2^2+4(1-q)^2\mathop{\max}(\delta^2q^2 a_{max}^2+b_{max}^2, a_{max}^2)\|\hat{\boldsymbol{\theta}}_{1}-\boldsymbol{\theta}_1\|_2^2}{a_{min}^4}\\
&\leq \frac{4\mathop{\max}\{\delta^2 a_{max}^2+b_{max}^2, a_{max}^2\}[\|\hat{\boldsymbol{\theta}}_{0}-\boldsymbol{\theta}_0\|_2^2+\|\hat{\boldsymbol{\theta}}_{1}-\boldsymbol{\theta}_1\|_2^2]}{a_{min}^4}.
\end{aligned}  
\end{equation}

Combining (\ref{gammaerror}) and (\ref{p2}), we obtain
\begin{equation}\label{p4}
|p^*_{0t}-p_{0t}'|^2\leq \frac{4(1+x_{max}^2)\mathop{\max}\{\delta^2 a_{max}^2+b_{max}^2, a_{max}^2\}[\|\hat{\boldsymbol{\theta}}_{0}-\boldsymbol{\theta}_0\|_2^2+\|\hat{\boldsymbol{\theta}}_{1}-\boldsymbol{\theta}_1\|_2^2]}{a_{min}^4}.
\end{equation}
By Lemma \ref{lem1}, we observe that when $T_0\geq \frac{12L}{\lambda_0\mathop{\min}\{1-q, q\}}$,
\begin{equation}\label{p5}
\begin{aligned}
&~~~~\mathbb{E}\|\hat{\boldsymbol{\theta}}_{0}-\boldsymbol{\theta}_0\|_2^2+\|\hat{\boldsymbol{\theta}}_{1}-\boldsymbol{\theta}_1\|_2^2]\leq \frac{4L[\sigma_{\epsilon}^2+\lambda_0(a_{max}^2+b_{max}^2)(d+2)]}{\lambda_0^2q(1-q)T_0}.
\end{aligned}    
\end{equation}  
By (\ref{j1}), \eqref{p4} and \eqref{p5}, when $T_0\geq \frac{12L}{\lambda_0\mathop{\min}\{1-q, q\}}$, we have 
\begin{equation}\label{j(1)}
\begin{aligned}
J_1^{(1)}&\leq \frac{16La_{max}(1+x_{max}^2)\mathop{\max}\{\delta^2 a_{max}^2+b_{max}^2, a_{max}^2\}[\sigma_{\epsilon}^2+\lambda_0(a_{max}^2+b_{max}^2)(d+2)]}{a_{min}^4\lambda_0^2q(1-q)T_0}.
\end{aligned}
\end{equation}

\textbf{Case 2}. When $\frac{\boldsymbol{\beta}_{1}^\top \tilde{\boldsymbol{x}}_t}{2\alpha_{1}}-\frac{\boldsymbol{\beta}_{0}^\top \tilde{\boldsymbol{x}}_t}{2\alpha_{0}}\leq \delta$ and $\frac{\hat{\boldsymbol{\beta}}_{1}^\top \tilde{\boldsymbol{x}}_t}{2\hat{\alpha}_{1}}-\frac{\hat{\boldsymbol{\beta}}_{0}^\top \tilde{\boldsymbol{x}}_t}{2\hat{\alpha}_{0}}\leq \delta+c_\delta\sqrt{\frac{\log T_0}{T_0}}$, we have $p_{0t}'=-\frac{\hat{\boldsymbol{\beta}}_{0}^\top \tilde{\boldsymbol{x}}_t}{2\hat{\alpha}_{0}}$ and $p_{0t}^*=-\frac{\boldsymbol{\beta}_{0}^\top \tilde{\boldsymbol{x}}_t}{2\alpha_{0}}$. Therefore,
\begin{align*}
R_0(p^*_{0t})-R_0(p_{0t}')&=p^*_{0t}(\alpha_0p^*_{0t}+\boldsymbol{\beta}_0^\top\tilde{\boldsymbol{x}}_t)-p_{0t}'(\alpha_0p_{0t}'+\boldsymbol{\beta}_0^\top \tilde{\boldsymbol{x}})\\
&=\alpha_0(p_{0t}^{*2}-p_{0t}'^2)+\boldsymbol{\beta}_0^\top \tilde{\boldsymbol{x}}(p^*_{0t}-p_{0t}')\\
&=(p^*_{0t}-p_{0t}')[\alpha_0(p^*_{0t}-p_{0t}')+\boldsymbol{\beta}_0^\top \tilde{\boldsymbol{x}}]\\
&=-\alpha_0(p^*_{0t}-p_{0t}')^2.
\end{align*}
Similarly, we have $R_1(p^*_{1t})-R_1(p_{1t}')=-\alpha_1((p^*_{1t}-p_{1t}'))^2$. Thus,
 \begin{equation}
 \begin{aligned}
J_1^{(2)}&=  \mathbb{E}\{q[R_0(p^*_{0t})-R_0(p_{0t}')]\}+\mathbb{E}\{(1-q)[R_1(p^*_{1t})-R_1(p_{1t})]|\tilde{\mathcal{H}}_{t-1}\} \\
&=-\alpha_0 q \mathbb{E}(p^*_{0t}-p_{0t}')^2-\alpha_0(1-q)\mathbb{E}(p^*_{1t}-p_{1t}')^2\\
&\leq a_{max}[q\mathbb{E}(p_{0t}^*-p_{0t}')^2+(1-q)\mathbb{E}(p^*_{1t}-p_{1t}')^2].
\end{aligned}   
\end{equation}
We now upper bound the price difference between the optimal policy and our policy as follows,
\begin{equation}\label{pdiff}
  \begin{aligned}
(p_{0t}^*-p_{0t}')^2&=\left(\frac{\hat{\boldsymbol{\beta}}_{0}^\top \tilde{\boldsymbol{x}}_t}{2\hat{\alpha}_{0}}-\frac{\boldsymbol{\beta}_0^\top \tilde{\boldsymbol{x}}_t}{2\alpha_0}\right)^2\\
&=\left(\frac{\alpha_0\hat{\boldsymbol{\beta}}_{0}^\top \tilde{\boldsymbol{x}}_t-\hat{\alpha}_{0}\boldsymbol{\beta}_0^\top \tilde{\boldsymbol{x}}_t+\alpha_0\boldsymbol{\beta}_0^\top \tilde{\boldsymbol{x}}_t-\alpha_0\boldsymbol{\beta}_0^\top \tilde{\boldsymbol{x}}_t}{2\alpha_0\hat{\alpha}_{0}}\right)^2\\
&=\left(\frac{\alpha_0\tilde{x}_t^\top (\hat{\boldsymbol{\beta}}_{0}-\boldsymbol{\beta}_0)+(\alpha_0-\hat{\alpha}_{0})\boldsymbol{\beta}_0^\top \tilde{\boldsymbol{x}}_t}{2\alpha_0\hat{\alpha}_{0}}\right)^2\\
&\leq \frac{2a_{max}^2x_{max}^2\|\hat{\boldsymbol{\beta}}_{0}-\boldsymbol{\beta}_0\|_2^2+2b_{max}^2x_{max}^2\|\alpha_0-\hat{\alpha}_{0k}\|_2^2}{4a_{min}^4}\\
&\leq \frac{\mathop{\max}\{a_{max}^2, b_{max}^2\}x_{max}^2\|\hat{\boldsymbol{\theta}}_{0}-\boldsymbol{\theta}_0\|_2^2}{2a_{min}^4}.
\end{aligned}  
\end{equation}
Similarly, we have $(p_{1t}^*-p_{1t}')^2\leq \frac{\mathop{\max}\{a_{max}^2, b_{max}^2\}x_{max}^2\|\hat{\boldsymbol{\theta}}_{1}-\boldsymbol{\theta}_1\|_2^2}{2a_{min}^4}$. Therefore,
\begin{equation*}
J_1^{(2)}\leq \frac{a_{max}\mathop{\max}\{a_{max}^2, b_{max}^2\}x_{max}^2\mathbb{E}[q\|\hat{\boldsymbol{\theta}}_{0}-\boldsymbol{\theta}_0\|_2^2+(1-q)\|\hat{\boldsymbol{\theta}}_{1}-\boldsymbol{\theta}_1\|_2^2]}{2a_{min}^4}.   
\end{equation*}
By Lemma \ref{lem1}, we conclude that when $T_0\geq \frac{12L}{\lambda_0\mathop{\min}\{1-q, q\}}$,
\begin{equation}\label{j(2)}
 J_1^{(2)}\leq\frac{2La_{max}x_{max}^2\mathop{\max}\{a_{max}^2, b_{max}^2\}[\sigma_{\epsilon}^2+\lambda_0(a_{max}^2+b_{max}^2)(d+2)]}{a_{min}^4\lambda_0^2T_0}.   
\end{equation}

\textbf{Case 3.} When $\frac{\boldsymbol{\beta}_{1}^\top \tilde{\boldsymbol{x}}_t}{2\alpha_{1}}-\frac{\boldsymbol{\beta}_{0}^\top \tilde{\boldsymbol{x}}_t}{2\alpha_{0}}\geq \delta$ and  $\frac{\hat{\boldsymbol{\beta}}_{1}^\top \tilde{\boldsymbol{x}}_t}{2\hat{\alpha}_{1}}-\frac{\hat{\boldsymbol{\beta}}_{0}^\top \tilde{\boldsymbol{x}}_t}{2\hat{\alpha}_{0}}\leq \delta-c_\delta\sqrt{\frac{\log T_0}{T_0}}$, we know $p_{jt}=\boldsymbol{\gamma}_1^\top\tilde{\boldsymbol{x}}_t-j\cdot \delta+\gamma_2$, $p_{0t}'=-\frac{\hat{\boldsymbol{\beta}}_0^\top\tilde{\boldsymbol{x}}_t}{2\hat{\alpha}_0}$ and $p_{1t}=-\frac{\hat{\boldsymbol{\beta}}_1^\top\tilde{\boldsymbol{x}}_t}{2\hat{\alpha}_1}$.
We calculate the probability
\begin{align*}
&~~~~\mathbb{P} \left(\frac{\boldsymbol{\beta}_{1}^\top \tilde{\boldsymbol{x}}_t}{2\alpha_{1}}-\frac{\boldsymbol{\beta}_{0}^\top \tilde{\boldsymbol{x}}_t}{2\alpha_{0}}\geq \delta, \ \frac{\hat{\boldsymbol{\beta}}_{1}^\top \tilde{\boldsymbol{x}}_t}{2\hat{\alpha}_{1}}-\frac{\hat{\boldsymbol{\beta}}_{0}^\top \tilde{\boldsymbol{x}}_t}{2\hat{\alpha}_{0}}< \delta-c_\delta\sqrt{\frac{\log T_0}{T_0}}\right)\\
&=\mathbb{P} \left( \frac{\hat{\boldsymbol{\beta}}_{1}^\top \tilde{\boldsymbol{x}}_t}{2\hat{\alpha}_{1}}-\frac{\hat{\boldsymbol{\beta}}_{0}^\top \tilde{\boldsymbol{x}}_t}{2\hat{\alpha}_{0}}\leq \delta-c_\delta\sqrt{\frac{\log T_0}{T_0}}\bigg|\frac{\boldsymbol{\beta}_{1}^\top \tilde{\boldsymbol{x}}_t}{2\alpha_{1}}-\frac{\boldsymbol{\beta}_{0}^\top \tilde{\boldsymbol{x}}_t}{2\alpha_{0}}\geq \delta\right)\mathbb{P} \left(\frac{\boldsymbol{\beta}_{1}^\top \tilde{\boldsymbol{x}}_t}{2\alpha_{1}}-\frac{\boldsymbol{\beta}_{0}^\top \tilde{\boldsymbol{x}}_t}{2\alpha_{0}}\geq \delta\right).
\end{align*}
Given $\frac{\boldsymbol{\beta}_{1}^\top \tilde{\boldsymbol{x}}_t}{2\alpha_{1}}-\frac{\boldsymbol{\beta}_{0}^\top \tilde{\boldsymbol{x}}_t}{2\alpha_{0}}\geq \delta$, we have
\begin{align*}
 \frac{\hat{\boldsymbol{\beta}}_{1}^\top \tilde{\boldsymbol{x}}_t}{2\hat{\alpha}_{1}}-\frac{\hat{\boldsymbol{\beta}}_{0}^\top \tilde{\boldsymbol{x}}_t}{2\hat{\alpha}_{0}}&=  \frac{\hat{\boldsymbol{\beta}}_{1}^\top \tilde{\boldsymbol{x}}_t}{2\hat{\alpha}_{1}}-\frac{\boldsymbol{\beta}_{1}^\top \tilde{\boldsymbol{x}}_t}{2\alpha_{1}}+\frac{\boldsymbol{\beta}_{0}^\top \tilde{\boldsymbol{x}}_t}{2\alpha_{0}}-\frac{\hat{\boldsymbol{\beta}}_{0}^\top \tilde{\boldsymbol{x}}_t}{2\hat{\alpha}_{0}}+\frac{\boldsymbol{\beta}_{1}^\top \tilde{\boldsymbol{x}}_t}{2\alpha_{1}}-\frac{\boldsymbol{\beta}_{0}^\top \tilde{\boldsymbol{x}}_t}{2\alpha_{0}}\\
& \geq \frac{\hat{\boldsymbol{\beta}}_{1}^\top \tilde{\boldsymbol{x}}_t}{2\hat{\alpha}_{1}}-\frac{\boldsymbol{\beta}_{1}^\top \tilde{\boldsymbol{x}}_t}{2\alpha_{1}}+\frac{\boldsymbol{\beta}_{0}^\top \tilde{\boldsymbol{x}}_t}{2\alpha_{0}}-\frac{\hat{\boldsymbol{\beta}}_{0}^\top \tilde{\boldsymbol{x}}_t}{2\hat{\alpha}_{0}}+\delta.
\end{align*}
Therefore,
\begin{equation}\label{new1}
 \begin{aligned}
&~~~~\mathbb{P} \left(\frac{\boldsymbol{\beta}_{1}^\top \tilde{\boldsymbol{x}}_t}{2\alpha_{1}}-\frac{\boldsymbol{\beta}_{0}^\top \tilde{\boldsymbol{x}}_t}{2\alpha_{0}}\geq \delta, \ \frac{\hat{\boldsymbol{\beta}}_{1}^\top \tilde{\boldsymbol{x}}_t}{2\hat{\alpha}_{1}}-\frac{\hat{\boldsymbol{\beta}}_{0}^\top \tilde{\boldsymbol{x}}_t}{2\hat{\alpha}_{0}}< \delta-c_\delta\sqrt{\frac{\log T_0}{T_0}}\right)\\
&\leq \mathbb{P}\left(\frac{\hat{\boldsymbol{\beta}}_{1}^\top \tilde{\boldsymbol{x}}_t}{2\hat{\alpha}_{1}}-\frac{\boldsymbol{\beta}_{1}^\top \tilde{\boldsymbol{x}}_t}{2\alpha_{1}}+\frac{\boldsymbol{\beta}_{0}^\top \tilde{\boldsymbol{x}}_t}{2\alpha_{0}}-\frac{\hat{\boldsymbol{\beta}}_{0}^\top \tilde{\boldsymbol{x}}_t}{2\hat{\alpha}_{0}}\leq-c_\delta\sqrt{\frac{\log T_0}{T_0}}\bigg|\frac{\boldsymbol{\beta}_{1}^\top \tilde{\boldsymbol{x}}_t}{2\alpha_{1}}-\frac{\boldsymbol{\beta}_{0}^\top \tilde{\boldsymbol{x}}_t}{2\alpha_{0}}\geq \delta\right)\mathbb{P} \left(\frac{\boldsymbol{\beta}_{1}^\top \tilde{\boldsymbol{x}}_t}{2\alpha_{1}}-\frac{\boldsymbol{\beta}_{0}^\top \tilde{\boldsymbol{x}}_t}{2\alpha_{0}}\geq \delta\right) \\
&=\mathbb{P}\left(\frac{\hat{\boldsymbol{\beta}}_{1}^\top \tilde{\boldsymbol{x}}_t}{2\hat{\alpha}_{1}}-\frac{\boldsymbol{\beta}_{1}^\top \tilde{\boldsymbol{x}}_t}{2\alpha_{1}}+\frac{\boldsymbol{\beta}_{0}^\top \tilde{\boldsymbol{x}}_t}{2\alpha_{0}}-\frac{\hat{\boldsymbol{\beta}}_{0}^\top \tilde{\boldsymbol{x}}_t}{2\hat{\alpha}_{0}}\leq-c_\delta\sqrt{\frac{\log T_0}{T_0}},\ \frac{\boldsymbol{\beta}_{1}^\top \tilde{\boldsymbol{x}}_t}{2\alpha_{1}}-\frac{\boldsymbol{\beta}_{0}^\top \tilde{\boldsymbol{x}}_t}{2\alpha_{0}}\geq \delta\right)\\
&\leq \mathbb{P}\left(\frac{\hat{\boldsymbol{\beta}}_{1}^\top \tilde{\boldsymbol{x}}_t}{2\hat{\alpha}_{1}}-\frac{\boldsymbol{\beta}_{1}^\top \tilde{\boldsymbol{x}}_t}{2\alpha_{1}}+\frac{\boldsymbol{\beta}_{0}^\top \tilde{\boldsymbol{x}}_t}{2\alpha_{0}}-\frac{\hat{\boldsymbol{\beta}}_{0}^\top \tilde{\boldsymbol{x}}_t}{2\hat{\alpha}_{0}}\leq-c_\delta\sqrt{\frac{\log T_0}{T_0}}\right)\\
&\leq \mathbb{P}\left(\left(\frac{\hat{\boldsymbol{\beta}}_{1}^\top \tilde{\boldsymbol{x}}_t}{2\hat{\alpha}_{1}}-\frac{\boldsymbol{\beta}_{1}^\top \tilde{\boldsymbol{x}}_t}{2\alpha_{1}}+\frac{\boldsymbol{\beta}_{0}^\top \tilde{\boldsymbol{x}}_t}{2\alpha_{0}}-\frac{\hat{\boldsymbol{\beta}}_{0}^\top \tilde{\boldsymbol{x}}_t}{2\hat{\alpha}_{0}}\right)^2\geq \frac{c^2_\delta\log T_0}{T_0}\right).
\end{aligned}   
\end{equation}
By \eqref{e2026} and \eqref{new1}, we have
\begin{equation}\label{j(3)}
\mathbb{E}J_1^{(3)}\leq  2B(a_{max}B+b_{max}x_{max})\mathbb{P}\left(\left(\frac{\hat{\boldsymbol{\beta}}_{1}^\top \tilde{\boldsymbol{x}}_t}{2\hat{\alpha}_{1}}-\frac{\boldsymbol{\beta}_{1}^\top \tilde{\boldsymbol{x}}_t}{2\alpha_{1}}+\frac{\boldsymbol{\beta}_{0}^\top \tilde{\boldsymbol{x}}_t}{2\alpha_{0}}-\frac{\hat{\boldsymbol{\beta}}_{0}^\top \tilde{\boldsymbol{x}}_t}{2\hat{\alpha}_{0}}\right)^2\geq \frac{c^2_\delta\log T_0}{T_0}\right).
\end{equation}

\textbf{Case 4.} When $\frac{\boldsymbol{\beta}_{1}^\top \tilde{\boldsymbol{x}}_t}{2\alpha_{1}}-\frac{\boldsymbol{\beta}_{0}^\top \tilde{\boldsymbol{x}}_t}{2\alpha_{0}}\leq \delta$ and  $\frac{\hat{\boldsymbol{\beta}}_{1}^\top \tilde{\boldsymbol{x}}_t}{2\hat{\alpha}_{1}}-\frac{\hat{\boldsymbol{\beta}}_{0}^\top \tilde{\boldsymbol{x}}_t}{2\hat{\alpha}_{0}}\geq \delta+c_\delta\sqrt{\frac{\log T_0}{T_0}}$, we calculate the probability
\begin{align*}
&~~~~\mathbb{P} \left(\frac{\boldsymbol{\beta}_{1}^\top \tilde{\boldsymbol{x}}_t}{2\alpha_{1}}-\frac{\boldsymbol{\beta}_{0}^\top \tilde{\boldsymbol{x}}_t}{2\alpha_{0}}\leq \delta,\   \frac{\hat{\boldsymbol{\beta}}_{1}^\top \tilde{\boldsymbol{x}}_t}{2\hat{\alpha}_{1}}-\frac{\hat{\boldsymbol{\beta}}_{0}^\top \tilde{\boldsymbol{x}}_t}{2\hat{\alpha}_{0}}\geq  \delta+c_\delta\sqrt{\frac{\log T_0}{T_0}}\right)\\
&= \mathbb{P} \left( \frac{\hat{\boldsymbol{\beta}}_{1}^\top \tilde{\boldsymbol{x}}_t}{2\hat{\alpha}_{1}}-\frac{\hat{\boldsymbol{\beta}}_{0}^\top \tilde{\boldsymbol{x}}_t}{2\hat{\alpha}_{0}}\geq \delta+c_\delta\sqrt{\frac{\log T_0}{T_0}}\bigg|\frac{\boldsymbol{\beta}_{1}^\top \tilde{\boldsymbol{x}}_t}{2\alpha_{1}}-\frac{\boldsymbol{\beta}_{0}^\top \tilde{\boldsymbol{x}}_t}{2\alpha_{0}}\leq \delta\right)\mathbb{P} \left(\frac{\boldsymbol{\beta}_{1}^\top \tilde{\boldsymbol{x}}_t}{2\alpha_{1}}-\frac{\boldsymbol{\beta}_{0}^\top \tilde{\boldsymbol{x}}_t}{2\alpha_{0}}\leq \delta\right).
\end{align*}
Given $\frac{\boldsymbol{\beta}_{1}^\top \tilde{\boldsymbol{x}}_t}{2\alpha_{1}}-\frac{\boldsymbol{\beta}_{0}^\top \tilde{\boldsymbol{x}}_t}{2\alpha_{0}}\leq \delta$, we have
\begin{align*}
\frac{\hat{\boldsymbol{\beta}}_{1}^\top \tilde{\boldsymbol{x}}_t}{2\hat{\alpha}_{1}}-\frac{\hat{\boldsymbol{\beta}}_{0}^\top \tilde{\boldsymbol{x}}_t}{2\hat{\alpha}_{0}}&=  \frac{\hat{\boldsymbol{\beta}}_{1}^\top \tilde{\boldsymbol{x}}_t}{2\hat{\alpha}_{1}}-\frac{\boldsymbol{\beta}_{1}^\top \tilde{\boldsymbol{x}}_t}{2\alpha_{1}}+\frac{\boldsymbol{\beta}_{0}^\top \tilde{\boldsymbol{x}}_t}{2\alpha_{0}}-\frac{\hat{\boldsymbol{\beta}}_{0}^\top \tilde{\boldsymbol{x}}_t}{2\hat{\alpha}_{0}}+\frac{\boldsymbol{\beta}_{1}^\top \tilde{\boldsymbol{x}}_t}{2\alpha_{1}}-\frac{\boldsymbol{\beta}_{0}^\top \tilde{\boldsymbol{x}}_t}{2\alpha_{0}}\\
& \leq \frac{\hat{\boldsymbol{\beta}}_{1}^\top \tilde{\boldsymbol{x}}_t}{2\hat{\alpha}_{1}}-\frac{\boldsymbol{\beta}_{1}^\top \tilde{\boldsymbol{x}}_t}{2\alpha_{1}}+\frac{\boldsymbol{\beta}_{0}^\top \tilde{\boldsymbol{x}}_t}{2\alpha_{0}}-\frac{\hat{\boldsymbol{\beta}}_{0}^\top \tilde{\boldsymbol{x}}_t}{2\hat{\alpha}_{0}}+\delta.
\end{align*}

Therefore,
\begin{equation}\label{new4}
\begin{aligned}
&~~~~\mathbb{P} \left(\frac{\boldsymbol{\beta}_{1}^\top \tilde{\boldsymbol{x}}_t}{2\alpha_{1}}-\frac{\boldsymbol{\beta}_{0}^\top \tilde{\boldsymbol{x}}_t}{2\alpha_{0}}\leq \delta,\   \frac{\hat{\boldsymbol{\beta}}_{1}^\top \tilde{\boldsymbol{x}}_t}{2\hat{\alpha}_{1}}-\frac{\hat{\boldsymbol{\beta}}_{0}^\top \tilde{\boldsymbol{x}}_t}{2\hat{\alpha}_{0}}\geq  \delta+c_\delta\sqrt{\frac{\log T_0}{T_0}}\right)\\
&\leq \mathbb{P}\left(\frac{\hat{\boldsymbol{\beta}}_{1}^\top \tilde{\boldsymbol{x}}_t}{2\hat{\alpha}_{1}}-\frac{\boldsymbol{\beta}_{1}^\top \tilde{\boldsymbol{x}}_t}{2\alpha_{1}}+\frac{\boldsymbol{\beta}_{0}^\top \tilde{\boldsymbol{x}}_t}{2\alpha_{0}}-\frac{\hat{\boldsymbol{\beta}}_{0}^\top \tilde{\boldsymbol{x}}_t}{2\hat{\alpha}_{0}}\geq c_\delta\sqrt{\frac{\log T_0}{T_0}}\bigg|\frac{\boldsymbol{\beta}_{1}^\top \tilde{\boldsymbol{x}}_t}{2\alpha_{1}}-\frac{\boldsymbol{\beta}_{0}^\top \tilde{\boldsymbol{x}}_t}{2\alpha_{0}}\leq \delta\right)\mathbb{P} \left(\frac{\boldsymbol{\beta}_{1}^\top \tilde{\boldsymbol{x}}_t}{2\alpha_{1}}-\frac{\boldsymbol{\beta}_{0}^\top \tilde{\boldsymbol{x}}_t}{2\alpha_{0}}\leq \delta\right) \\
&=\mathbb{P}\left(\frac{\hat{\boldsymbol{\beta}}_{1}^\top \tilde{\boldsymbol{x}}_t}{2\hat{\alpha}_{1}}-\frac{\boldsymbol{\beta}_{1}^\top \tilde{\boldsymbol{x}}_t}{2\alpha_{1}}+\frac{\boldsymbol{\beta}_{0}^\top \tilde{\boldsymbol{x}}_t}{2\alpha_{0}}-\frac{\hat{\boldsymbol{\beta}}_{0}^\top \tilde{\boldsymbol{x}}_t}{2\hat{\alpha}_{0}}\geq c_\delta\sqrt{\frac{\log T_0}{T_0}},\ \frac{\boldsymbol{\beta}_{1}^\top \tilde{\boldsymbol{x}}_t}{2\alpha_{1}}-\frac{\boldsymbol{\beta}_{0}^\top \tilde{\boldsymbol{x}}_t}{2\alpha_{0}}\leq \delta\right)\\
&\leq \mathbb{P}\left(\frac{\hat{\boldsymbol{\beta}}_{1}^\top \tilde{\boldsymbol{x}}_t}{2\hat{\alpha}_{1}}-\frac{\boldsymbol{\beta}_{1}^\top \tilde{\boldsymbol{x}}_t}{2\alpha_{1}}+\frac{\boldsymbol{\beta}_{0}^\top \tilde{\boldsymbol{x}}_t}{2\alpha_{0}}-\frac{\hat{\boldsymbol{\beta}}_{0}^\top \tilde{\boldsymbol{x}}_t}{2\hat{\alpha}_{0}}\geq c_\delta\sqrt{\frac{\log T_0}{T_0}}\right)\\
&\leq \mathbb{P}\left(\left(\frac{\hat{\boldsymbol{\beta}}_{1}^\top \tilde{\boldsymbol{x}}_t}{2\hat{\alpha}_{1}}-\frac{\boldsymbol{\beta}_{1}^\top \tilde{\boldsymbol{x}}_t}{2\alpha_{1}}+\frac{\boldsymbol{\beta}_{0}^\top \tilde{\boldsymbol{x}}_t}{2\alpha_{0}}-\frac{\hat{\boldsymbol{\beta}}_{0}^\top \tilde{\boldsymbol{x}}_t}{2\hat{\alpha}_{0}}\right)^2\geq \frac{c_\delta^2\log T_0}{T_0}\right).
\end{aligned}    
\end{equation}
By \eqref{e2026} and \eqref{new4}, we have
\begin{equation}
\mathbb{E}J_1^{(4)}\leq  2B(a_{max}B+b_{max}x_{max})\mathbb{P}\left(\left(\frac{\hat{\boldsymbol{\beta}}_{1}^\top \tilde{\boldsymbol{x}}_t}{2\hat{\alpha}_{1}}-\frac{\boldsymbol{\beta}_{1}^\top \tilde{\boldsymbol{x}}_t}{2\alpha_{1}}+\frac{\boldsymbol{\beta}_{0}^\top \tilde{\boldsymbol{x}}_t}{2\alpha_{0}}-\frac{\hat{\boldsymbol{\beta}}_{0}^\top \tilde{\boldsymbol{x}}_t}{2\hat{\alpha}_{0}}\right)^2\geq \frac{c^2_\delta\log T_0}{T_0}\right).
\end{equation}
 By \eqref{j(3)} and \eqref{j(4)}, we have
 \begin{equation*}
\begin{aligned}
\mathbb{E}[J_1^{(3)}+J_1^{(4)}]\leq 4B(a_{max}B+b_{max}x_{max})\mathbb{P}\left(\left(\frac{\hat{\boldsymbol{\beta}}_{1}^\top \tilde{\boldsymbol{x}}_t}{2\hat{\alpha}_{1}}-\frac{\boldsymbol{\beta}_{1}^\top \tilde{\boldsymbol{x}}_t}{2\alpha_{1}}+\frac{\boldsymbol{\beta}_{0}^\top \tilde{\boldsymbol{x}}_t}{2\alpha_{0}}-\frac{\hat{\boldsymbol{\beta}}_{0}^\top \tilde{\boldsymbol{x}}_t}{2\hat{\alpha}_{0}}\right)^2\geq \frac{c^2_\delta\log T_0}{T_0}\right).
\end{aligned}
 \end{equation*}
 By \eqref{new2} and \eqref{new3}, we obtain
\begin{align*}
\mathbb{P}\left(\left(\frac{\hat{\boldsymbol{\beta}}_{1}^\top \tilde{\boldsymbol{x}}_t}{2\hat{\alpha}_{1}}-\frac{\boldsymbol{\beta}_{1}^\top \tilde{\boldsymbol{x}}_t}{2\alpha_{1}}+\frac{\boldsymbol{\beta}_{0}^\top \tilde{\boldsymbol{x}}_t}{2\alpha_{0}}-\frac{\hat{\boldsymbol{\beta}}_{0}^\top \tilde{\boldsymbol{x}}_t}{2\hat{\alpha}_{0}}\right)^2\geq \frac{c^2_\delta\log T_0}{T_0}\right)\leq \frac{4}{T_0}+(d+2)[(e/2)^{-\frac{\lambda_0qT_0}{2L}}+(e/2)^{-\frac{\lambda_0(1-q)T_0}{2L}}].
\end{align*}
Therefore,
\begin{equation}\label{j(4)}
\begin{aligned}
\mathbb{E}[J_1^{(3)}+J_1^{(4)}]&\leq 4B(a_{max}B+b_{max}x_{max})\left[\frac{4}{T_0}+(d+2)[(e/2)^{-\frac{\lambda_0qT_0}{2L}}+(e/2)^{-\frac{\lambda_0(1-q)T_0}{2L}}]\right]\\
&\leq  \frac{8B(a_{max}B+b_{max}x_{max})[2\lambda_0 q(1-q)+(d+2)L]}{\lambda_0 q(1-q)T_0},    
\end{aligned}
\end{equation}
where the inequality follows from the fact that when $T_0\geq 6, (e/2)^{-T_0}<1/T_0$. By \eqref{j(0)}, \eqref{j1total}, \eqref{j(1)}, \eqref{j(2)} and \eqref{j(4)}, when $T_0$ is larger than some constant, we have
\begin{equation}\label{j11}
\begin{aligned}
    \mathbb{E}J_0+ \mathbb{E}J_1&\leq \frac{4B(a_{max}B+b_{max}x_{max})[(d+2)L+2\lambda_0q(1-q)]}{\lambda_0q(1-q)T_0}\\
     &~~~+\frac{16La_{max}(1+x_{max}^2)\mathop{\max}\{\delta^2 a_{max}^2+b_{max}^2, a_{max}^2\}[\sigma_{\epsilon}^2+\lambda_0(a_{max}^2+b_{max}^2)(d+2)]}{a_{min}^4\lambda_0^2q(1-q)T_0}\\
    &~~~+  \frac{2La_{max}x_{max}^2\mathop{\max}\{a_{max}^2, b_{max}^2\}[\sigma_{\epsilon}^2+\lambda_0(a_{max}^2+b_{max}^2)(d+2)]}{a_{min}^4\lambda_0^2T_0}\\
&~~~+\frac{8B(a_{max}B+b_{max}x_{max})[2\lambda_0 q(1-q)+(d+2)L]}{\lambda_0 q(1-q)T_0}\\
&=\frac{c'_2(d+2)+c'_3}{q(1-q)T_0}
\end{aligned}
\end{equation}
for some positive constants $c'_2, c'_3$.
Now, we analyze $J_2$. By Lemma \ref{lem2},  we have
\begin{equation*}
\mathbb{P}(\hat{\delta}\leq C_0)\geq 1-2\eta_{t-1}.
\end{equation*}
Therefore,
\begin{equation*}
 \mathbb{P}(\hat{\delta}> C_0)=1-\mathbb{P}(\hat{\delta}\leq C_0)\leq 2\eta_{t-1}.
\end{equation*}
We have
\begin{equation}\label{e2027}
\begin{aligned}
&~~~~q\mathbb{E}\{q[R_0(p^*_{0t})-R_0(p_{1t})]+(1-q)[R_1(p^*_{1t})-R_1(p_{1t})]|\tilde{\mathcal{H}}_{t-1}\}\\
&= q\mathbb{E}\{q[R_0(p^*_{0t})-R_0(p'_{0t})]+(1-q)[R_1(p^*_{1t})-R_1(p_{1t})]|\tilde{\mathcal{H}}_{t-1}\}+\mathbb{E}\{q[R_0(p'_{0t})-R_0(p_{1t})]|\tilde{\mathcal{H}}\}\\
&=J_1+\mathbb{E}\{q[R_0(p'_{0t})-R_0(p_{1t})]|\tilde{\mathcal{H}}\}\\
&\leq J_1+2qB(a_{max}B+b_{max}x_{max}).
\end{aligned}
\end{equation}
By \eqref{e2027}, we have
\begin{equation}\label{j2}
\begin{aligned}
J_2&=\mathbb{E}\{q[R_0(p^*_{0t})-R_0(p_{1t})]+(1-q)[R_1(p^*_{1t})-R_1(p_{1t})]|\tilde{\mathcal{H}}_{t-1}\}\mathbb{P}(\hat{\delta}> C_0)\\
&\leq 2[J_1+2qB(a_{max}B+b_{max}x_{max})]\eta_{t-1}.
\end{aligned}
\end{equation}
Denote $\Bar{B}=2B(a_{max}B+b_{max}x_{max})$ and $c'_4=\frac{c'_2(d+2)+c'_3}{q(1-q)}$. We set $T_0=\sqrt{c'_4T/\Bar{B}}$. By \eqref{j(0)}, \eqref{regret1}, \eqref{regt}, \eqref{j11}  and \eqref{j2}, when $T>c_1$ for some positive constant $c_1$, the total regret at $T$ is
\begin{align*}
Regret_T&= T_0 \Bar{B}+(T-T_0)\mathbb{E}J_1+\sum_{t=T_0}^T\mathbb{E}J_2\\
&\leq T_0 \Bar{B}+\frac{c'_4T}{T_0}+2\sum_{t=T_0}^T(\frac{c'_4}{T_0}+q\Bar{B})\eta_{t-1}
\\
&=2\sqrt{c'_4\bar{B}T}+2\left(\sqrt{\frac{c'_4\Bar{B}}{T}}+q\Bar{B}\right)\sum_{t=\sqrt{c'_4T/\Bar{B}}}^T\eta_{t-1}\\
&=\sqrt{\left[\frac{c''_2(d+2)+c''_3}{q(1-q)}\right]T}+(c''_4+c''_5q)\sum_{t=2}^T\eta_{t-1}
\end{align*}
for some positive constants $c''_2, c''_3, c''_4$ and $c''_5$. To minimize the number of constants, we have
$$Regret_T\leq \sqrt{\frac{c_2dT}{q(1-q)}}+c_3q\sum_{t=2}^T\eta_{t-1}$$
for some positive constants $c_2=\frac{c''_2(d+2)+c''_3}{d}$ and $c_3=\frac{c''_4+c''_5q}{q}$.
\subsection{Proof of Theorem \ref{thm3}}\label{ssec5}
Our proof is inspired by \cite{Broder2012}. We first define some new notations. Let $p_{G_t}(\alpha)$ be the price for group $G_t\in \{0,1\}$ with the underlying parameter $\alpha$ and $p_{G_t}^*(\alpha)$ be the corresponding optimal price under the fairness constraint. We denote $d(p_t,G_t,\alpha)=1/2+\alpha [(G_t+1)p_t-1-G_t/2]$ as the expected demand for group $G_t$ at price $p_t$, and $R_{G_t}(p_t,\alpha)=p_td(p_t,G_t,\alpha)$ as the revenue from group $G_t$  with the underlying parameter $\alpha$. We assume $y\in\{0,1\}$, and define the price set satisfying the fairness constraint as $\mathcal{P}=\{(p_0, p_1):p_0-p_1=\delta, p_0\in[1/2, 9/8],p_1\in[1/2, 9/8]\}$, where $p_0$ is the price for group 0 and $p_1$ is the price for group 1.

We first present some properties used in the proof of Theorem \ref{thm3} in the following lemma.
\begin{lemma}\label{property}
Let $q=1/2, \delta=1/4$ and $\alpha_{0}=-2/5$. For any $\alpha\in [-1/2,-1/5]$, $p\in[1/2,9/8], G_t\in\{0,1\}$ and $(p_0,p_1)\in\mathcal{P}$, we have
\begin{itemize}
    \item[1.] $p_{G_t}^*(\alpha)=\frac{7}{12}-\frac{1}{6\alpha}-\frac{G_t}{4}$.
    \item[2.] $p_0^*(\alpha_{0})=1$ and $p_1^*(\alpha_{0})=3/4$.
    \item[3.] $d(p_0^*(\alpha_{0}), 0, \alpha)=\frac{1}{2}$ and $d(p_1^*(\alpha_{0}), 1, \alpha)=\frac{1}{2}$
    \item[4.] $R_0(p_0^*(\alpha),\alpha)-R_0(p_0,\alpha)+R_1(p_1^*(\alpha),\alpha)-R_1(p_1,\alpha)\geq \frac{3}{5}(p_0^*(\alpha)-p_0)^2$,

    $R_0(p_0^*(\alpha),\alpha)-R_0(p_0,\alpha)+R_1(p_1^*(\alpha),\alpha)-R_1(p_1,\alpha)\geq \frac{3}{5}(p_1^*(\alpha)-p_1)^2$.
    \item [5.] $|p_0^*(\alpha)-p_0^*(\alpha_{0})|> \frac{5}{6}|\alpha-\alpha_{0}|$ and $|p_1^*(\alpha)-p_1^*(\alpha_{0})|> \frac{5}{6}|\alpha-\alpha_{0}|$.
    \item[6.]$|d(p, 0, \alpha)-d(p, 0, \alpha_0)|\leq 2|p_0^*(\alpha_0)-p||\alpha-\alpha_0|$,
    
    $|d(p, 1, \alpha)-d(p, 1, \alpha_0)|\leq 2|p_1^*(\alpha_0)-p||\alpha-\alpha_0|$.
\end{itemize}
\end{lemma}
\begin{proof}
 We prove the properties one by one.

\begin{itemize}[leftmargin=0.5cm]
    \item[1.] The expected demands for group 0 and group 1 are
$\mathbb{E}(y_t|p, 0,\alpha)=1/2-\alpha+\alpha p$ and $\mathbb{E}(y_t|p, 1,\alpha)=1/2-3\alpha/2+2\alpha p$, respectively. By \eqref{eq6}, the optimal price with fairness constraint for group $G$ is $p_{G_t}^*(\alpha)=\frac{7}{12}-\frac{1}{6\alpha}-\frac{G_t}{4}$.
\item[2.] By Property 1, we get $p_0^*(\alpha_0)=\frac{7}{12}-\frac{1}{6\alpha_0}=\frac{7}{12}+\frac{5}{12}=1$ and $p_1^*(\alpha_0)=\frac{1}{3}-\frac{1}{6\alpha_0}=\frac{3}{4}$.
\item[3.] By Property 2, we have $p_0^*(\alpha_0)=1$ and $p_1^*(\alpha_0)=3/4$. Therefore, $d(p_0^*(\alpha_{0}), 0, \alpha)=1/2+\alpha(1-1)=1/2$ and $d(p_1^*(\alpha_{0}), 0, \alpha)=1/2+\alpha(2*3/4-3/2)=1/2$.
\item[4.] For simplicity, we denote $p_0^*=p_0^*(\alpha)$ and $p_1^*=p_1^*(\alpha)$. 
\begin{align*}
&~~~R_0(p_0^*,\alpha)-R_0(p_0,\alpha)+R_1(p_1^*,\alpha)-R_1(p_1,\alpha)\\
&=p_0^*\left(\frac{1}{2}-\alpha+\alpha p_0^*\right)-p_0\left(\frac{1}{2}-\alpha+\alpha p_0\right)+p_1^*\left(\frac{1}{2}-\frac{3\alpha}{2}+2\alpha p_1^*\right)-p_1\left(\frac{1}{2}-\frac{3\alpha}{2}+2\alpha p_1\right)\\
&=\alpha(p_0^{*2}-p_0^2)+2\alpha(p_1^{*2}-p_1^2)+\left(\frac{1}{2}-\alpha\right)(p_0^*-p_0)+\left(\frac{1}{2}-\frac{3\alpha}{2}\right)(p_1^*-p_1)\\
&=3\alpha(p_0^{*2}-p_0^2)-4\alpha\delta (p_0^*-p_0)+\left(1-\frac{5\alpha}{2} \right)(p_0^*-p_0)\\
&=\left[3\alpha(p_0^*+p_0)+1-\frac{7\alpha}{2}\right](p_0^*-p_0)\\
&=-3\alpha(p_0^*-p_0)^2\\
&\geq \frac{3}{5}(p_0^*-p_0)^2.
\end{align*}
The third equality is from $p_0^*-p_1^*=\delta$ and $p_0-p_1=\delta$. The fourth equality is due to $\delta=1/4$ and $6\alpha p_0^*=\frac{7\alpha}{2}-1$ derived from $p_0^*=\frac{7}{12}-\frac{1}{6\alpha}$. the last line is from $\alpha\in[-1/2,-1/5]$. Similarly, we can obtain $R_0(p_0^*,\alpha)-R_0(p_0,\alpha)+R_1(p_1^*,\alpha)-R_1(p_1,\alpha)\geq \frac{3}{5}(p_1^*-p_1)^2$.
\item[5.] By Property 1 and $\alpha\in[-1/2,-1/5]$, we have
$$|p_0^*(\alpha)-p_0^*(\alpha_0)|=\left|\frac{1}{6\alpha}-\frac{1}{6\alpha_0}\right|=\frac{1}{6}\left|\frac{\alpha-\alpha_0}{\alpha\alpha_0}\right|\geq \frac{5}{6}|\alpha-\alpha_0|,$$
$$|p_1^*(\alpha)-p_1^*(\alpha_0)|=\left|\frac{1}{6\alpha}-\frac{1}{6\alpha_0}\right|=\frac{1}{6}\left|\frac{\alpha-\alpha_0}{\alpha\alpha_0}\right|\geq \frac{5}{6}|\alpha-\alpha_0|.$$
\item[6.] Since $d(p_t,G_t,\alpha)=1/2+\alpha [(G_t+1)p_t-1-G_t/2]$ and $p_0^*(\alpha_0)=1$ and $p_1^*(\alpha_0)=3/4$, we have
$$|d(p, 0, \alpha)-d(p, 0, \alpha_0)|=|(p-1)(\alpha-\alpha_0)|=  |p_0^*(\alpha_0)-p||\alpha-\alpha_0|,$$
$$|d(p, 1, \alpha)-d(p,1, \alpha_1)|=|(2p-3/2)(\alpha-\alpha_1)|=  2|p_1^*(\alpha_0)-p||\alpha-\alpha_0|.$$
\end{itemize}
\end{proof}
\endproof
Let $Q_t^{\psi,\alpha}$ denote the probability distribution of the buyer responses  $\boldsymbol{Y}_t=(Y_1,\cdots,Y_t)$ in the first $t$ periods when the pricing policy $\psi$ is conducted under the parameter $\alpha$. Thus, for the sequence of demands $\boldsymbol{y}_t=(y_1,\cdots,y_t)$, we have $Q_t^{\psi,\alpha}(\boldsymbol{y})=\Pi_{i=1}^td(p_i,G_i,\alpha)^{y_i}[1-d(p_i,G_i,\alpha)]^{1-y_i}$, where $p_i$ is the price at time $i$ under the pricing policy $\psi$. We define the expected cumulative regret at time $t$ for the policy $\psi$ with the parameter $\alpha$ as
$$Regret(\alpha,t,\psi)=\frac{1}{2}\sum_{s=1}^t\mathbb{E}_{\alpha}[R_0(p_0^*(\alpha),\alpha)-R_0(p_{0s},\alpha)+R_1(p_1^*(\alpha),\alpha)-R_1(p_{1s},\alpha)]$$
We now present a lemma to establish that learning the parameters is costly.
\begin{lemma}\label{cost1}
Let $\alpha_0=-2/5, G_t=(t\bmod 2)$ and $\delta=1/4$. For any $\alpha\in[-1/2,-1/5]$ and any pricing policy $\psi$ satisfying the fairness constraint, we have
$$\mathcal{K}(Q_t^{\psi, \alpha_0},Q_t^{\psi, \alpha})\leq \frac{768}{35}(\alpha_0-\alpha)^2Regret(\alpha_0, t,\psi).$$
\end{lemma}
\begin{proof}
We note that $G_t$ is determined by $t$ and hence not a random variable.
Following \cite{Broder2012}, we have
\begin{equation}\label{chain}
\mathcal{K}(Q_t^{\psi, \alpha_0},Q_t^{\psi, \alpha})=\sum_{s=1}^T\mathcal{K}(Q_s^{\psi, \alpha_0},Q_s^{\psi, \alpha}|\boldsymbol{Y}_{s-1})
\end{equation}  
and
\begin{align*}
&~~~\mathcal{K}(Q_s^{\psi, \alpha_0},Q_s^{\psi, \alpha}|\boldsymbol{Y}_{s-1})\\
&=\sum_{\boldsymbol{y}_s\in\{0,1\}^s}Q_s^{\psi, \alpha_0}(\boldsymbol{y}_s)\log\left[\frac{Q_t^{\psi, \alpha_0}(y_s|\boldsymbol{y}_{s-1})}{Q_t^{\psi, \alpha}(y_s|\boldsymbol{y}_{s-1})}\right]\\
&=\sum_{\boldsymbol{y}_{s-1}\in\{0,1\}^{s-1}}Q_{s-1}^{\psi, \alpha_0}(\boldsymbol{y}_{s-1})\sum_{\boldsymbol{y}_{s}\in\{0,1\}^{s}}Q_s^{\psi, \alpha_0}(y_s|\boldsymbol{y}_{s-1})\log\left[\frac{Q_t^{\psi, \alpha_0}(y_s|\boldsymbol{y}_{s-1})}{Q_s^{\psi, \alpha}(y_s|\boldsymbol{y}_{s-1})}\right]\\
&=\sum_{\boldsymbol{y}_{s-1}\in\{0,1\}^{s-1}}Q_{s-1}^{\psi, \alpha_0}(\boldsymbol{y}_{s-1})\mathcal{K}(Q_s^{\psi,\alpha_0}(y_s|\boldsymbol{y}_{s-1}),Q_s^{\psi,\alpha}(y_s|\boldsymbol{y}_{s-1}))\\
&\leq \sum_{\boldsymbol{y}_{s-1}\in\{0,1\}^{s-1}}Q_{s-1}^{\psi, \alpha_0}(\boldsymbol{y}_{s-1})\left\{\mathbb{I}(G_s=0)\frac{[d(p_s, 0,\alpha_0)-d(p_s, 0,\alpha)]^2}{d(p_s, 0,\alpha)[1-d(p_s, 0,\alpha)]}\right. \\
&~~~\left.+\mathbb{I}(G_s=1)\frac{[d(p_s, 1,\alpha_0)-d(p_s, 1,\alpha)]^2}{d(p_s, 1,\alpha)[1-d(p_s, 1,\alpha)]}\right\}\\
&\leq \frac{64}{7}\sum_{\boldsymbol{y}_{s-1}\in\{0,1\}^{s-1}}Q_{s-1}^{\psi, \alpha_0}\{\mathbb{I}(G_s=0)[d(p_s, 0,\alpha_0)-d(p_s, 0,\alpha)]^2+\mathbb{I}(G_s=1)[d(p_s, 1,\alpha_0)-d(p_s, 1,\alpha)]^2\}\\
&\leq \frac{128}{7}(\alpha_0-\alpha)^2\sum_{\boldsymbol{y}_{s-1}\in\{0,1\}^{s-1}}Q_{s-1}^{\psi, \alpha_0}(\boldsymbol{y}_{s-1})[\mathbb{I}(G_s=0)(p_0^*(\alpha_0)-p_s)^2+\mathbb{I}(G_s=1)(p_1^*(\alpha_0)-p_s)^2].
\end{align*}
The first inequality follows Lemma \ref{kl}. The last second line follows the fact that $d(p, G,\alpha)\in[1/8,3/4]$ derived from $\alpha\in[-1/2, -1/5]$ and $p\in[1/2,9/8]$. The last line dues to Property 6 in Lemma \ref{property}. Therefore, by \eqref{chain}, we have
\begin{align*}
&~~~\mathcal{K}(Q_t^{\psi, \alpha_0},Q_t^{\psi, \alpha})\\
&=\sum_{s=1}^t\mathcal{K}(Q_s^{\psi, \alpha_0},Q_s^{\psi, \alpha}|\boldsymbol{Y}_{s-1})\\
&\leq \frac{128}{7}(\alpha_0-\alpha)^2\sum_{s=1}^t\sum_{\boldsymbol{y}_{s-1}\in\{0,1\}^{s-1}}Q_{s-1}^{\psi, \alpha_0}[\mathbb{I}(G_s=0)(p_0^*(\alpha_0)-p_s)^2+\mathbb{I}(G_s=1)(p_1^*(\alpha_0)-p_s)^2]\\
&\leq \frac{128}{7}(\alpha_0-\alpha)^2\sum_{s=1}^t\mathbb{E}_{\alpha_0}[\mathbb{I}(G_s=0)(p_0^*(\alpha_0)-p_s)^2+\mathbb{I}(G_s=1)(p_1^*(\alpha_0)-p_s)^2]\\
&\leq \frac{384}{35}(\alpha_0-\alpha)^2\sum_{s=1}^t\mathbb{E}_{\alpha_0}\{\mathbb{I}(G_s=0)[R_0(p_0^*(\alpha_0),\alpha_0)-R_0(p_s,\alpha_0)+R_1(p_1^*(\alpha_0),\alpha_0)-R_1(p_{1s},\alpha_0)]\\
&~~~+\mathbb{I}(G_s=1)[R_0(p_0^*(\alpha_0),\alpha_0)-R_0(p_{0s},\alpha_0)+R_1(p_1^*(\alpha_0),\alpha_0)-R_1(p_{s},\alpha_0)]\}\\
&=\frac{768}{35}(\alpha_0-\alpha)^2Regret(\alpha_0,t,\psi).
\end{align*}
The last second line is from Property 4 in Lemma \ref{property} with $p_1^*(\alpha)=p_0^*(\alpha)-\delta, p_{1s}=p_s-\delta$ and $p_{0s}=p_s+\delta$.
\end{proof}
Now, we present a lemma to show that any pricing policy that does not reduce the uncertainty about the parameters incurs an increase in regret.
\begin{lemma}\label{cost2}
Let $\psi$ be any pricing policy satisfying the fairness constraint. For $T\geq 2, \alpha_0=-2/5$ and $\alpha_1=\alpha_0+\frac{1}{4T^{1/4}}$, we have
$$Regret(\alpha_0, T,\psi)+Regret(\alpha_1, T,\psi)\geq \frac{\sqrt{T}}{1152}e^{-\mathcal{K}(Q_T^{\psi,\alpha_0}, Q_T^{\psi,\alpha})},$$
where $\mathcal{K}(Q_0,Q_1)$ denotes the KL divergence of $Q_0$ and $Q_1$.
\end{lemma}
\begin{proof}
We define two intervals :
$$C_{\alpha_0}=\{p:|p_0^*(\alpha_0)-p|\leq \frac{1}{24T^{1/4}}\}\ \text{and}\ C_{\alpha_1}=\{p:|p_0^*(\alpha_1)-p|\leq \frac{1}{24T^{1/4}}\}. $$ Since $|p_0^*(\alpha)-p_0^*(\alpha_0)|\geq \frac{5}{6}|\alpha-\alpha_0|=\frac{5}{24T^{1/4}}$ from Property 5 in Lemma \ref{property}, $C_{\alpha_0}$ and $C_{\alpha_1}$ are disjoint. For each $\alpha\in\{\alpha_0, \alpha_1\}$, $p_0\in [1/2,9/8] \setminus C_{\alpha}$ and $(p_0,p_1)\in\mathcal{P}$, by Property 4 in Lemma \ref{property}, we obtain
$$R_0(p_0^*(\alpha),\alpha)-R_0(p_0,\alpha)+R_1(p_1^*(\alpha),\alpha)-R_1(p_1,\alpha)\geq \frac{3}{5}[p_0^*(\alpha)-p_0]^2\geq \frac{1}{960\sqrt{T}}.$$
Let $((p_{01}, p_{11}),\cdots, (p_{0T}, p_{1T}))$ be the sequence of prices generated by the pricing policy $\psi$. We define $H_t=\mathbb{I}(p_{0t}\in C_{\alpha_1})$.  We have
\begin{align*}
&~~~Regret(\alpha_0, T,\psi)+Regret(\alpha_1, T,\psi)\\
&=\frac{1}{2}\sum_{t=1}^T[R_0(p_0^*(\alpha),\alpha)-R_0(p_{0t},\alpha)+R_1(p_1^*(\alpha),\alpha)-R_1(p_{1t},\alpha)]\\
&\geq \frac{1}{1920\sqrt{T}}\sum_{t=1}^T[\mathbb{P}_{\alpha_0}(p_{0t}\notin C_{\alpha_0})+\mathbb{P}_{\alpha_1}(p_{0t}\notin C_{\alpha_1})\\
&\geq \frac{1}{1920\sqrt{T}}\sum_{t=1}^T[\mathbb{P}_{\alpha_0}(H_t=1)+\mathbb{P}_{\alpha_1}(H_t=0)]\\
&\geq \frac{1}{1920\sqrt{T}}\frac{1}{2}\sum_{t=1}^Te^{-\mathcal{K}(Q_t^{\phi,\alpha_0},Q_t^{\phi,\alpha_1})}\\
&\geq \frac{\sqrt{T}}{3840}e^{-\mathcal{K}(Q_T^{\phi,\alpha_0},Q_T^{\phi,\alpha_1})}.
\end{align*}
The last second inequality is from lemma \ref{ec13}. The last line is from the fact that $\mathcal{K}(Q_t^{\phi,\alpha_0},Q_t^{\phi,\alpha_1})$ is non-decreasing in $t$ (see proof of Lemma 3.4 in \cite{Broder2012}).
\end{proof}
We now continue with the proof of Theorem \ref{thm3}. Let $\alpha_0=-2/5$ and $\alpha_1=\alpha_0+\frac{1}{4T^{1/4}}$. We have $(\alpha_0-\alpha_1)^2=\frac{1}{16\sqrt{T}}$.
Therefore,
\begin{align*}
&~~~2[Regret(\alpha_0, T,\psi)+Regret(\alpha_1, T,\psi)]\\
&\geq Regret(\alpha_0, T,\psi)+[Regret(\alpha_0, T,\psi)+Regret(\alpha_1, T,\psi)]\\
&\geq \frac{35\sqrt{T}}{48}\mathcal{K}(Q_T^{\psi,\alpha_0}, Q_T^{\psi,\alpha})+\frac{\sqrt{T}}{3840}e^{-\mathcal{K}(Q_T^{\psi,\alpha_0}, Q_T^{\psi,\alpha}})\\
&\geq \frac{1}{3840}\sqrt{T}.
\end{align*}
The second inequality follows Lemma \ref{cost1} and Lemma \ref{cost2}. The last line dues to the fact $x+e^{-x}\geq 1$ for all $x\in\mathbb{R}$. Then, 
\begin{align*}
\mathop{\max}_{\alpha\{\alpha_0,\alpha_1\}}Regret(\alpha, T,\psi)\geq \frac{Regret(\alpha_0, T,\psi)+Regret(\alpha_1, T,\psi)}{2}\geq \frac{\sqrt{T}}{15360}.
\end{align*}
\section{Support Lemmas}\label{ssec6}
\begin{lemma} (Corollary 5.2 \citep{Tropp2012}) \label{tropp}
Consider a finite sequence $\{\mathbf{X}_k\}$ of independent, random, self-adjoint matrices with dimension $d$ that satisfy
$$\mathbf{X}_k\succeq \boldsymbol{0}\ \textit{and}\ \lambda_{max}(\mathbf{X}_k)\leq L\ \textit{almost surely.}$$
Compute the minimum eigenvalue of the sum of expectations,
$\mu_{min}:=\lambda_{min}\bigg(\sum_{k}\mathbb{E}\mathbf{X}_k\bigg)$. Then for $\zeta\in [0,1]$,
\begin{align*}
\mathbb{P}\bigg\{\lambda_{min}\bigg(\sum_{k}\mathbf{X}_k\bigg)\leq (1-\zeta)\mu_{min}\bigg\}\leq d\bigg[\frac{e^{-\zeta}}{(1-\zeta)^{1-\zeta}}\bigg]^{\mu_{min}/L}.
\end{align*}
\end{lemma}
\begin{lemma}(Proposition 2.5, \citep{wainwright_2019})\label{Hoeffding}
Suppose that the variables $X_i, i=1,\cdots,n$ are
independent, and $X_i$ has mean $\mu_i$ and sub-Gaussian parameter $\sigma_i$. Then for all $t\geq 0$, we have
$$\mathbb{P}\bigg[\sum_{i=1}^n(X_i-\mu_i)\geq t\bigg]\leq e^{-\frac{t^2}{2\sum_{i=1}^n\sigma_i^2}}.$$
\end{lemma}
\begin{lemma}(Lemma EC.1.2, \citep{Broder2012})\label{kl}
Suppose $B_1$ and $B_2$ are distributions of Bernoulli random variables with parameters $q_1$ and $q_2$, respectively, with $q_1, q_2\in (0,1)$. Then
$$\mathcal{K}(B_1;B_2)\leq \frac{(q_1-q_2)^2}{q_2(1-q_2)}.$$
\end{lemma}
\begin{lemma}(Lemma EC.1.3, \citep{Broder2012})\label{ec13}
Let $Q_0$ and $Q_1$ be two probability distributions on a finite space $\mathcal{Y}$, with $Q_0(y), Q_1(y)>0$ for all $y\in\mathcal{Y}$. Then for any function $J:\mathcal{Y}\rightarrow\{0, 1\}$,
$$Q_0\{J=1\}+Q_1\{J=1\}\geq \frac{1}{2}e^{-\mathcal{K}(Q_0,Q_1)},$$
where $\mathcal{K}(Q_0,Q_1)$ denotes the KL divergence of $Q_0$ and $Q_1$.
\end{lemma}
\end{document}